\documentclass[12pt]{article}
\usepackage{amsmath}
\usepackage{graphicx}
\usepackage{enumerate}
\usepackage{natbib}
\usepackage{url} 

\newcommand{\blind}{1}

\usepackage{epsfig,epstopdf,multirow,rotating,times,setspace}
\usepackage{graphicx}
\usepackage{amsthm,amsmath,calligra}
\usepackage{amssymb}
\usepackage{booktabs}
\usepackage[all]{xy}
\usepackage{mathtools}
\usepackage{color}
\usepackage{natbib}
\usepackage{mathrsfs}
\usepackage{apptools}
\AtAppendix{\counterwithin{lemma}{section}}

\def\argmin{\mathop{\rm argmin}}

\def\T{{ \mathrm{\scriptscriptstyle T} }}

\newtheorem{theorem}{Theorem}
\newcounter{dummy} \numberwithin{dummy}{subsection}
\newtheorem{lem}[dummy]{Lemma}
\newtheorem{lemma}{Lemma}

\newtheorem{cor}{Corollary}

\usepackage{psfrag}
\usepackage{mathrsfs}
\usepackage{accents}
\usepackage{authblk}
\usepackage{blindtext}

\usepackage{graphicx}
\usepackage{amssymb,amsmath,bbm,xcolor,mathrsfs}
\RequirePackage{float,epsfig,multirow,rotating,times}
\RequirePackage{upgreek,wrapfig}
\usepackage{mhequ}
\usepackage{natbib}
\usepackage{accents,url}
\usepackage{comment}
\usepackage{times}
\usepackage{bm}

\usepackage[plain,noend]{algorithm2e}
\usepackage{caption}
\RequirePackage{graphicx,subcaption,latexsym}
\RequirePackage{float}
\RequirePackage{multirow,rotating}
\RequirePackage{upgreek,wrapfig}
\usepackage{mathrsfs}
\usepackage{epstopdf}
\usepackage{psfrag}
\usepackage{mathrsfs}
\usepackage{blindtext}
\usepackage{epstopdf}

\def\T{{ \mathrm{\scriptscriptstyle T} }}

\newcommand{\be}{\begin{equs}}
\newcommand{\ee}{\end{equs}}
\newcommand{\bpm}{\begin{pmatrix}}
\newcommand{\epm}{\end{pmatrix}}
\DeclareMathOperator{\E}{\mathbf E}

\newcommand{\ind}{\mathbbm 1}

\newcommand{\xj}{x_j}

\newcommand{\klt}{\underaccent{\tilde}{\kappa}}
\newcommand{\klb}{\underaccent{\bar}{\kappa}}

\def\T{{ \mathrm{\scriptscriptstyle T} }}

\def\ind{\mathbbm{1}}
\def \bX{\mathbf{X}}
\def \bx{\mathbf{x}}
\def \bz{\mathbf{z}}
\def \mz{\mathrm{z}}

\usepackage{accents}

\usepackage{multirow}
\usepackage{array}
\usepackage{hyperref}
\usepackage{xr}

\begin{document}
\def\spacingset#1{\renewcommand{\baselinestretch}%
{#1}\small\normalsize} \spacingset{1}

\newtheorem{assumption}{Assumption}


\if1\blind
{
  \title{\bf An Approximate Bayesian Approach to Covariate-dependent Graphical Modeling}
  
  \author{Sutanoy Dasgupta, Peng Zhao, Jacob Helwig, Prasenjit Ghosh,
  	Debdeep Pati and Bani K. Mallick \\
  	Department of Statistics, Texas A\&M University}
  
  \maketitle
} \fi
\if0\blind
{  \author{}
  \bigskip
  \bigskip
  \bigskip
   \title{\bf  An Approximate Bayesian Approach to Covariate-dependent Graphical Modeling}
   
} \fi
 
\maketitle

	\bigskip
	\begin{abstract}
 Gaussian graphical models typically assume a homogeneous structure across all subjects, which is often restrictive in applications. In this article, we propose a weighted pseudo-likelihood approach for graphical modeling which allows different subjects to have different graphical structures depending on extraneous covariates. The pseudo-likelihood approach replaces the joint distribution by a product of the conditional distributions of each variable.  We cast the conditional distribution as a heteroscedastic regression problem, with covariate-dependent variance terms, to enable information borrowing directly from the data instead of a hierarchical framework.  This allows independent graphical modeling for each subject, while retaining the benefits of a hierarchical Bayes model and being computationally tractable. An efficient embarrassingly parallel variational algorithm is developed to approximate the posterior and obtain estimates of the graphs. Using a fractional variational framework, we derive asymptotic risk bounds for the estimate in terms of a novel variant of the $\alpha$-R\'{e}nyi divergence. We theoretically demonstrate the advantages of information borrowing across covariates over independent modeling. We show the practical advantages of the approach through simulation studies and illustrate the dependence structure in protein expression levels on breast cancer patients using CNV information as covariates.
	\end{abstract}
	\noindent%
	\vfill

	\newpage
\spacingset{1.75}

\section{Introduction}
Undirected graphical models provide a widely used framework for modeling multivariate distributions, with applications ranging across diverse disciplines such as statistical physics, bioinformatics, computational biology and sociology. Here, one exploits the structure in the distribution in the form of assumptions of conditional independence among the involved variables. Suppose we observe a $p$-dimensional sample $\bx=(x_{1},x_{2}, \ldots, x_{p})$ from a multivariate Gaussian distribution with a non-singular covariance matrix. Then the conditional independence structure of the distribution can be represented with a graph $\mathrm{G}$. The graph $\mathrm{G}=(\mathrm{V},\mathrm{E})$ is characterized by a node set $\mathrm{V}=(1,2, \ldots, p)$ corresponding to the $p$ variables, and an edge set $\mathrm{E}$ such that $(i,j) \in \mathrm{E}$ if, and only if, $x_i$ and $x_j$ are conditionally dependent given all other variables. 

Several methods have been developed with the goal of estimating this underlying graph $\mathrm{G}$ given $n$ independent and identically distributed observations $\bx_1, \bx_2, \ldots, \bx_n$, such as \cite{friedman2008sparse,yuan2007model,giudici1999decomposable}. However, in practice, the $n$ observations might not be identically distributed, that is, there is no {\it homogeneous} underlying graph describing the conditional dependence structure among the variables for all the observations.  It is imperative, therefore, to develop efficient graph modeling schemes that can take into account the variability in the graph structure across observations depending on additional covariate information. 

\subsection{Current Literature}
Perhaps surprisingly, the literature on handling this {\it heterogeneity} in the underlying graph structure is relatively sparse. 
Some approaches attempt to model heterogeneous graphs without using covariate information, as in \cite{guo2011joint,danaher2014joint,peterson2015bayesian,ha2015dingo,ren_gaussian_2022}. These methods depend on the criteria of first splitting the data into homogeneous groups and then sharing information within and across groups as appropriate. However, a clear criterion for the choice of homogeneous groups is difficult to obtain without extraneous data, and the performance can suffer when the identified groups have small samples. A second approach focuses on adding the covariates into the mean structure of Gaussian graphical models as multiple linear regressions such that the mean is a continuous function of the covariates. \cite{bhadra2013joint} proposed a Bayesian joint model for estimating the mean structure and the graph together. \cite{yin2011sparse, cai2013covariate, lee2012simultaneous}  studied similar models from a frequentist
 perspective. Such an approach estimates the graph after eliminating the effects of the
 covariates from its mean structure. However, the graph structure is still assumed to be homogeneous for all observations. Another approach for estimating heterogeneous graphs is to model the underlying covariance matrix as a function of the covariates, as considered in \cite{hoff2012covariance,fox2015bayesian,pourahmadi1999joint,pourahmadi2000maximum,pourahmadi2013high,zhang2012moving}. The main challenge here is to enforce sparsity in the precision matrix while being positive definite, as the sparsity in the covariance matrix does not normally carry to the precision matrix through matrix inversion. Recently, \cite{liu2010graph} developed a graph-valued regression model that partitions the covariate space into different groups by classification and regression trees (CART). This method assumes that there exists a true partition of the covariate space such that the graph structure is homogeneous inside each of the partitions. Second, tree structures may not be flexible enough to capture the true partition, even if such a partition exists. \cite{ni2019bayesian} proposed a graphical regression method that estimates covariate-dependent continuously varying directed acyclic
graphs (DAGs). But, the conditional dependence structure cannot be extended to undirected graphs. In related literature, \cite{kolar2010sparse,wang2014inference} developed a penalized kernel smoothing method for conditional precision matrices under an additional simplifying assumption that the precision matrix is a function of a low-dimensional
index variable. \cite{kolar2010estimating,zhou2010time,qiu2016joint} proposed methods for
inferring time-varying graphs, which are however, difficult to extend to non-time indexed covariates. 


\subsection{Proposed formulation}
In what follows, $\bX \in \mathbb{R}^{n \times p}$ refers to the data matrix corresponding to $n$ individuals on $p$ variables, with the rows $\bX_i \in \mathbb{R}^p$ corresponding to the observation on individual $i$. The columns  $x_j \in \mathbb{R}^n, j=\{1, \ldots, p\}$ correspond to the $p$ variables.
The main goal of this paper is to learn the graph structure $\mathrm{G}$ from a collection of $p$-variate independent samples $\bX_i$, as a function of some extraneous covariates $\bz_i$ corresponding to the samples. The only assumption on the dependence structure is that the graph parameters vary smoothly with respect to the covariates, that is, if $\bz_i$ and $\bz_j$ are similar, then the graph structure corresponding to $\bX_i$ and $\bX_j$ will be similar. To the best of our knowledge, there is no method available in the literature that can model the graph itself as a continuous function of covariates without putting additional restrictive or simplifying assumptions on the dependence structure of the graphs on the covariates. A natural way to achieve the sharing of information through covariates is to consider a hierarchical model in a Bayesian setting. Embedding a complex graphical modeling framework in a hierarchical setting involves manifold challenges. In the following, we develop a novel weighted pseudo-likelihood based approach that obviates these challenges, which is computationally efficient and yet retains all the benefits of a hierarchical model.  

Our modeling scheme can be organized into two main steps. First, we use a novel weighted pseudo-likelihood (W-PL) function (described in Section \ref{sec:methods}) to obtain a posterior distribution for the graph structure for a fixed individual, with the weights defined as a function of the covariates. The idea of a pseudo-likelihood approach is to tackle each of the variables $x_j, j=\{1,2,\ldots,p\}$ separately instead of trying to jointly model them. See, for instance, \cite{meinshausen2006high} and \cite{atchade2019quasi} for a more detailed discussion in this context. It is important to note that the pseudo-likelihood model is not a valid probability model, as the conditional distributions for a multivariate Gaussian distribution do not determine the joint distribution. However, consistency and other benefits of the pseudo-likelihood models have been extensively explored in \cite{besag1975statistical}, demonstrating the efficacy of such an approach. The standard pseudo-likelihood approach replaces the original joint likelihood function by the product of the conditional likelihoods of the random variables $x_j$s. Thus, this approach casts the conditional distribution of each of the variables $x_j$ given the remaining variables as a standard homoscedastic regression problem. Instead, we cast the conditional distribution as a {\it weighted} regression problem, by introducing covariate-dependent weights in the error variance term, which leads to a weighted pseudo-likelihood function. 

 Second, we use a variational algorithm to efficiently approximate the posterior distribution and obtain an estimate of the graph for a fixed individual. Repeating this process for every individual, we obtain an empirical distribution of the graph structure over the support of the covariates associated with the individuals. 
The advantages of this two-step approach are manifold.

\noindent {\bf Embarrassingly parallel:} 
The approach allows independent  estimation of the graph structure parameters for different individuals, by sharing information across individuals directly from the data rather than through the parameters themselves. 

\noindent {\bf Borrowing of information:}
Observe that the standard approach to sharing information across the parameters would be to consider a full-blown hierarchical Bayesian model. To illustrate this idea, consider the following simple setup as an example.
Assume we have $k$ groups $\{y_{ij}, i=1, \ldots, n_j, j=1, \ldots, k\}$ and $\bar{y}_j$ denotes the mean of the $j$-th group. 
Instead of considering a hierarchical model to estimate the true group mean $\theta_j$, consider for every fixed $j \in \{1, \ldots, k\},$
\begin{eqnarray*}
\bar{y}_l \mid \theta_j \sim \mathcal{N}(\theta_j, (\sigma^2/n_l)w(\bar{y}_j, \bar{y}_l)^{-1}), l=1, \ldots, k,
\end{eqnarray*}
for some similarity function $w(x,y)$ which 
takes higher values for $x \approx y$, and lower  values as $x$ moves further away from $y$.
Then, assuming a balanced sample, the weighted MLE of $\theta_j$ is simply 
\begin{eqnarray}\label{eq:wmle}
\hat{\theta}_j = \frac{\sum_{l=1}^k \bar{y}_l w(\bar{y}_j, \bar{y}_l)}{\sum_{l=1}^k w(\bar{y}_j, \bar{y}_l)}.
\end{eqnarray}
In \eqref{eq:wmle}, the information is still shared among the different groups, but the sharing of information comes directly through the data instead of a common prior. This is because of the nature of the weighted MLE which borrows more information from subjects with similar group means. As a result, the weighted likelihood of the $j$-th and $k$-th true group means $\theta_j$ and $\theta_k$ would be similar if $\bar{y}_j$ and $\bar{y}_k$ are similar. This approach avoids the computational overhead of a hierarchical Bayesian model and forms the basis of our covariate-dependent graphical model where we share information across the model parameters directly through the data via a weighted (pseudo)-likelihood function, rather than the standard hierarchical modeling framework.

Finally, we derive risk bounds for the variational estimator by casting it in a fractional variational framework adopting \cite{yang2020alpha}, using a novel variant of the $\alpha$-R\'{e}nyi divergence. In particular, assuming a careful interplay between the sparsity and smoothness in the conditional regression coefficients, we showed that the W-PL framework achieves optimal variational risk bounds irrespective of whether the underlying distribution is homogeneous or heterogeneous across different covariates. Our theory also shows how W-PL improves on the independent modeling framework in the case of imbalanced samples corresponding to different covariate levels, leveraging its ability to borrow information from the entire data while estimating the graph for a specific covariate level.   In the fractional variational framework, the term $\alpha$ controls the relative trade-off between the model fit and the prior regularization term.  In the current study, the variational estimator corresponds to $\alpha=1$. However, for technical simplicity, we restrict the theoretical analysis to estimators with $\alpha<1$. The results for $\alpha=1$ can be derived under stronger assumptions on prior tails, as discussed in \cite{yang2020alpha}. However, that extension does not alter the main message of the theoretical results, and has been omitted.


The rest of the paper is organized as follows.
Section \ref{sec:methods} describes the proposed weighted pseudo-likelihood approach. Section \ref{sec:var} describes the variational algorithm used to approximate the posterior distribution obtained in Section \ref{sec:methods}. Section \ref{sec:vartheory} provides variational risk bounds for the parameter estimates for both discrete and continuous covariates and demonstrates the advantage of our W-PL model over a standard approach that assumes the observations for each covariate level to be independent.  A thorough simulation study is conducted in Section  \ref{sec:sims}. Finally,  Section \ref{sec:real} illustrates the performance of the approach to estimate the dependence structure in protein expression levels in cancer patients using copy number variation values as covariates.

\section{A weighted pseudo-likelihood (W-PL) approach}\label{sec:methods}
Likelihood based approaches provide a sound basis for comparing the plausibility of different graphs given the observations. Unfortunately, likelihood based approaches to modeling graph structures are intractable in general for non-chordal graphs because of an intractable normalizing term. This has sparked a lot of interest in tractable learning of non-chordal graphs in high dimensions. A pseudo-likelihood approach discussed in \cite{besag1975statistical,besag1977efficiency} became a convenient alternative to likelihood approaches for modeling the underlying graph dependence structure.  Over the last few years, the pseudo-likelihood approach has been used widely for learning Markov random fields and neighborhood detection in Markov fields, as in \cite{ji1996consistent,csiszar2006consistent} and others.
\cite{heckerman1995learning,freno2009scalable} discussed a pseudo-likelihood based model class to learn the dependence structure in Bayesian networks. Recently, \cite{pensar2017marginal} used the idea of marginal pseudo-likelihood and proved the consistency of the pseudo-marginal likelihood estimator in learning the dependence structure (neighborhood detection) of the Markov network. The set of {\it neighbors} of a variable $x_j$ are variables such that given these neighbor variables, the conditional distribution of $x_j$ is independent of all other variables. \cite{besag1975statistical} argued the consistency of the maximum pseudo-likelihood estimator as the dimension of the random variable $\bX$ increases, under the assumption that the number of neighbors is small and finitely bounded. \cite{besag1977efficiency} further studies the efficiency of the pseudo-likelihood estimators under Gaussian schemes. 

The pseudo-likelihood approach can be described as follows: Suppose there are $n$ individuals, indexed $i=1,2, \ldots, n$, in a study. Let the $i$-th observation in the dataset $\bX$ be denoted as $\bX_i=(x_{i,1},x_{i,2},\ldots, x_{i,p})$, which corresponds to the $i$-th individual. Let $x_{i,-j} \in \mathbb{R}^{p-1}$ denote the vector of the $i$-th observation including all variables except $x_{i,j}$. 
This approach tries to model the conditional distribution of each of the $x_j$'s given all other variables, denoted by $\bX_{-j} \in \mathbb{R}^{n \times (p-1)}$. Let the $p-1$ dimensional vector $\beta_{j}$ indicate the regression effect of $\bX_{-j}$ on $x_{j}$. The assumption commonly used here is that the conditional distribution of a variable $x_j$ depends on only a few of the remaining variables referred to as the neighbors of $x_j$, and can be completely specified in terms of a regression function comprising of the neighbors as predictors. Here, we further assume a Gaussian likelihood. Then, the conditional likelihood of $x_j$, denoted by $\mathrm{L}(j)$, can be written as
 \begin{equation}
 \mathrm{L}(j)=p(x_{j}\mid\bX_{-j}, \beta_j) \propto \prod_{i=1}^n \exp{\left\{-({x_{i,j} - x_{i, -j}^{\T}\beta_{j})}^2/2\sigma^2 \right\}}
 \label{eq1}
\end{equation}
with a possibly sparse coefficient vector $\beta_{j}$.  Consequently, the pseudo-likelihood $\mathrm{L}(\mathrm{G})$ for a fixed graph $\mathrm{G}$ can be written as
 \begin{equation}
     \mathrm{L}(\mathrm{G})=\prod_{j=1}^p \mathrm{L}(j) = \prod_{j=1}^p p(x_{j}\mid \bX_{-j}, \beta_{j}).
     \label{eq2}
\end{equation}
If the true data generating distribution $p(\bX)$ is a zero-mean multivariate Gaussian with a precision matrix ${\Omega}^*$, i.e., 
\begin{equation}\label{homogen}
   \bX \sim  \mathcal{N}(0,\Omega^{*-1}),
\end{equation}
then it is well-known that the conditional distribution of $x_{j}\mid\bX_{-j}$ is given by \eqref{eq1} with $\beta_{jk} = - \Omega^*_{jk}/\Omega^*_{jj}$. However,  \eqref{eq2} is the product of conditional distributions and is not a valid probability density in general. However, it serves as an effective computational tool for estimating the precision matrix.

In this paper, we propose a novel adaptation of this approach, and define a weighted version of this conditional likelihood for each individual in the study. We assume that the underlying graph structure is a function of extraneous covariates $\mathrm z$. That is, given a covariate $\bz=(\bz_1,...,\bz_n)'$, our population model assumes the true data generating distribution $p(\bx_i)$ is a zero-mean multivariate Gaussian with a precision matrix ${\Omega}^*(\bz_i)$ for $i=1,...,n$. 
\begin{equation}\label{eq:population}
    \bx_i \mid \bz_i \sim \mathcal{N}(0,{\Omega}^{*-1}(\bz_i)), \quad i=1,...,n, \quad \bX = [\bx_1,...,\bx_n]'.
\end{equation}
Thus, we allow the coefficient vector $\beta_j$'s to be different for different individuals, depending on the extraneous covariates.  We use the notation $\beta_j^l \in \mathbb{R}^{p-1}$ to denote the coefficient vector corresponding to the regression of the variable $x_j$ on the remaining variables for individual $l$. More generally, we use the notation $\beta_j(\mathrm{z})$ to denote the coefficient vector given an arbitrary covariate $\mathrm{z}$. Let $\bz_i$ denote the covariate vector associated with the $i$-th individual in the study, and define $\bz=(\bz_1, \bz_2,\ldots, \bz_n)$. Next, relative to the covariate $\mz$, we assign weights $\mathrm{w}(\mz,\bz_i) = \phi_\tau(\|\mz - \bz_l\|)$ to every individual in the study, where $\phi_\tau$ is the Gaussian density with mean $0$ and variance $\tau^2$. When $\mz=\bz_l$ corresponds to the $l$-th individual in the study, we use the notation $\mathrm{w}_{l}(\bz_i)=\mathrm{w}(\bz_l, \bz_i)$ to denote the weight associated with the $i$-th individual in the study. Next, we provide a brief overview of our approach, and then move on to more details about the steps involved.
\begin{enumerate}
    \item For a fixed individual $l$, we attach weights $\mathrm{w}_{l}(\bz_i), i \in \{1,2,\ldots, n\}$, to every individual in the study relative to the $l$-th individual, and perform a weighted regression of the $j$-th variable on the remaining variables.
    \item The likelihood function of regression parameters $\beta_j^l \in \mathbb{R}^{(p-1)}$ corresponding to the $j$-th variable for the $l$-th individual is given by  $\prod_{i=1}^n \exp{\left\{-{({x_{i,j} - x_{i, -j}^{\T}\beta_{j}^l)}^2\mathrm{w}_{l}(\bz_i)}/{2\sigma^2 }\right\}}$. We place a suitable spike-and-slab prior on the parameters to enforce sparsity on the conditional dependency structure of $x_j$ given $x_{-j}$. 
    
    \item The collection of vectors $\beta_j^l, j \in \{1,2, \ldots, p\}$ together form the parameter of interest $\mathrm{B}^l$ for the $l$-th individual in the study, which is used to learn the undirected conditional dependency structure among $x_1, \ldots, x_p$, for the $l$-th individual. 
    
    \item Rather than a fully Bayesian approach, a variational approximation is made to the posterior to obtain estimates of the coefficients.  
\end{enumerate}
   Note that the likelihood of the parameter $\beta_j$, associated with individual $l$, is independent of the regression parameters associated with other individuals in the study, allowing parallel estimation of the parameters associated with the $n$ different individuals in the study. However, information is being borrowed from every individual for each regression parameter estimation through the associated weights which attach more importance to individuals with covariates similar to the individual $l$ and lower weights to individuals with different covariate values. 

In what follows, we describe the steps involved in the pseudo-likelihood approach for graph estimation given an arbitrary covariate $\mz$ in more detail. First, we introduce some more notations. Throughout, $i,l,l_1$ and $l_2$ have been used as indices for the $n$ individuals in the study (or their corresponding observations), while $j$ and $k$ have been used to index the $p$ variables. 
We propose the following conditional working model
\begin{eqnarray}
x_{i, j} \mid x_{i, -j}, \bz, \mz \sim  \mathcal{N}(x_{i, -j}^{\T}\beta_{j}(\mz), \sigma^2/\mathrm{w}(\mz, \bz_i)), \quad i=1, \ldots, n. 
\label{eqmodel1}
\end{eqnarray}
This results in a weighted version of the conditional likelihood form shown in \eqref{eq1}. Let $x_j$ denote the $n$ observations on the $j$-th variable. Then, the weighted conditional likelihood of $x_j$ given a covariate value $\mz$, denoted by $\mathrm{L}_j^w(\mz)$, is given by
\begin{equation*}\label{condwt0}
 \mathrm{L}_j^w(\mz)=p^w(x_{j}\mid \bX_{-j}, \beta_j(\mz), \bz, \mz) ~~\propto~~ \prod_{i=1}^n \exp{\left\{-\frac{({x_{i,j} - x_{i, -j}^{\T}\beta_{j}(\mz))}^2\mathrm{w}(\mz, \bz_i)}{2\sigma^2 }\right\}}.
\end{equation*} 
Here, the superscript $w$ is used to indicate that the conditional distribution $p$ involves a weighted likelihood function.
Thus the weighted pseudo-likelihood for the graph $\mathrm{G}(\mz)$ corresponding to a covariate value $\mz$, denoted by $\mathrm{L}^w(\mathrm{G}(\mz))$ can be written as
\begin{equation*}
     \mathrm{L}^w(\mathrm{G}(\mz))=\prod_{j=1}^p \mathrm{L}_j^w(\mz) = \prod_{j=1}^p p^w(\xj \mid \bX_{-j}, \beta_{j}(\mz),\bz, \mz).
     \label{eq5a}
\end{equation*}

Next, suppose that we are interested in the underlying graph structure of an individual $l$ in the study, that is, $\mz=\bz_l$. Then, the weighted conditional likelihood of $x_j$ for the $l$-th individual in the study, denoted by $\mathrm{L}_l^w(j)$, is given by
\begin{equation*}
 \mathrm{L}_l^w(j)=p^w(x_{j}\mid \bX_{-j}, \beta_j^l, \bz) ~~\propto~~ \prod_{i=1}^n \exp{\left\{-\frac{({x_{i,j} - x_{i, -j}^{\T}\beta_{j}^l)}^2\mathrm{w}_{l}(\bz_i)}{2\sigma^2 }\right\}}.
 \label{eq4}
\end{equation*}
 Also, the weighted pseudo-likelihood for the graph $\mathrm{G}^l$ corresponding to the $l$-th individual, denoted by $\mathrm{L}^w(\mathrm{G}^l)$ can be written as
\begin{equation}
     \mathrm{L}^w(\mathrm{G}^l)=\prod_{j=1}^p \mathrm{L}_l^w(j) = \prod_{j=1}^p p^w(\xj \mid \bX_{-j}, \beta_{j}^l,\bz).
     \label{eq5}
\end{equation}

For the graph $\mathrm{G}^l$, \eqref{eqmodel1} can be expressed in a structural equation form, similar to the form discussed in \cite{han2016estimation,pearl2000models}. Given the covariate matrix $\bz$, let us define a $p\times p$ coefficient matrix $\mbox{A}^l$ with $\mbox{A}_{jj}^l=0, j \in \{1, \ldots,p\}$, and $\mbox{A}_{jk}^l=\beta_{j,k}^l, j,k \in \{1,2,, \ldots, p\}$. For $i\in \{1, \ldots, n\}$, let $\epsilon_i^l$ be $(p-1)$-variate independent random variables such that given $\bz$, $\epsilon_i^l \sim \mathcal{N}(\mu_{\epsilon^l}={\bf 0},\Sigma_i^l)$, where 
$\Sigma_i^l=\sigma^2\mathrm{w}_{l}(\bz_i)^{-1}\mathbb{I}_p$. Now, define the $n\times n$ diagonal covariance matrix $\Sigma_0^l=\mathrm{Diag}( \sigma^2 w_{l}(\bz_1)^{-1},\sigma^2 w_{l}(\bz_2)^{-1}, \ldots,\sigma^2 w_{l}(\bz_n)^{-1})$, and $\epsilon^l={(\epsilon_1^l, \epsilon_2^l, \ldots, \epsilon_n^l)}$. Then, it can be shown that given $\bz$,  $\epsilon^l\sim \mathcal{MN}({\bf 0}, \mathbb{I}, \Sigma_0^l)$, where $\mathcal{MN}(\cdot, \cdot, \cdot)$ is the matrix normal distribution with suitably chosen parameters. Thus, conditioned on the covariate $\bz$, we can express $\bX^{\mathrm{T}}$ as $\bX^{\mathrm{T}}=\mbox{A}^l\bX^{\mathrm{T}} + \epsilon^l$. Then, for $\Xi^l={(\mathbb{I}-\mbox{A}^l)}^{-1}$, we have $\bX^{\mathrm{T}}=\Xi^l\epsilon^l$, and given $\bz \textrm{ and } \Xi^l$, 
$$\bX^{\mathrm{T}} \sim \mathcal{MN}(0, \Xi^l{\Xi^l}^{\mathrm{T}},  \Sigma_0^l).$$
Thus, estimating the graph $\mathrm{G}^l$ is equivalent to estimating the coefficient matrix $\mbox{A}^l$, with $\Sigma_0^l$ being the nuisance parameter, and then setting $\mathrm{G}^l_{ij}=\mathbb{I}{\{\mbox{A}^l_{ij} \neq 0\}}$, similar to the literature on directed acyclic graphs. This graph estimate $\hat{\mathrm{G}}^l$ given the covariates $\bz$ is not a proper undirected graph, and one needs to perform appropriate post-processing to obtain a proper undirected graph.

In this paper, we compare the covariance matrix $\Sigma_0^l=\mathrm{Diag}( \sigma^2w_{l}(\bz_1)^{-1},\sigma^2 w_{l}(\bz_2)^{-1}, \ldots,\sigma^2 w_{l}(\bz_n)^{-1})$ of $\epsilon^l$ in the proposed setup with the covariate-independent setup in traditional DAG literature, where $\Sigma_0 \equiv \sigma^2 \mathbb{I}_n$, for all individuals. The borrowing of information in standard DAG literature is uniform across all observations resulting in a common graph structure for all individuals. However in the proposed setup, the amount of information borrowed from the $i$-th observation is covariate-dependent, via the associated weights $\mathrm{w}_{l}(\bz_i)$, resulting in different graph estimates for different individuals. Recall that $\mathrm{w}_{l}(\bz_i) = \phi_\tau(\|\bz_i - \bz_l\|)$, where $\phi_\tau$ is the Gaussian density with mean $0$ and variance $\tau^2$. So, the amount of borrowing is controlled by the bandwidth parameter $\tau$. As $\tau \rightarrow \infty$, the weights become equal for all the observations $i \in \{1, \ldots, n\}$, and thus the amount of information borrowed becomes uniform across observations. In this situation, conditioned on $\bz$ and $\Xi^l$, \ $\bX^{\mathrm{T}} 
\xrightarrow[]{D} \mathcal{MN}(0,  \Xi^l {\Xi^l}^{T}, \Sigma_0)$ for $l=\{1,2, \ldots, n\}$. As a result, the graph estimates $\mathrm{G}^l$ would be the same for all individuals in the study, as is assumed in the classic DAG literature. On the other hand, when $\tau \rightarrow 0$, we have  $\mathrm{w}_{l}(\bz_i) \rightarrow c_0\mathbb{I}{\{i=l\}}$ (for some constant $c_0$). In this scenario, ${(\Sigma_0^l)}^{-1}$ reduces to $\sigma^{-2}\mathrm{Diag}(\mathbb{I}\{1=l\},\mathbb{I}\{2=l\}, \ldots, \mathbb{I}\{n=l\})$. When the covariates have a discrete distribution, this results in a separate estimation algorithm where one estimates the underlying graphs corresponding to the different covariate levels separately with no information shared across different covariate levels. For practical experiments, the covariates vary across different observations, and the choice of the bandwidth parameter $\tau$ used for defining the weights becomes important for efficient borrowing of information. Ideally, we want to obtain a bandwidth estimate $\hat{\tau}$ such that $\hat{\tau}$ is larger for individuals when there are relatively few remaining individuals with similar covariates (sparse region in the support of the covariates). Conversely, we want $\hat{\tau}$ to be smaller for individuals when there are several other individuals with similar covariates. Towards this end, we follow the estimate discussed in \cite{dasgupta2020two,abramson1982bandwidth,van2003adaptive} and others. Specifically, if $k(\bz)$ is a kernel density estimate of the covariate $\bz$ with bandwidth parameter $h$, then $\hat{\tau}(\bz_0)=h/\sqrt{k(\bz_0)}$ is the adaptive bandwidth parameter at covariate value $\bz_0$. When the covariates are multi-dimensional, the bandwidth parameter $h$ can be replaced by the harmonic mean of the bandwidths in each direction. If the dimension of the associated covariates is low, one can efficiently estimate the probability density function of the covariates and obtain a variable kernel bandwidth. 
In our simulation study, we follow this approach to select bandwidth hyperparameters. To select the prior residual variance $\sigma^2$, the prior slab variance $\sigma^2_\theta$ and the prior inclusion probability $\pi$, we use a hybrid of grid search and model averaging, which is discussed in greater detail in Supplement \ref{ssec:algo}.    


 Note that, for $l_1 \neq l_2$, our (weighted) likelihood in \eqref{eq5} for the parameter $\mathrm{B}^{l_1}$ is different from the (weighted) likelihood of the parameter $\mathrm{B}^{l_2}$ because of the associated weights, that is, we do not have a {\em single, coherent} probability model consisting of the parameters $\mathrm{B}^l, l \in \{1, \ldots, n\}$. This step is crucial in ensuring that different individuals have potentially different underlying graph structures, depending on the covariate values. Indeed, if we have $\bz_{l_1}=\bz_{l_2}$, then $\mathrm{B}^{l_1}$ and $\mathrm{B}^{l_2}$ do have the same probability model.  

Using a standard hierarchical model approach, we would put a common prior structure on the parameters $\mathrm{B}^l$s to facilitate the borrowing of information across different individuals. However, in our proposed approach, we allow the sharing of information to come directly from the observations, effectively borrowing information while simultaneously allowing independent estimation of each of the parameters $\mathrm{B}^l$. Thus, this approach allows us to define an empirical distribution on the parameters $\mathrm{B}^l$, and hence the underlying graph $\mathrm{G}^l$, given the covariates $\bz$.

Next, we specify the prior distribution for the coefficient parameters corresponding to the regression problem introduced in \eqref{eqmodel1}. Fix an observation $l \in \{1, \ldots, n\}$, and a variable $j \in \{1, \ldots, p\}.$ Note that, a significantly non-zero regression coefficient corresponds to an edge in the underlying graph structure. With this goal in mind, we use a spike-and-slab prior on the parameter $\beta_{j}^l.$ That is, for $k \in \{1, \ldots, p\}$, $\beta_{j,k}^l$ is assumed to come from a zero-mean Gaussian density with a variance component $\sigma^2 \sigma_{\beta}^2$ (``slab" density) with a probability $\pi$, and equals zero (``spike" density) with probability $1-\pi$. Let us define $\gamma_{j,k}^l=\mathbb{I}\{\beta_{j,k}^l \neq 0\}$ which can be treated as Bernoulli random variables with a common probability of success $\pi$. Define the row vector $\gamma_{j}^l= (\gamma_{j,1}^l, \gamma_{j,2}^l, \ldots, \gamma_{j,p}^l)$, and $\Gamma^l=\{\gamma_{j}^l,j=1,2, \ldots, p\}$. Then, we use the following prior distribution for $(\beta_{j}^l,\gamma_{j}^l)$ given by
\[
p_0\left(\beta_{j}^l,\gamma_{j}^l\right)=\prod_{k=1, k \neq j}^{p} \delta_{\{0\}}{\left(\beta_{j,k}^l\right)}^{1-\gamma_{j,k}^l}{\mathcal{N}\left(\beta_{j,k}^l; 0, \sigma^2\sigma_{\beta}^2\right)}^{\gamma_{j,k}^l}\prod_{k=1, k \neq j}^{p} \pi^{\gamma_{j,k}^l}{\left(1-\pi\right)}^{(1-\gamma_{j,k}^l)}.
\]
Using the data model as described in \eqref{eqmodel1} for an individual $l$, we obtain the following posterior distribution for $(\beta_{j}^l,\gamma_{j}^l)$ as
\[
p(\beta_{j}^l,\gamma_{j}^l \mid \bX) \propto \exp\Bigg\{-\frac{1}{2\sigma^2}{\sum_{i=1}^n {\bigg(x_{ij} -\sum_{\substack{k\neq j,k=1}}^p x_{ik}\beta_{j,k}^l\bigg)}^2\mathrm{w}_{l}(\bz_i)}\Bigg\}p_0(\beta_{j}^l,\gamma_{j}^l).
\]

Note that the posterior distributions of $\mathrm{B}^{l_1}$ and $\mathrm{B}^{l_2}$ are independent, but are similar depending on the covariates through the associated weights. This notion allows independent and fast estimation of the parameters $\mathrm{B}^l, l=1,2 ,\ldots, n$, while ensuring that subjects with similar covariates have similar coefficient parameter estimates. This also allows us to effectively gauge the variability of the graph structure across subjects. However, since the posterior distribution does not have a closed-form solution, we would require MCMC samples in order to obtain posterior samples for the parameters. This can be very time consuming when $p$ is large, especially since we have to essentially compute $np$ such distributions. In the following, we develop an efficient parallelized mean-field variational inference to approximate the posterior distribution.

\section{An efficient parallelized block mean-field variational inference}\label{sec:var}
Variational Bayes approximations are deterministic approaches where instead of finding the posterior probability distributions, we aim to find an approximation of them by first introducing a class of approximating distributions and then finding the distribution that best approximates the posterior obtained through some optimizing criterion over the aforesaid class. See, for instance, \cite{jordan1999introduction, wainwright2008graphical, ormerod2010explaining,Blei_2017}, only to name a few. In this section, we adopt the block-mean-field approach proposed by \cite{carbonetto2012scalable} in the context of Bayesian variable selection with spike-and-slab priors in high dimensional regression problems.


Suppose we have a parameter of interest $\xi$ with an intractable posterior distribution $p(\xi)$, an observed data vector $y$, and the variational tractable family of densities $q(\xi)$. Let $\mathrm{D_{KL}}(\cdot \| \cdot)$ denote the Kullback-Leibler divergence between two density functions. Then the ``best approximating density" over a tractable family of densities $\Gamma$ is a density $q^*(\xi)$  such that
\[
q^*(\xi)=\underset{q \in \Gamma}{\argmin} \, \mathrm{D_{KL}}(q \, \| \, p(\xi\mid y)).
\]
Since $\mathrm{D_{KL}}(q\|p(\cdot|y))\geq 0$, we have $\log p(y) \geq$ ELBO, where ELBO$=\int q(\xi) \log \left\{ {p(y,\xi)}/{q(\xi)} \right\} d\xi$ is the evidence-lower bound. 

In the current study, the parameter of interest is $\xi=(\beta_{j}^l,\gamma_{j}^l)$. Here, we adopt the block mean-field approach for the variational approximation considered in \cite{carbonetto2012scalable} given by 
\[
q(\beta_{j}^l,\gamma_{j}^l;\phi_j^l)=\prod_{k=1, k \neq j}^{p}q_k(\beta_{j,k}^l,\gamma_{j,k}^l;\phi_{j,k}^l).
\]
In the above expression, $\phi_j^l$'s are free parameters corresponding to the $l$-th individual, and the separate factors $q_k$ have the following form:
\begin{equation*}
q_k(\beta_{j}^l,\gamma_{j}^l;\phi_j^l)=
      \mathcal{N} \left(\beta_{j,k}^l; \mu_{j,k}^l, {(s_{j,k}^l)}^2\right)^{\gamma_{j,k}^l} \delta_0\left(\beta_{j,k}^l\right)^{1-\gamma_{j,k}^l}\left(\alpha_{j,k}^l\right)^{\gamma_{j,k}^l} \left(1-\alpha_{j,k}^l\right)^{1-\gamma_{j,k}^l}
\end{equation*}
where $\phi_{j,k}^l=(\alpha_{j,k}^l,\mu_{j,k}^l,{(s_{j,k}^l)}^2), k=\{1, \ldots, p, k \neq j\}$ are the free parameters, and $\delta_0$ is the  ``spike" density degenerate at zero. Thus, the individual factors $q_k$ are independent spike-and-slab densities. The parameter $\beta_{j,k}^l$ comes from a Gaussian density with mean $\mu_{j,k}^l$ and standard deviation $s_{j,k}^l$ (the ``slab" part) with probability $\alpha_{j,k}^l$, and is zero (the ``spike" part) with probability $1-\alpha_{j,k}^l$. Therefore, we can write the variational density family as 
\[
q(\beta_{j}^l,\gamma_{j}^l)=\prod_{k=1}^{p-1} 
\mathcal{N} \left(\beta_{j,k}^l; \mu_{j,k}^l, {(s_{j,k}^l)}^2\right)^{\gamma_{j,k}^l} \delta_0\left(\beta_{j,k}^l\right)^{1-\gamma_{j,k}^l}\left(\alpha_{j,k}^l\right)^{\gamma_{j,k}^l} \left(1-\alpha_{j,k}^l\right)^{1-\gamma_{j,k}^l},
\]
where the ``best approximating density" $q^*(\beta_{j}^l,\gamma_{j}^l)$ can be obtained by optimizing over the free parameters $\phi_{j,k}^l=(\alpha_{j,k}^l,\mu_{j,k}^l,{(s_{j,k}^l)}^2)$ in order to maximize the ELBO over the variational family.

The coordinate descent updates for the variational parameters can be obtained by taking partial derivatives of the ELBO,  setting them to zero, and solving for $\alpha_{j,k}^l,\mu_{j,k}^l$ and ${(s_{j,k}^l)}^2$. \cite{carbonetto2012scalable} proposed a component-wise algorithm where one iterates between updating $\mu_{j,k}^l$ and $\alpha_{j,k}^l$ for a fixed $j$, and then updating $j \in \{1,2,\ldots, p\}.$ \cite{huang2016variational} proposes a batch-wise updating scheme where one iterates between updating ${(s_{j,k}^l)}^2, j=1,2,\ldots,p$, in a batch, followed by updating $\mu_{j,k}^l,j=1,2,\ldots,p$, in a batch, and subsequently by updating $\alpha_{j,k}^l, j=1,2,\ldots, p$, in a batch. \cite{huang2016variational} argued that the component-wise update scheme in high-dimensional settings might lead to noise accumulation, and can cause the variational estimates to move away from the true parameter value. \cite{huang2016variational} also asserted that the batch-wise updating algorithm achieves frequentist as well as Bayesian consistency even when the dimension $p$ diverges to infinity at an exponential rate as the sample size grows to infinity.
We closely follow the batch-wise updating algorithm discussed in \cite{huang2016variational}, which leads to the following variational parameter updates.  
\begin{eqnarray*}
{(s_{j,k}^l)}^2 &=&\frac{\sigma^2}{(1/\sigma_{\beta}^2 +\sum_{i=1}^n x_{ik}^2\mathrm{w}_{l}(\bz_i))}; \quad \mbox{logit}(\alpha_{j,k}^l)= \mbox{logit}(\pi) + \frac{(\mu_{j,k}^l)^2}{{2s_{j,k}^l}^2} +\log \frac{s_{j,k}^l}{\sigma\sigma_{\beta}};\\
{\mu_{j,k}^l}&=&\frac{{(s_{j,k}^l)}^2}{\sigma^2} \sum_{i=1}^n \Bigg\{ \mathrm{w}_{l}(\bz_i)x_{ik}\bigg(x_{ij} - \sum_{m \neq j,k} x_{im}\mu_{j,m}^l\alpha_{j,m}^l\bigg)\Bigg\}.
\end{eqnarray*}

Note that, the two-step approach proposed in this paper might not result in a proper undirected graph as the posterior inclusion probability estimates $\hat{\alpha}_{j,k}^l$ and $\hat{\alpha}_{k,j}^l$ might not be the same. Hence, we perform post-processing steps in order to obtain a bonafide undirected graph estimate in practice. For our purposes, we set $\tilde{\alpha}_{j,k}^l=\tilde{\alpha}_{k,j}^l=(\hat{\alpha}_{j,k}^l+\hat{\alpha}_{k,j}^l)/2$ in order to symmetrize the inclusion probabilities and obtain a proper undirected graph estimate.

\section{Variational risk bounds}\label{sec:vartheory}
In this section, we derive risk bounds for our variational estimates and demonstrate the efficiency of the weighted pseudo-likelihood model over a model that treats the covariate levels independently. To that end, we start by defining a few notations. First, it is well-known that the conditional distribution of all variables $x_j, j\in \{1, \ldots, p\}$,  given the remaining variables $x_{-j}$, is
\begin{equation}
{p}(x_{i,j} \mid \bX_{i,-j}, \bz_i) \sim \mathcal{N}\left(- \sum_{k \neq j} \frac{{\Omega}_{kj}^*(\bz_i)}{{\Omega}_{jj}^*(\bz_i)}x_{i,k}, \frac{1}{{\Omega}_{jj}^*(\bz_i)}\right), i \in \{1, \ldots, n\}.
\end{equation}
 Let $\Gamma^*(\mathrm{z})$ represent the $p \times (p-1)$ matrix of parameters controlling the latent indicator variables corresponding to the sparsity structure of ${\Omega}^*(\mathrm{z})$. $\gamma^*_j(\mathrm{z})$, the $j$-th row of $\Gamma^*(\mathrm{z})$, is the indicator variable corresponding to the $j$-th variable.  Let $\mathrm{B}^*(\mathrm{z})$ denote the true coefficient parameter matrix given covariate $\mathrm{z}$, with $j$-th row $\beta_{j,k}^*(\mathrm{z})= {{\Omega}_{kj}^*(\mathrm{z})}/{{\Omega}_{jj}^*(\mathrm{z})}, k\in \{ 1, \ldots, p\}.$ Here $\gamma_j^*(\mathrm{z})$ is the indicator variable associated with the truth $\beta_j^*(\mathrm{z})$. 
 For simplicity of presentation, we assume that the true variance parameter $\sigma_*^2$ is correctly specified in the model. 
  Let $\Theta^l(\bz)=(\mathrm{B}^l(\bz), \Gamma^l(\bz))$, where $\mathrm{B}^l(\bz)$ represents the $p \times (p-1)$ coefficient matrix with $\beta^l_j(\bz)$ as the $j$-th row corresponding to the $j$-th variable as a function of covariates $\bz$.  Denote by $\theta^l_j(\bz)=(\beta^l_j(\bz),\gamma^l_j(\bz))$ the parameter associated with the $j$-th variable given the covariates $\bz$. 
 Let $p_{\theta^l_j}=p_{\gamma^l_j}p_{\beta^l_j \mid \gamma^l_j}$ denote the spike-and-slab prior distribution of $\theta^l_j$ used in the analysis, $\Gamma$ denote the variational family of distributions $q(\theta^l_j)$ for the parameter $\theta^l_j$ and $\Lambda$ denote the parameter space for the coefficient parameter $\mathrm{B}$, where $\Lambda^l_j$ denotes the parameter space for $\beta^l_j$ for an individual $l \in \{1,...,n\}$. 
 Then, $\theta_j^l(\bz)$ denotes the parameter associated with the conditional distribution of the $j$-th variable given the other variables  for the $l$-th individual. For the estimation of the parameters associated with the $l$-th individual, we assign weights $\mathrm{W}_l=\mbox{Diag}\{w_{l}(\bz_1), w_{l}(\bz_2), \ldots, w_{l}(\bz_n)\}$ to the likelihood contribution of the $n$ individuals under study, depending on the covariates associated with the individuals. In what follows, the results are derived for a fixed individual, and is valid for each individual graph parameter $\Theta^l(\bz)$ in the study.
 


We derive risk bounds for the weighted pseudo-likelihood model by casting it as a misspecified model, as in \cite{kleijn2006misspecification}. Following \eqref{condwt0}, define the misspecified conditional distribution function $p^w(x_j \mid \bX_{-j}, \theta^l_j(\bz),\bz)$ as
\begin{eqnarray}
p^w(x_j \mid \theta^l_j(\bz), \bX_{-j}, \bz) = \left(\prod_{i=1}^n \frac{\sqrt{\mathrm{w}_{l}(\bz_i)}}{\sqrt{2\pi}\sigma_*}\right)\exp \Bigg \{ -\frac{{(\xj - \bX_{-j}\beta^l_j(\bz))}^{\mathrm{T}}\mathrm{W}_l(\xj - \bX_{-j}\beta^l_j(\bz))}{2\sigma_*^2}\Bigg \}.
\label{condwt}
\end{eqnarray}

Let $p(x_j \mid \bX_{-j}, \theta_j^*(\bz_l))$ denote the true well-specified conditional distribution with respect to the true parameter $\theta_j^*(\bz_l)$.
Misspecified models as treated in \cite{kleijn2006misspecification} gives rise to Kullback-Leibler balls centered around $\tilde{\theta}^l_j(\bz)$ where
\begin{eqnarray}
\tilde{\theta}^l_j(\bz)=\underset{\theta_j(\bz) \in \mathbb{R}^{(p-1)}\times {\{0,1\}}^{(p-1)}}{\argmin}{\int p(x_j \mid \bX_{-j}, \theta_j^*(\bz_l)) \log \frac{p(\xj \mid \bX_{-j}, \theta_j^*(\bz_l))}{p^w(\xj \mid \theta^l_j(\bz), \bX_{-j},\bz)} d x_j}.
\label{modelproj}
\end{eqnarray}
For $\tilde{\theta}^l_j(\bz)$ as defined in \eqref{modelproj} and the true parameter value $\theta_j^*(\bz_l)$, and $\alpha \in (0,1)$, we measure closeness between $\tilde{\theta}^l_j(\bz)$ and a value $\theta_j^l(\bz)$ using the divergence
\begin{eqnarray*}
D_{\theta_j^*(\bz_l),\alpha}(\theta_j^l(\bz), \tilde{\theta}^l_j(\bz) \mid \bX_{-j},\bz)=\frac{1}{\alpha -1} \log \int \bigg\{\frac{p^w(\xj \mid \bX_{-j}, \tilde{\theta}^l_j(\bz),\bz)}{p^w(\xj \mid \bX_{-j}, {\theta}^l_j(\bz),\bz)}\bigg\}^\alpha p(x_j \mid \bX_{-j}, \theta_j^*(\bz_l))d x_j.
\end{eqnarray*}
Here, $D_{\theta_j^*(\bz_l),\alpha}(\theta_j^l(\bz), \tilde{\theta}^l_j(\bz) \mid \bX_{-j},\bz)$ is the $\alpha$-R\'{e}nyi divergence measure between the Kullback-Leibler minimizer $\tilde{\theta}^l_j(\bz)$ and a candidate parameter value $\theta_j^l(\bz)$ with respect to the true underlying distribution, conditioned on $\bX_{-j}$ and covariates $\bz$.  
 We define a weighted version of our divergence measure between $\tilde{\Theta}^l(\bz)$ and $\Theta^l(\bz)$ as
\begin{eqnarray}
\label{newdivwt}
d_{\Theta^*(\bz), \alpha}(\Theta^l(\bz), \tilde{\Theta}^l(\bz)) &=& \max_{j} d_{\theta_j^*(\bz_l),\alpha} (\theta^l_j(\bz), \tilde{\theta}^l_j(\bz)), \\
d_{\theta_j^*(\bz_l),\alpha} (\theta^l_j(\bz), \tilde{\theta}^l_j(\bz)) &=& -\frac{1}{1-\alpha}\log \E_{-j} \exp \left \{ -(1- \alpha)D_{\theta_j^*(\bz_l),\alpha}(\theta_j^l(\bz), \tilde{\theta}^l_j(\bz)| \bX_{-j},\bz)\right \}.
\end{eqnarray}
Since 
$$
0 < \exp \left \{ -(1- \alpha)D_{\theta_j^*(\bz_l),\alpha}(\theta_j^l(\bz), \tilde{\theta}^l_j(\bz) \mid \bX_{-j},\bz)\right \} \leq 1,
$$ with the right equality holds if and only if $\theta_j^l(\bz)=\tilde{\theta}^l_j(\bz)$, $d_{\theta_j^*(\bz_l),\alpha}(\theta_j^l(\bz), \tilde{\theta}^l_j(\bz))$ is a valid divergence measure. Let $\mathrm{D_{KL}}(q_{\theta} \, \| \, p_{\theta})$ denote the Kullback-Leibler divergence. Let $\hat{q}_{\theta_j}(\theta_j^l(\bz))$ be the $\alpha$-variational estimate associated with the fractional posterior distribution of $\theta_j^l(\bz)$ as in \cite{yang2020alpha}. 
\begin{equation}
    \hat{q}_{\theta_j}(\theta_j^l(\bz))=\underset{q_{\beta_j,\gamma_j}=\prod_{k=1}^p q_{\beta_{jk},\gamma_{jk}}}{\argmin} \Bigg \{-\int_{\Lambda_j} \sum_{\delta \in {\{0,1\}}^{p-1}}  R d\theta_j^l(\bz) + \alpha^{-1} \mathrm{D_{KL}}({q}_{\theta_j^l(\bz)} \| p_{\theta_j})\Bigg \},
    \label{varapprox}
\end{equation}
where $R=\log \frac{p^w(\xj \mid \beta^l_j(\bz),\gamma^l_j(\bz),\bX_{-j},\bz)}{p^w(\xj \mid \beta_j^*(\bz_l),\gamma_j^*(\bz_l),\bX_{-j},\bz)}{q}_{\theta_j}(\theta^l_j(\bz))$. Let ${\| \mbox{A}\|}_2=\sqrt{\tau_{\max}(\mbox{A}^{\mathrm{T}}\mbox{A})}$ denote the operator norm of a matrix $\mbox{A}$, where $\tau_{\max}(\mbox{A})$ is the maximum eigenvalue of $\mbox{A}$. Let $\|a\|_\infty=\max_i a_i, \|a\|_2=\sqrt{a^\mathrm{T}a}, \|a\|_0=\sum_i \ind\{{a_i\ne 0\}}$ be the corresponding $\ell_\infty,\ell_2$ and $\ell_0$ norms of a vector $a$. For any $\epsilon \in (0,1)$, let 
\begin{eqnarray}
\underaccent{\bar}{r}(n,\epsilon)=\frac{\alpha \epsilon^2}{(1-\alpha)}    + \frac{s^* \log p}{n(1-\alpha)}.
\label{riskeps}
\end{eqnarray}
Although the methodology in \S \ref{sec:var} corresponds to $\alpha =1$, we present the risk bounds for $\alpha \in (0, 1)$ as it greatly simplifies the technicalities without deviating from the main idea.  
We now adapt \cite{yang2020alpha} to develop risk bounds for this variational estimate in three situations. First, we examine the case in which the true covariates are continuously distributed, and then we discuss the case in which they are discrete.  Finally, we consider the situation where the underlying true distribution is independent of the available covariates.  We list our assumptions of theoretical analysis in  Supplement~\ref{ssec:conditions}. All proofs are deferred to  Supplement \ref{ssec:proofs} with auxiliary results in  Supplement \ref{ssec:aux}.

     \subsection{Continuous covariate-dependent model} In this subsection, we consider the case where the covariates $\mathrm{z}$ are drawn from a density which is absolutely continuous with respect to the Lebesgue measure. For simplicity, we consider $\mathrm{z}$ as a scalar. Define $\beta_j^*(\mathrm{z})$ as the coefficient corresponding to the $j$-th variable as a function of the covariate value $\mathrm{z}$. Define $\dot{\beta}_{j}^{*}(\mathrm{z})$ and $\ddot{\beta}_{j}^{*}(\mathrm{z})$ as the first and second order point-wise derivatives respectively of  $\beta_{j}^{*}(\mathrm{z})$ as a function of $\mathrm{z}$. Similarly, define $\Sigma(\mathrm{z})$ as the covariance matrix as a function of $\mathrm{z}$, with $\dot{\Sigma}(\mathrm{z})$, $\ddot{\Sigma}(\mathrm{z})$ as the point-wise first and second order derivatives with respect to $\mathrm{z}$.

     We consider the predefined misspecified weighted pseudo-likelihood with the parameter space: $\|\beta_j^{l}(\bz)\|_0 \leq C_0 s_j^*$, $l=1,...,n,j=1,...,p$ for constant $C_0 \geq 1$.
         Then  for a subject $l$, the KL divergence between the truth and weighted pseudo-likelihood is
         \begin{equation}\label{eq:KL}
         \int p(\bX \mid  \Theta^{*}(\bz_1),...,\Theta^{*}(\bz_n))  \log  \frac{p(\bX \mid  \Theta^{*}(\bz_1),...,\Theta^{*}(\bz_n)  ) }{\prod_{j=1}^p  p^{w_l}(x_{j} \mid \bX_{-j}, \Theta_{-j}^l(\bz) )}d\bX, 
         \end{equation} 
      where the true model $ p(\bX \mid  \Theta^{*}(\bz_1),...,\Theta^{*}(\bz_n))$ is induced from equation~\eqref{eq:population} and $\Theta_{-j}^l(\bz) $ is the parameter of interest  given the constraint that $\|\beta_j^{l}(\bz)\|_0 \leq C_0 s_j^*$ with $C_0 \geq 1$.  We choose   $w_{l}(\bz_k) = {c_l}  K  \left((\bz_k-\bz_l)/\tau \right) /{\tau}$, where $c_l$ is a subject-specific constant.	Then the following lemma characterizes the property of the KL minimizer.

          \begin{lemma}\label{lem2}
          	Under  {\bf Assumptions K, T1- T4} in Supplement~\ref{ssec:conditions}, let $\tilde \beta^l_j(\bz)$ be the minimizer of the KL divergence~\eqref{eq:KL} under the constraint $\|\beta^l_j(\bz)\|_0 \leq C_0 s_j^*$ for constant $C_0\geq 1$. If $\tau \rightarrow 0$ and $n\tau \rightarrow \infty$, then we have for $j=1,..,p$ and $l=1,...,n$
          	\begin{equation*}
         \E_{\bz}\| \tilde \beta_j^{l}(\bz) - \beta^{*}_{j}(\bz_l) \|^2 = O \left(s_j^*\tau^4+\frac{s_j^*}{n\tau}\right).
          	\end{equation*}
          \end{lemma}    
                Therefore, by balancing the bias $O(\tau^4)$ and variance $O(1/(n\tau))$, one can achieve an MSE of order $n^{-4/5}s_j^*$ corresponding to the optimal bandwidth $\tau=n^{-1/5}$. It is important to note that the optimal bandwidth has the same optimal rate as kernel density estimation so that we can borrow the tuning strategy from these problems in practice. There could be multiple KL minimizers due to the non-convexity of the $\ell_0$ constraint, however, all the KL minimizers have the same convergence behavior depicted in Lemma~\ref{lem2}. The following theorem characterizes the convergence rate of the obtained estimator towards the KL minimizers.

  \begin{theorem}\label{cor3}
 	Under {\bf Assumptions T1-T5} and {\bf P, K} in Supplement~\ref{ssec:conditions},  let $\hat{q}^l(\theta^l_j(\bz))$ be the variational estimate of the $j$-th graph coefficients for subject $l$. Suppose that $\tau =n^{-1/5}$ and $\|\tilde \beta_j^{l}(\bz)-\beta_j^{*}(\bz_l)\| \leq c_1 n^{-4/5}s_j^*$ for $\tilde \beta_j^{l}(\bz)$ specified in Lemma~\ref{lem2} for $l=1,...,n$ and $j=1,...,p$.  Then with probability at least $1-c_2n\exp(-c_3 n)-c_4 n/p^{c_0}-e^{- (c_5s^*-1) \log(np)}$, we have
 	\begin{equation*}
 \max_{l=1,...,n}	\max_{j=1,...,p}\int \frac{1}{n}d_{\alpha}(\theta_j^l(\bz), \tilde{\theta}^l_{j}(\bz))\hat{q}_{\theta_j^l(\bz)}(\theta_j^l(\bz)) d\theta_j^l(\bz) \leq C\frac{1+\alpha}{1-\alpha} \left( \frac{s^* \log(np)}{n} + \frac{s^*}{n^{\frac{3}{5}}} \right) , 
 	\end{equation*}
 	for positive constants $c_0,c_1,c_2,c_3,c_4,c_5,C>0$.
 \end{theorem}
While the first term in the upper bound of the risk is due to the model selection associated with estimating a sparse precision matrix and cannot be improved, observe that the second term is primarily due to the misspecification error under the true data generating distribution \eqref{eq:population}. Note that this term ($s^* n^{-3/5}$) is in contrast with the convergence rate in Lemma~\ref{lem2}, which is upper bounded by $s^*_j n^{-4/5}$. The extra factor of $n^{1/5}$ is due to the discrepancy in the Euclidean norm and the size of the misspecified Kullback-Leibler neighborhood, which is upper bounded by $\|\mathrm W^{1/2}_l\bX_{-j}(\tilde \beta_j^{l}(\bz)-\beta^{*}_{j}(\bz_l))\|^2_2$. Since, $\|\mathrm W_l\|_2 \leq c \max_x K(x)/\tau \leq c/\tau \leq c n^{1/5}$ and the rate of convergence of  $d_\alpha(\theta_j^l(\bz), \tilde{\theta}^l_{j}(\bz)) $ is associated with the size of the misspecified KL ball, the second term is slowed down by a factor of $n^{1/5}$. We conjecture that this cannot be improved unless one considers the discrete covariate setting as we discuss below.

According to the Theorem \ref{cor3}, graph estimation can be carried out consistently for each observation under proper smoothness assumptions. Note that due to the continuous covariate structure, one cannot perform a valid separate estimation across different covariate values. We also compare our Theorem~\ref{cor3} with Theorem 3.4 of \cite{qiu2016joint} who  considered joint estimation of multiple graphical models. Our result is  non-asymptotic and considers joint risk across all subjects $l=1,...,n$, while \cite{qiu2016joint} obtained risk bound for estimating a single graph.

\subsection{Discrete covariate-dependent underlying graph}

Next, we consider the scenario where there are $K$ covariate levels $\bz_1, \bz_2, \ldots, \bz_K$ corresponding to $K$ different populations with parameters $\Theta^*(\bz_l)$ corresponding to the $l$-th covariate level. Let $n_l, \ l \in \{1, \ldots, K\}$, be the number of sample observations corresponding to the $l$-th group, with $n=\sum_{l=1}^{K} n_l$ being the total number of observations.  Let $\bX_{\bz_l}$ denote the $n_l \times p$ data matrix corresponding to the covariate $\bz_l$.
We assume that the true data generating distribution $p(\bX_{i,\bz_l})$ for the rows of $\bX_{\bz_l}$  is a zero-mean multivariate Gaussian with a sparse precision matrix ${\Omega}^*(\bz_l), l=1,\ldots, K$. In this case, the borrowing of information in the misspecified weighted pseudo-likelihood approach causes the Kullback-Leibler minimizer $\tilde{\theta}^{l}_{j}(\bz)$ as defined in \eqref{modelproj} to be different from $\theta_j^{*}(\bz_l)$, where $\theta_j^{*}(\bz_l)=(\beta_j^{*}(\bz_l),\delta_{j}^*(\bz_l))$ is the $j$-th row of $\Theta^*(\bz_l)$.  


\begin{lemma}\label{lem:2r}
If {\bf Assumption W} in Supplement~\ref{ssec:conditions} is satisfied, let $\tilde \beta_j^{l}(\bz)$ be the minimizer of the KL divergence~\eqref{modelproj} under the constrained $\|\beta_j^{l}(\bz)\|_0 \leq C_0 s_j^*$ for constant $C_0\geq 1$. Then we have for $j=1,...,p$ and $l=1,...,n$:
\begin{equation}
 \| \tilde \beta_j^{l}(\bz)- \beta_{j}^{*}(\bz_l)\|^2_2 \leq c  \exp{ (- 2c_{l}^2/\tau^2+ 2\log(n/n_l-1))} s_j^*,
\end{equation}
for some constant $c>0$.
\end{lemma}

Note that when $\tau<\min_l c_l/\sqrt{\log n} $, the convergence rate of the KL minimizer towards the truth is faster than $s_j^*/n$.  {\bf Assumption W} implies that the weighted estimator will finally converge to the separate estimation, so that our method is no worse than separate estimation asymptotically.

 Assume that $s^* \log (np) /n \rightarrow 0$. Then  we have the following theorem where the risk bound is derived for an individual with covariate value $\bz_l$.
\begin{theorem}
For any $\zeta \in (0,1)$ and $\epsilon \in (0,1)$ and $\underaccent{\bar}{r}(n,\epsilon)$ as in \eqref{riskeps}, if {\bf Assumptions T, W, P} in Supplement~\ref{ssec:conditions} are satisfied,  we have  with probability at least $1 - \zeta - c_2p^{-c_1} - p\exp\{-a_2n_l\}$,
\begin{equation*}
  \int \frac{1}{n}d_{\Theta^{*}(\bz_l), \alpha}(\Theta^l(\bz), \tilde{\Theta}^l(\bz))\hat{q}_{\Theta^l(\bz)}(\Theta^l(\bz)) d\Theta^l(\bz) \leq \underaccent{\bar}{r}(n,\epsilon) + \frac{s^*}{n(1-\alpha)} \log \left(\frac{s^*\sqrt{n}}{\epsilon\sqrt{n_l}}\right)  + \frac{\log(p/\zeta)}{n(1-\alpha)},
\end{equation*}
for some positive constants $a_1, a_2, c_1, c_2$ and $D$.
\label{cor2}
\end{theorem} 
The optimal error rate is obtained by balancing $\underaccent{\bar}{r}(n,\epsilon)$ and $\log(1/\epsilon)$. The proof of the above theorem is similar to that of Theorem~\ref{thm2} and Corollary~\ref{indep} below, and is therefore omitted.


\subsection{Covariate-independent underlying graph}
Here, we consider the situation where the underlying structure is independent of the covariate levels. In this case, we assume that the underlying graph structure is  homogeneous, as described in \eqref{homogen}. Thus there is a common true graph parameter $\Theta^*=(\mathrm{B}^*, \Gamma^*)$ associated with every individual in the study. Then we have the following lemma.

\begin{lemma}
For the weighted likelihood model, when the true conditional distribution is covariate-independent, we have $\tilde{\beta}^l_j(\bz)=\beta_j^*$ for $l=1,...,K$.
\label{lemwtequiv}
\end{lemma}

\begin{theorem}
Under Assumptions {\bf P, W, A} in Supplement~\ref{ssec:conditions} and data generating process~\eqref{homogen}, for any $\zeta \in (0,1)$ and $\epsilon \in (0,1)$ and $\underaccent{\bar}{r}(n,\epsilon)$ as in \eqref{riskeps}, for any $l=1,...,K$, we have  with probability at least $1 - \zeta - c_2 p^{-c_1} - p\exp(-a_2n_l)$, for some positive constants $a_1, a_2,c_1,c_2$ and $D$, 
\begin{equation*}\label{thm2}
  \int \frac{1}{n}d_{\Theta^*, \alpha}(\Theta^l(\bz), \Theta^*)\hat{q}_{\Theta^l(\bz)}(\Theta^l(\bz)) d\Theta^l(\bz) \leq \underaccent{\bar}{r}(n,\epsilon)+ \frac{s^*}{n(1-\alpha)} \log \left(\frac{s^*\sqrt{n}}{\epsilon\sqrt{n_l}}\right) + \frac{\log (p/\zeta)}{n(1-\alpha)}.
\end{equation*}
\end{theorem}
The strength of the proposed approach is that the risk bound for the variational estimate is sharper for every individual in the current study as compared to independent modeling of the two groups separately. If one were to perform independent modeling of a {\it homogeneous} graphical structure for the two covariate levels separately. The risk bound of the parameters of an individual belonging to a group with size $n_l$ would be as below:
\begin{cor}
Under Assumptions {\bf P, A} in Supplement~\ref{ssec:conditions} and data generating process~\eqref{homogen}, for any $\zeta \in (0,1)$ and $\epsilon \in (0,1)$ and $\underaccent{\bar}{r}(n_1,\epsilon)$ as in \eqref{riskeps}, we have  with probability at least $1 - \zeta - c_2 p^{-c_1} - p\exp\{-a_2n_l\}$,
\begin{align*}
 \int \frac{1}{n_l}d_{ \alpha}(\Theta, \Theta^*)\hat{q}_{\Theta}(\Theta) d\Theta \leq \underaccent{\bar}{r}(n_l,\epsilon)+ \frac{s^*}{n_l(1-\alpha)} \log \left(\frac{s^*}{\epsilon}\right) + \frac{1}{n_l(1-\alpha)}\log \left(\frac{p}{\zeta}\right),
\end{align*}
for some positive constants $a_1, a_2,c_1,c_2$ and $D$.
\label{indep}
\end{cor}
If one of the groups has a sample size $n_l=\mathcal{O}(n^h)$ with $h < 1$, then the risk bounds for that group will increase compared to the proposed model because the information is borrowed from every subject in the study for the proposed approach.

\section{Simulation Study}\label{sec:sims}

We begin our simulation study with a setting defined by a unidimensional covariate before considering a multidimensional covariate. In both cases, the covariate $\bz$ is randomly drawn from a uniform distribution. To generate the data for each of the settings, we first define the precision matrix $\Omega_i$ for the $i$-th individual as a function of the covariate $\bz_i$. We then generate the observation $\bX_i$ for the $i$-th individual according to the population model~\eqref{eq:population} from a mean-zero $(p+1)$-dimensional normal distribution with precision matrix $\Omega_i$. Finally, we apply W-PL to estimate the graphs $\hat{\mathrm{G}}^i$ describing the sparsity structure of $\Omega_i$. In addition to varying the dimensionality of the covariate, we also perform experiments in each setting with varying data dimensionality, examining performance for $p\in\{10, 30, 50\}$. 

We perform 50 trials for each experiment, repeating the data generation process in each trial. In each trial, we first select an individual-specific bandwidth hyperparameter $\tau$ for W-PL using a two-step kernel density estimation technique. We then average over a grid of $\pi$ candidates using the exponentiated ELBO as our unnormalized model averaging weights, and conduct a two-dimensional grid search to select $\sigma^2$ and $\sigma^2_\beta$ for each $\pi$ candidate, using the ELBO as our grid search objective function. More details on this hyperparameter specification scheme are included in Supplement \ref{ssec:algo}. These hyperparameters are used in our final variational estimate to the posterior inclusion probabilities $\alpha_{j,k}$, which we symmetrize as $\tilde\alpha_{j,k} = (\alpha_{j,k} + \alpha_{k,j})/2$. Finally, we threshold the symmetrized probabilities at $0.5$ as $\hat{\mathrm{G}}^i_{j,k}=\ind\{\tilde{\alpha}_{j,k} >0.5\}$ to construct the final graph estimates $\hat{\mathrm{G}}^i$. To evaluate the performance of W-PL, we compute the sensitivity and specificity of these estimates compared to  to the ground-truth precision structure $\mathrm{G}^*$, where these metrics are defined as: 

\[
\text{sensitivity}=\frac{\#\{(j,k):(\mathrm{G}_{jk}^*=1) \cap (\hat{\mathrm{G}}_{jk}^i=1)\}}{\#\{(j,k):(\mathrm{G}_{jk}^*=1)\}}, \quad
\text{specificity}=\frac{\#\{(j,k):(\mathrm{G}_{jk}^*=0) \cap (\hat{\mathrm{G}}_{jk}^i=0)\}}{\#\{(j,k):(\mathrm{G}_{jk}^*=0)\}}.
\]

We consider two competitors for W-PL in these experiments. The first is a time-varying graphical model from \cite{haslbeck_mgm_2020} that uses kernel smoothing and elastic net regularization (mgm). The second is also a time-varying graphical model from \cite{yang_estimating_2020} that uses a local group LASSO penalty (loggle). In both cases, we select hyperparameters using cross-validation.   

We consider experiments where the covariate is discrete in Supplement \ref{app:disc} and where the data distribution departs from Gaussian in Supplement \ref{ssec:dGauss}, as well as a comparison to the method of \cite{qiu2016joint} in Supplement \ref{app:Qiu_comp}.  

\subsection{Unidimensional Covariate}\label{1cont}

We first consider a unidimensional covariate $\bz_i\in [-3,3]$ and define the $j,k$ entry of the ground-truth precision matrices as $\Omega^i_{j,k} = 2$ if $j = k$,  $\Omega^i_{j,k} = 1$ if $(j,k)\in\{(2,3), (3,2)$, $\Omega^i_{j,k} = \ind\{\bz_i < 1\}\cdot \min(1, \frac12 - \frac12\bz_i)$ if $(j,k)\in\{(1, 2), (2, 1)\}$ and  $\Omega^i_{j,k} = \ind\{\bz_i > -1\}\cdot \min(1, \frac12 + \frac12\bz_i)$ if $(j,k)\in\{(1, 3), (3, 1)$.
The ground truth precision structures are given in Supplement \ref{uni_gr}. To generate the covariate, we sample from uniform distributions on $[-3, -1],[-1,1],$ and $[1,3]$ 50 times each. Thus, in this experiment, $n=150$. 

We present the results for these experiments in Table \ref{tab:cont_cov_dep}. At each of the considered dimensionalities, W-PL outperforms loggle and mgm in terms of sensitivity. mgm consistently offers the lowest false positive rate, although W-PL remains competitive in this metric. Further, although the sensitivity differential between W-PL and loggle is roughly constant across the different dimensionalities, as $p$ get larger, the performance of mgm relative to W-PL substantially decreases. 

\begin{table}[h]
    \centering
\begin{tabular}[t]{llll}
\toprule
$p$ & Method & Sensitivity$(\uparrow)$ & Specificity$(\uparrow)$\\
\midrule
 & W-PL & $\mathbf{0.8382} (0.0743)$ & $0.9951 (0.0057)$\\

 & loggle & $0.7802 (0.0707)$ & $0.9926 (0.0071)$\\

\multirow{-3}{*}{\raggedright\arraybackslash 10} & mgm & $0.7057 (0.0953)$ & $\mathbf{0.9991} (0.0018)$\\
\cmidrule{1-4}
 & W-PL & $\mathbf{0.7758} (0.1068)$ & $0.9977 (0.0014)$\\

 & loggle & $0.7211 (0.0935)$ & $0.9981 (0.0009)$\\

\multirow{-3}{*}{\raggedright\arraybackslash 30} & mgm & $0.5894 (0.1112)$ & $\mathbf{0.9999} (0.0002)$\\
\cmidrule{1-4}
 & W-PL & $\mathbf{0.7387} (0.0907)$ & $0.9984 (0.0009)$\\

 & loggle & $0.6982 (0.0895)$ & $0.9984 (0.0005)$\\

\multirow{-3}{*}{\raggedright\arraybackslash 50} & mgm & $0.5149 (0.0761)$ & $\mathbf{1.0000} (0.0000)$\\
\bottomrule
\end{tabular}
    \caption{\it Results for 1-dimensional continuous covariate-dependent setting}
    \label{tab:cont_cov_dep}
\end{table}

In order to gauge the practical performance of the proposed method, we look at the estimated inclusion probability, specifically $\alpha_{12}$ for the edge between $x_1$ and $x_2$, and $\alpha_{13}$ for the edge between $x_1$ and $x_3$. To gauge the variability in the estimates, we study not only the mean posterior inclusion probability across the trials, but also the $5$-th and $95$-th quantiles. Figure \ref{fig5cs} illustrates the true precision value between the edges and the corresponding mean inclusion probability. This figure shows that the presence (or absence) of an edge between pairs of variables is almost always correctly recovered for the first and third clusters. The behavior of the inclusion probability completely mimics the behavior of the true precision value across individuals, and the variability is naturally most apparent in the middle cluster where the precision matrix varies with the covariate. 
Note that the dependence structure for variables where the corresponding entry in the precision matrix does not change across subjects is correctly recovered for all subjects across all trials.

\begin{figure}
\begin{center}
\begin{tabular}{cc}
\includegraphics[width=0.3\linewidth]{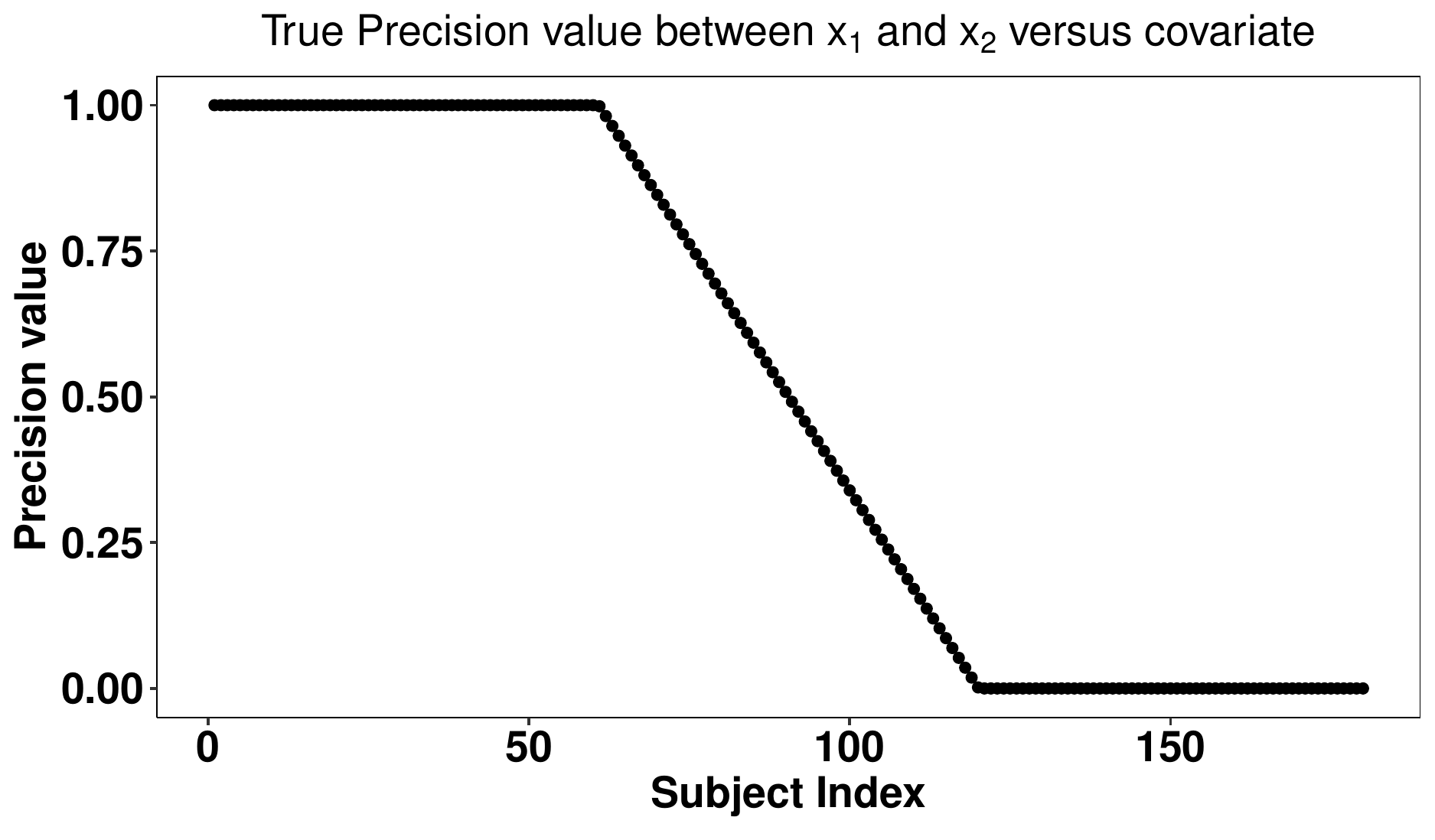} &
\includegraphics[width=0.3\linewidth]{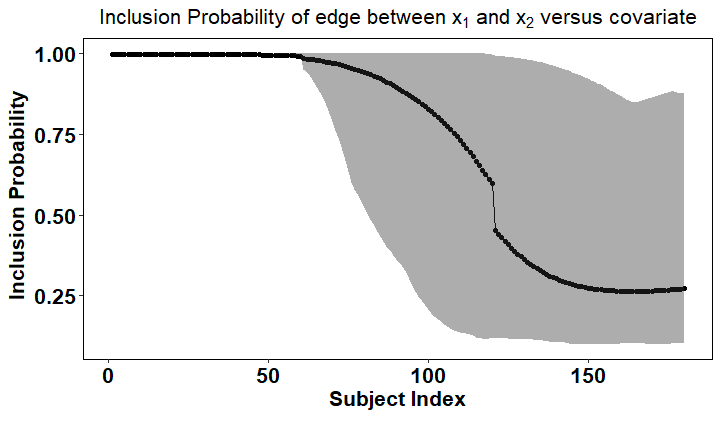}\\
\includegraphics[width=0.3\linewidth]{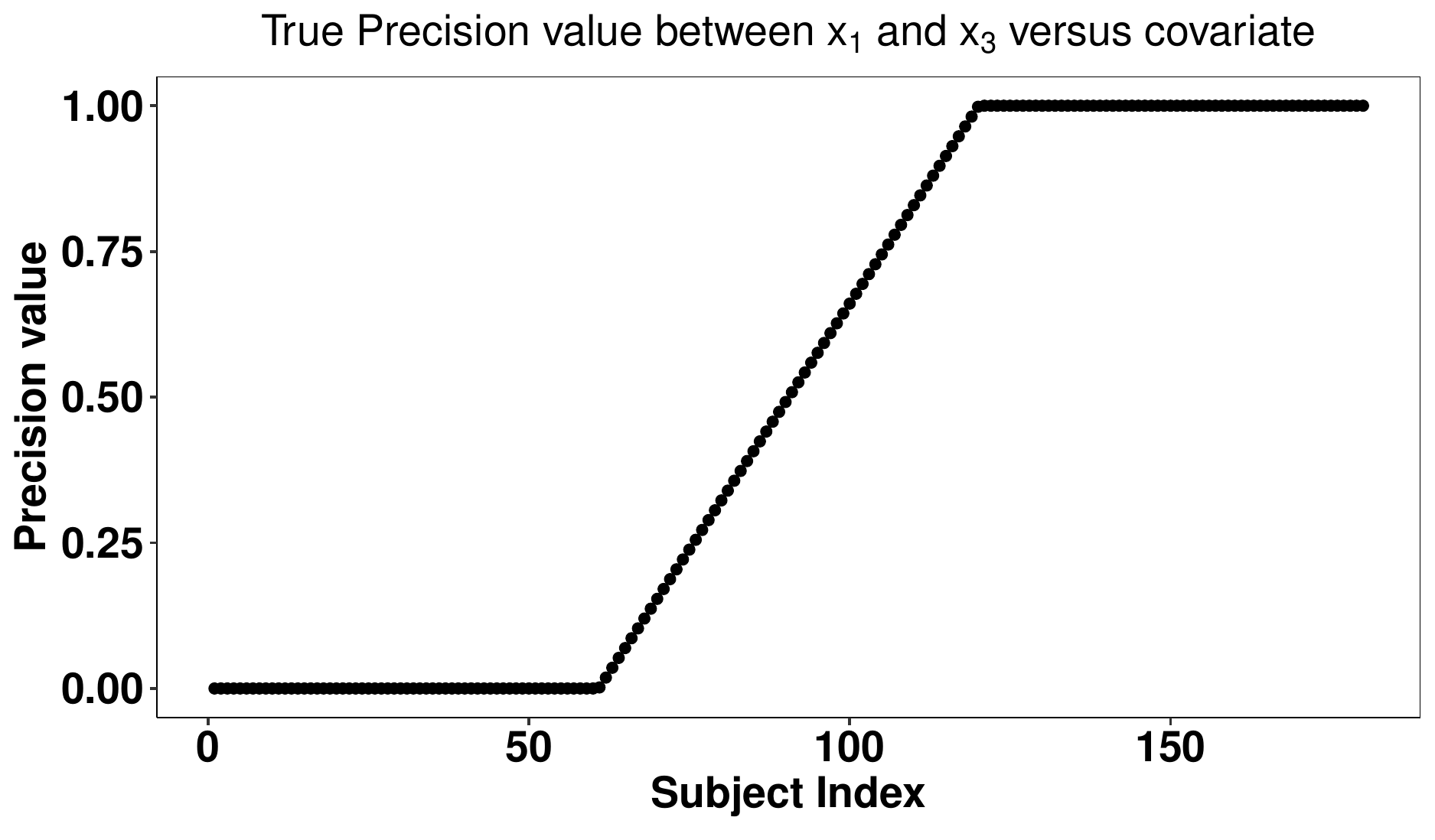} &
\includegraphics[width=0.3\linewidth]{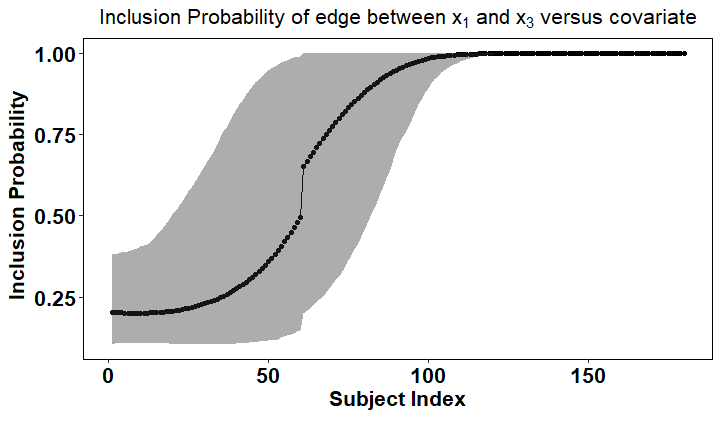}\\
\end{tabular}
\caption{\it Left: True precision value for the edge between Variable 1 and 2 (top panel); and Variable 1 and 3 (bottom panel). Right: Corresponding mean inclusion probabilities across 50 simulations, and $95\%$ confidence interval of the probabilities.}
\label{fig5cs}
\end{center}
\end{figure}


\subsection{Multidimensional Covariate}\label{sec:multi}

We next consider a 2-dimensional covariate $\bz\in [-3,3]\times[-3, 3]$. We define the $j,k$ entry of the ground-truth precision matrices similar to the 1-dimensional case as $\Omega^i_{j,k} = 2$ if $j = k$,  $\Omega^i_{j,k} = 1$ if $(j,k)\in\{(2,3), (3,2)$, $\Omega^i_{j,k} = \ind\{\bz_{i1} < 1\}\cdot \min(1, \frac12 - \frac12\bz_i)$ if $(j,k)\in\{(1, 2), (2, 1)\}$ and  $\Omega^i_{j,k} =  \ind\{\bz_{i2} > -1\}\cdot \min(1, \frac12 + \frac12\bz_i)$ if $(j,k)\in\{(1, 3), (3, 1)$.
The ground truth precision structures are given in Supplement \ref{mul_gr}. We generate a sample of size $n=225$ by sampling uniformly $25$ times from each of the 9 sets generated by taking the Cartesian product of the intervals resulting from partitioning the horizontal and vertical axes of the covariate space into intervals of length $2$. 

In the unidimensional continuous covariate setting, $\bz$ may be thought of as indexing time. Thus, mgm and loggle can both be directly compared to W-PL. However, to include these methods in our multidimensional covariate experiments, a reduction to the dimensionality of the covariate is necessary, as neither model can directly handle a multidimensional extraneous covariate. To do this, we apply a greedy sorting algorithm that re-indexes $\bz_1,...,\bz_n$ to $\bz_{(1)},...,\bz_{(n)}$. First, we set $\bz_{(1)} = \bz_1$. Then, at the $t$-th step of the algorithm, $t>1$, we define $\mathcal S_t = \{\bz_1,...,\bz_n\}\setminus\{\bz_{(1)},...,\bz_{(t-1)}\}$ as the covariates that have not yet been sorted and set 
\[\bz_{(t)}=\underset{\bz\in\mathcal S_t}{\arg\min}\lVert \bz -\bz_{(t-1)}\rVert\]

This gives us a bijection $\xi$ mapping from $z_1,...,z_n$ to $z_{(1)},...,z_{(n)}$. We use this mapping to define the 1-dimensional covariate $\mathbf{v}\in1,...,n$ mimicking a time index for loggle and mgm, where $\mathbf{v}_l=l'$ if, and only if, $\xi(\bz_l)=\bz_{(l')}$, i.e., the $l'$-th timepoint is the individual whose covariate was sorted to the $l'$-th position. To demonstrate the fairness of this reduction of the covariate, we apply W-PL both to the original covariate $\bz$, as well as to the time-indexing covariate $\mathbf{v}$. We refer to the results from the latter as time-varying W-PL (tv W-PL).

We present the results from this experiment in Table \ref{tab:multi_cov}. W-PL has the best sensitivity of the $4$ considered methods across all of the considered dimensionalities. The sensitivity of tv W-PL is less than that of W-PL, but still greater than loggle and mgm in all of the experiments, which is expected, given the results of the previous experiments. Although the differential between W-PL and tv W-PL is only about $0.05$ in each experiment, this demonstrates the importance of utilizing covariate information fully in achieving optimal performance. Thus, in addition to the ability to model the precision matrix as varying continuously, another key attribute of W-PL is its ability to directly incorporate a multidimensional covariate into the estimation procedure. 

\begin{table}[h]
    \centering
\begin{tabular}[t]{llll}
\toprule
$p$ & Method & Sensitivity$(\uparrow)$ & Specificity$(\uparrow)$\\
\midrule
 & W-PL & $\mathbf{0.8890} (0.1023)$ & $\mathbf{0.9968} (0.0035)$\\

 & tv W-PL & $0.8363 (0.1221)$ & $0.9906 (0.0063)$\\

 & loggle & $0.6360 (0.1101)$ & $0.9917 (0.0076)$\\

\multirow{-4}{*}{\raggedright\arraybackslash 10} & mgm & $0.6273 (0.2031)$ & $0.9963 (0.0047)$\\
\cmidrule{1-4}
 & W-PL & $\mathbf{0.8216} (0.1250)$ & $0.9995 (0.0004)$\\

 & tv W-PL & $0.7689 (0.1399)$ & $0.9976 (0.0012)$\\

 & loggle & $0.5322 (0.1159)$ & $0.9996 (0.0004)$\\

\multirow{-4}{*}{\raggedright\arraybackslash 30} & mgm & $0.5173 (0.1422)$ & $\mathbf{0.9998} (0.0003)$\\
\cmidrule{1-4}
 & W-PL & $\mathbf{0.8399} (0.1173)$ & $0.9997 (0.0002)$\\

 & tv W-PL & $0.7886 (0.1211)$ & $0.9980 (0.0007)$\\

 & loggle & $0.4809 (0.1054)$ & $0.9997 (0.0002)$\\

\multirow{-4}{*}{\raggedright\arraybackslash 50} & mgm & $0.4792 (0.0749)$ & $\mathbf{1.0000} (0.0000)$\\
\bottomrule
\end{tabular}    
\caption{\it Results for the continuous multidimensional covariate-dependent setting}
    \label{tab:multi_cov}
\end{table}

\section{Real data analysis}\label{sec:real}
The notion of {\it non-homogeneous} underlying graphical structure is particularly significant in the field of cancer research, because it is well known that cancer initiates and evolves through coordinated changes across multiple molecular levels, networks and pathways. This causes the underlying graph to vary across individuals depending on demographics, genetic markers, and other biological factors(\cite{bolli2014heterogeneity,lohr2014widespread}). These factors can be looked at as extraneous covariates which contain valuable information about how the underlying graph structure varies across the individuals.

We use data on patients with Breast Invasive Carcinoma (BRCA) from The Cancer Genome Atlas (TCGA) program website at \url{http://www.compgenome.org/TCGA-Assembler/}. We consider 70 patients with Breast Invasive Carcinoma and 30 patients with normal cells. ``FOXC2" is a gene that is well known to be associated with breast cancer, as discussed in \cite{mani2007mesenchyme}.  We use the unnormalized copy number variation (\texttt{cnv}) values of the gene as our choice of the covariate. We notice that the \texttt{cnv} values were very similar among the normal cells and were concentrated in the range of $1.83$ to $2.06$. However, the values were much more varied among the cancer cells, as shown in the left panel of Figure \ref{realcov}. We estimate the graph dependence structure among the protein expression values of eight genes corresponding to the individuals in our study by treating the given \texttt{cnv} values as continuous associated covariates. The eight genes considered were ``CTNNB1",``BRCA2", ``MET", ``E-cadherin", ``N-cadherin", ``NFkB1", ``snail" and ``STAT3". Based on the covariate values, we describe the \texttt{cnv} to be ``under-expressed" if the values are less than $1.6$, ``normally expressed" if the values are between $1.6$ and $2.1$, and ``over-expressed" if the values are greater than $2.1$. As opposed to the hyperparameter specification scheme used in Section \ref{sec:sims}, here, we utilize the scheme described in Supplement \ref{hidim}.

Figure \ref{realcom} shows the estimated dependence structure of three individuals with different levels of ``FOXC2" \texttt{cnv} expression. There is a visible evolution of the dependence structure as the covariate value changes. In particular, we focus on the edge between ``N-Cadherin" and ``NFkB1" which is present in the under-expressed ``FOXC2" gene, but is otherwise not present. The inclusion probability with covariate value is shown in the right panel of Figure \ref{realcov}. We notice a steady decrease of the inclusion probability as the expression level of the ``FOXC2" \texttt{cnv} increases. The sharp jump on the right is probably because of the sparsity of data points in that neighborhood resulting in inaccurate estimation. N-cadherin is known to promote breast cancer irrespective of the E-cadherin levels, as discussed in \cite{nieman1999n}. However, NFkB1 is known to promote breast cancer by suppression of E-cadherin expression in cells, as discussed in \cite{chuahl2007nf,criswell2007modulation} and others. Our study corroborates this observation, as  we do notice a significant change in the dependence pattern between ``E-cadherin" and ``NFkB1" at different expression levels of the ``FOXC2" gene. For normally expressed cells there is a significant dependence between the protein expressions of the two genes. However, for under-expressed or over-expressed cells, the dependence is no longer present. This is displayed in the middle panel of Figure \ref{realcov} where we notice a sharp peak in inclusion probability for the normally expressed cells only, except for the outliers on the right.

 \begin{figure}[htbp]
\begin{center}
\begin{tabular}{ccc}
\includegraphics[width=0.33\linewidth]{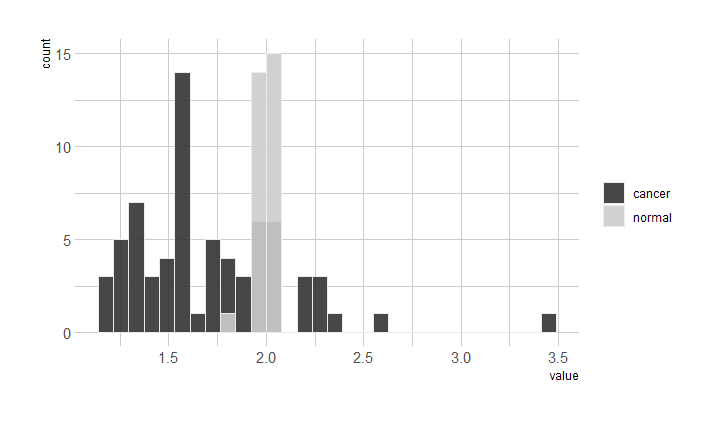}&
\includegraphics[width=0.3\linewidth]{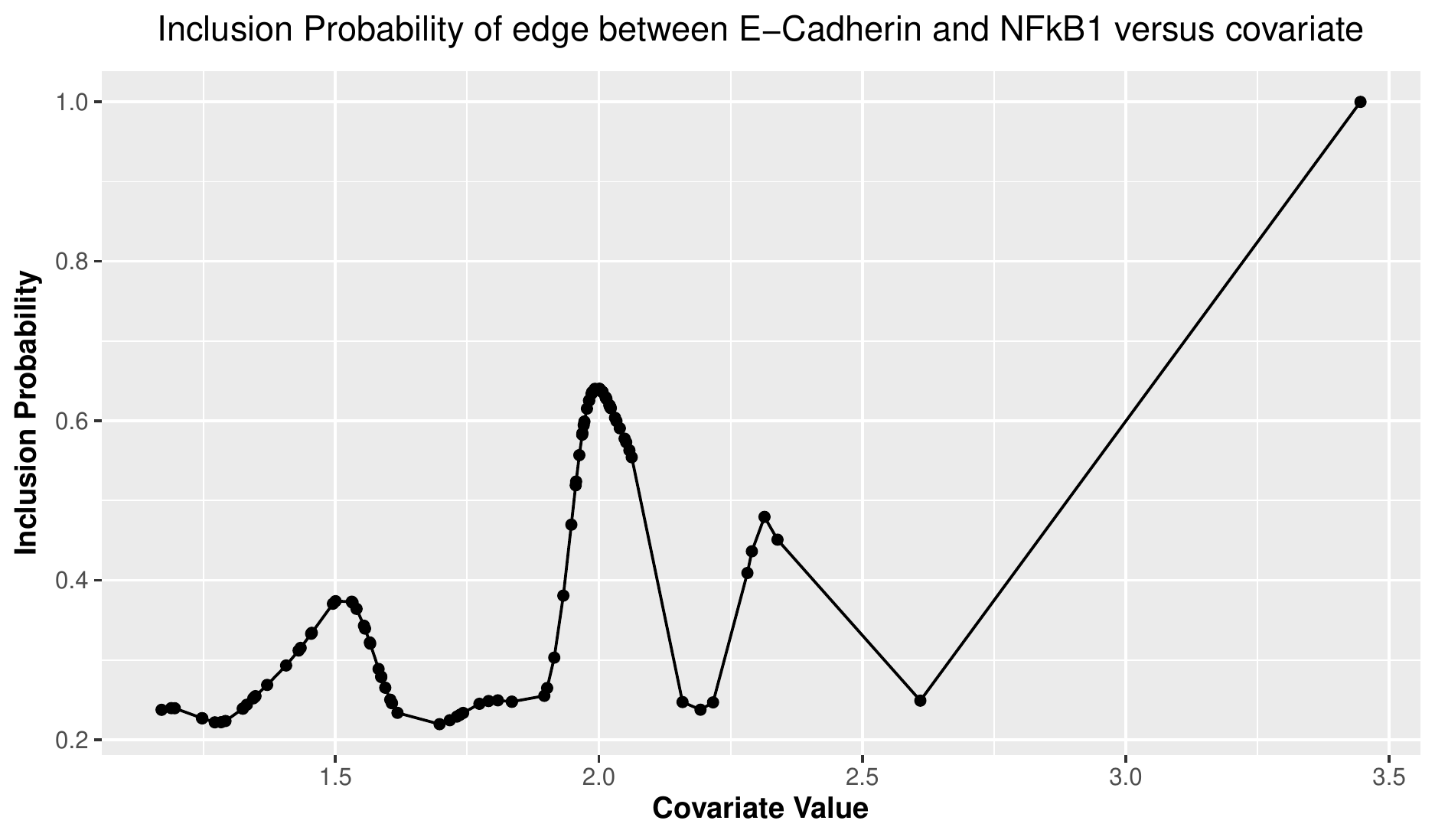}&
\includegraphics[width=0.3\linewidth]{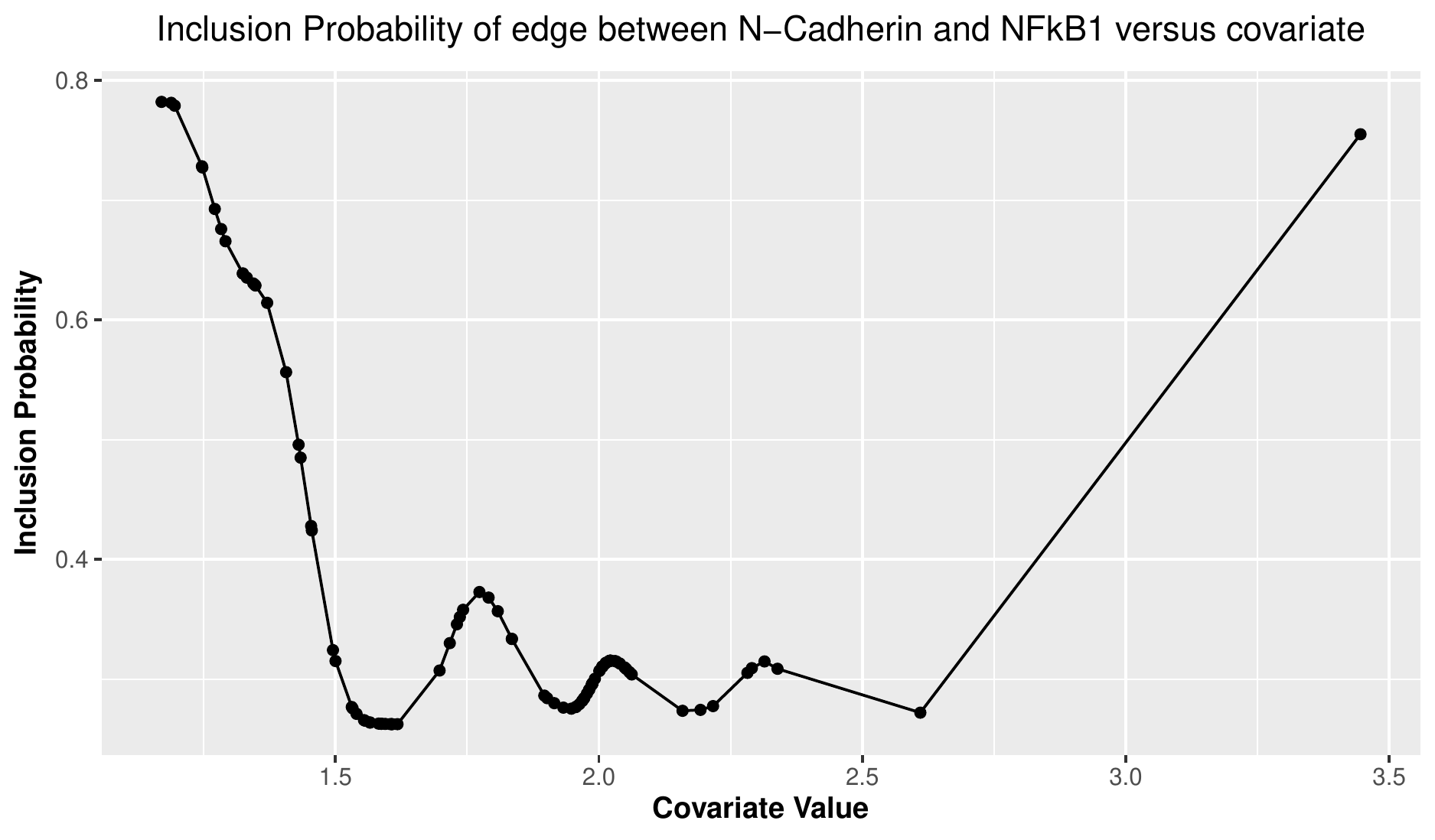}\\
\end{tabular}
\caption{\it Left: Histogram of covariate value for normal cells (lighter shade) versus cancer cells (darker shade). Middle: Inclusion Probability between E-Cadherin and NFkB1 versus covariate values. Right: Inclusion Probability between N-Cadherin and NFkB1 versus covariate values.}
\label{realcov}
\end{center}
\end{figure}

\begin{figure}[htbp]
\begin{center}
\begin{tabular}{|c|c|c|}
\hline
\includegraphics[width=0.33\linewidth]{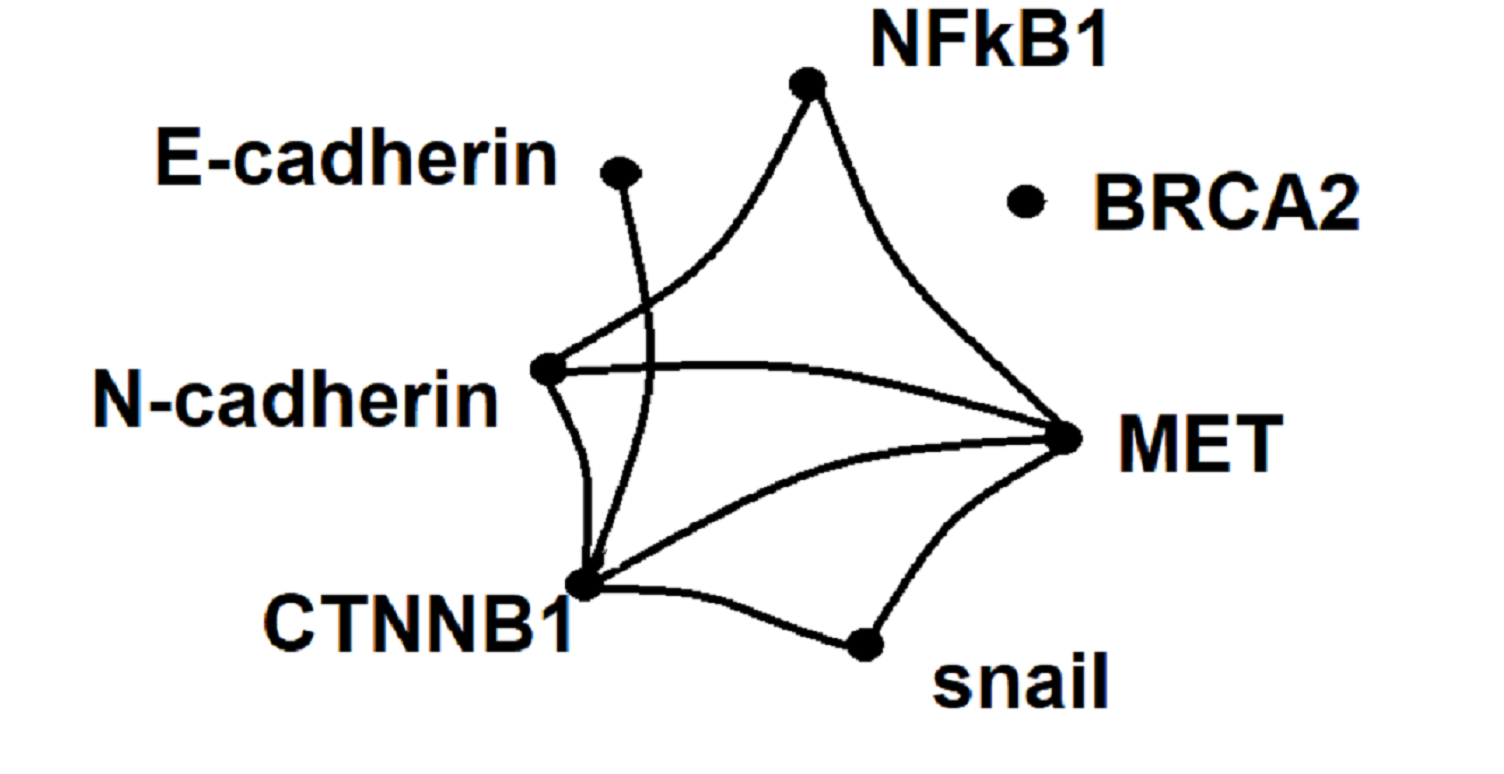}&
\includegraphics[width=0.33\linewidth]{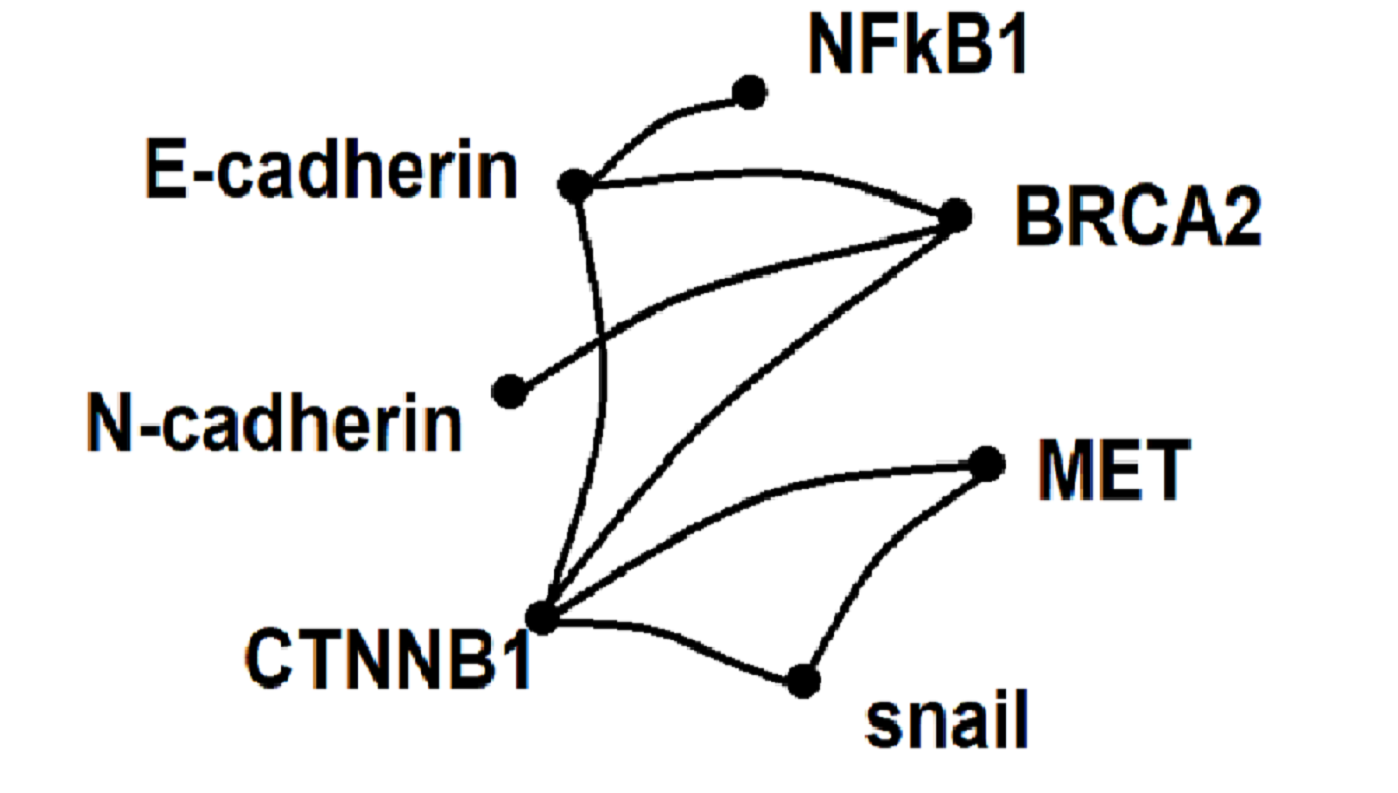}&
\includegraphics[width=0.33\linewidth]{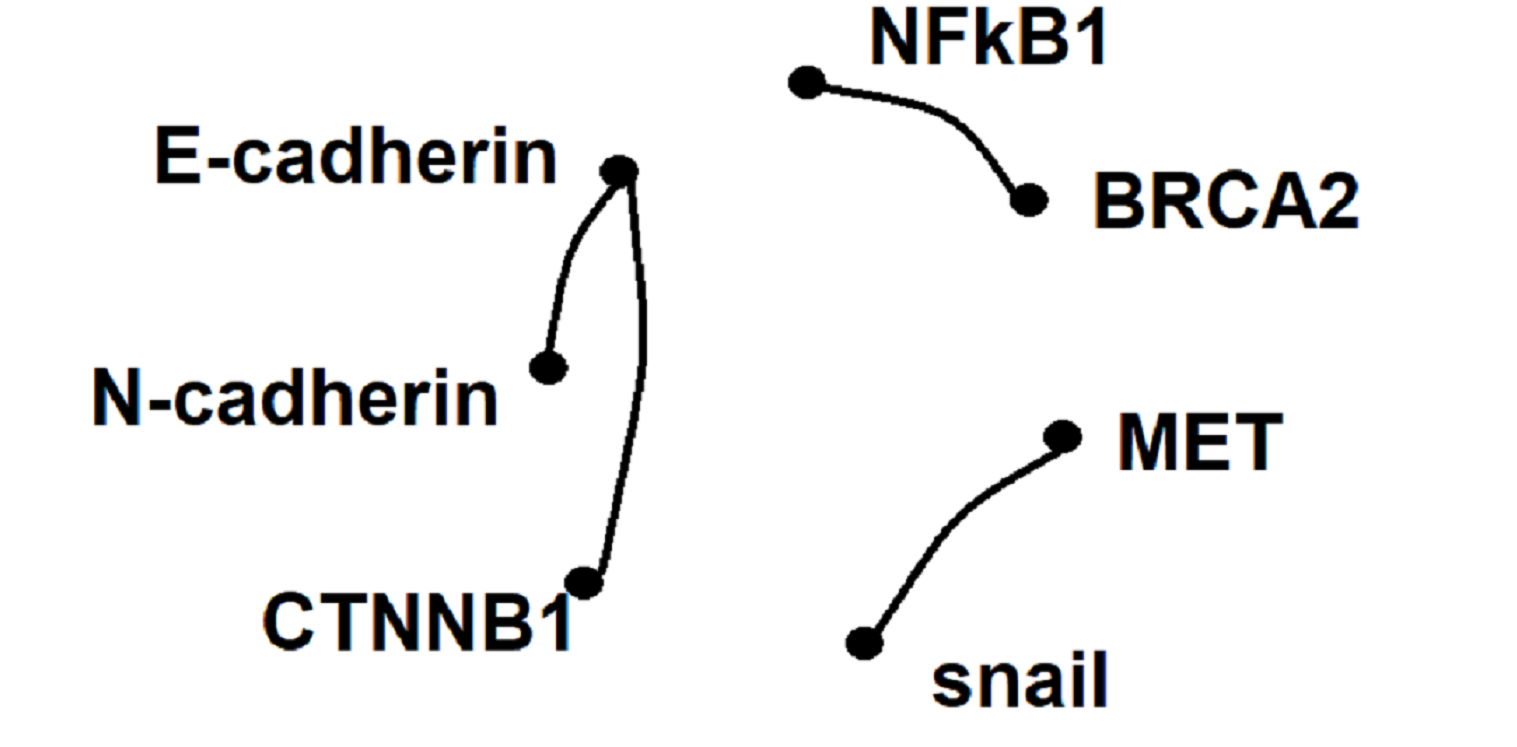}\\
\hline
\end{tabular}
\caption{\it Predicted network structures for individuals with under expressed (left), normally expressed (middle) and over expressed (right) FOXC2 gene.}
\label{realcom}
\end{center}
\end{figure}

\section{Discussion}
In this article, we have introduced a novel weighted-pseudo likelihood approach that can provide an estimate of the underlying dependence structure at an individual level using extraneous covariate information. An appealing feature of the proposed approach is that the performance of the estimates does not suffer when the underlying structure does not actually depend on the extraneous covariates, which we demonstrate in Supplement \ref{sec:covindep}. The variational approach, together with the embarrassingly parallel structure of the parameter estimation avoids the computational complexities associated with running a full-blown Markov chain Monte Carlo. In addition, we also established optimal risk bounds of the proposed method, demonstrating that the approximation through either the variational inference or the pseudo-likelihood framework does not hinder the statistical properties of the method. The theory further demonstrates how borrowing information allows us to obtain a better fit.


Non-Gaussian responses are another direction worth exploring in the future. When the true distribution is non-Gaussian, there is no direct interpretation of the conditional regression coefficients. In contrast, the pseudo-likelihood approach is a practical technique to go beyond the Gaussian assumption by changing the error distribution, see, for instance, \cite{guha2020quantile}. However, it is unclear what true data generation mechanism can be approximated by such a pseudo-likelihood. 

Finally, the detection of high-dimensional graphs could be challenging if the SNR is not high enough, which we explore in Supplement \ref{hidim}. 
This low SNR issue is more prominent for the continuous covariate setting when $\Theta^*_{ij}(\mz)$ is a continuous function of $\mz$, and $\Theta^*_{ij}(\mz)$ takes both zeros and non-zero values. Then by the continuity, $\Theta^*_{ij}(\mz)$  takes values arbitrarily close to zero, where it is challenging to recover the graphs due to low SNR. 

The codes used in our analysis are available online on Github at \href{https://anonymous.4open.science/r/covariate-dependent_graphical_modeling/simulation_study_graph_learning/main.R}{\url{https://anonymous.4open.science/r/covariate-dependent_graphical_modeling/simulation_study_graph_learning/main.R}}.

\newpage

\section*{Appendices}

\addcontentsline{toc}{section}{Appendices}

\renewcommand{\thesubsection}{\Alph{subsection}}

\spacingset{1.25}

\subsection{Assumptions for the theoretical results}\label{ssec:conditions}

\subsubsection{Assumptions for continuous covariate-dependent model}
{\em Assumptions on true data generating distribution:} 

\noindent {\bf Assumption T1 (Sparsity in $\beta$)}. Assume that $\beta_j^*(\mathrm{z})$ has at most $s_j^*$ non-zero elements for any $\mathrm{z}$ in its support, and let $s^*= \max_{j=1}^{p}s_j^*$ with $s^*\geq 1$. Suppose $n \geq c \max\{s^*\log p, s^*\log n\}$ for some constant $c>1$. In addition, we assume $n =o(p^{C})$ for some positive constant $C>0$. 

\noindent {\bf Assumption T2 (Sparsity in derivatives)}. Assume that for  any $\mathrm{z}$,  $\dot{\beta}_{j}^{*}(\mathrm{z})$ and $\ddot{\beta}_{j}^{*}(\mathrm{z})$  have at most $c_5 s_j^*$ non-zero elements for $j=1,...,p$, for some constant $c_5>0$.

\noindent {\bf Assumption T3 (Smoothness)}. We assume that up to the second order derivatives for all the components of the graph coefficient with respect to the covariates  are uniformly bounded by a constant. That is, $\|\beta_{j}^{*}(\mathrm{z})\|_\infty$,$\|\dot{\beta}_{j}^{*}(\mathrm{z})\|_\infty$, $ \|\ddot{\beta}_{j}^{*}(\mathrm{z})\|_\infty$ are uniformly bounded above by constants for any $\mathrm{z},j$. 

\noindent {\bf Assumption T4 (Random design)}. Suppose  $\bz_1,...,\bz_n$ are i.i.d. samples from a distribution with density $f(\mathrm{z})$ on a compact support, where    $|f(\mathrm{z})|$, $\dot{f}(\mathrm{z})$ and  $\ddot{f}(\mathrm{z})$ are all bounded below and above by constants.

\noindent {\bf Assumption T5 (Eigenvalue Conditions)}.  All eigenvalues of $\Sigma(\mathrm{z})$ and $\dot{\Sigma}(\mathrm{z})$ are uniformly upper and lower bounded by constants for any $\mathrm{z}$. In addition, suppose that the marginal distribution of $x_i$ with density $\int f(x \mid \Sigma(\mathrm{z}))f(\mathrm{z}) dz$  is a sub-Gaussian random vector with  covariance $\Sigma$. All eigenvalues of $\Sigma$ are also uniformly upper and lower bounded by constants. 

\noindent  {\em Assumptions on the model and prior:}

\noindent  {\bf Assumption K (Kernel property)}. Suppose the kernel function $K$ used to fit weights satisfies: $\sup_x|K(x)| \leq c_3<\infty$,  $\int K(x) dx =1$, $\int K^2(x) dx =c_0< \infty$, $\int xK(x)dx =0$,$\int x^2 K(x) dx =c_2 <\infty$.

\noindent  {\bf Assumption P}. We assume a spike-and-slab prior $p_{\beta_j \mid \gamma_j}(\beta_j)p_{\gamma_j}(\gamma_j)$ for the parameter $\theta_j=(\beta_j,\gamma_j)$ with $p_{\beta_j \mid \gamma_j}=\prod_{k=1,k \neq j}^p{\mathcal{N}(\beta_{jk};0,\sigma_*^2)}^{\gamma_{jk}}{\delta_0(\beta_{jk})}^{1-\gamma_{jk}}$, and $p
_{\gamma_j}(\gamma_j) \geq \exp\{-C{\|\gamma_j\|}_0\log p\}$, where ${\|\gamma_j\|}_0$ is the number of non-zero elements of $\gamma_j$. 



{\bf Assumption T1} describes a relationship between $n,p$ and $s$. A constraint of $n \geq c s^*\log p$ is assumed to ensure that the restricted eigenvalue conditions hold for sample covariances, see \cite{raskutti2010restricted,zhou2009restricted}. In addition, $n \geq c s^*\log n$ is assumed to guarantee that the error rate is $O(1)$ in the case when $n>p$. Finally, $n = o(p^{C})$ ensures the consistency holds with a high probability when a bound for the maximum of risks  across subjects is considered. Since the risk bound for a single subject holds with probability $1-p^{-C}$, the maximum of the risk can be bounded with  probability $1-np^{-C}$ using the union bound which requires $n=O(p^C)$.  {\bf Assumption T2} posits sparsity in the first and second derivatives of $\beta(\mathrm{z})$ with $\mathrm{z}$. Essentially, the assumption implies that $\beta_{jk}(\mathrm{z})$ satisfies $\beta_{jk}(\mathrm{z})=0$ for $\mathrm{z} \in [a_{jk},b_{jk}]$, where $[a_{jk},b_{jk}]$ is a constant length interval in the support of $\beta_{jk}(\mathrm{z})$ for $j=1,...,p$, $k=1,...,(p-1)$. Here, $\beta_{jk}(\mathrm{z})$ denotes the $k$-th coordinate of $\beta_j$ as a function of $\mathrm{z}$. One such example is $(1- \mathrm{z}^2)^2 I(|\mathrm{z}| \leq 1)$ for $\mathrm{z} \in [-2,2]$, which is zero for $\mathrm{z} \in [-2,-1]$ and $\mathrm{z} \in [1,2]$. 
{\bf Assumption T3} ensures that the covariates carry information about the graph coefficients, together with some regularity conditions on the covariance matrix.  {\bf Assumption T4} asserts that the sampled covariates are representative in the sense that they are i.i.d. from some homogeneous distributions, e.g., uniform distributions on a bounded interval.     {\bf Assumption K} indicates that the kernel function should be smooth enough to capture the shared information across subjects. The {\bf Assumption P} for priors encompasses a wide variety of prior distributions, as discussed in \cite{castillo2012needles}.

\subsubsection{Assumptions for discrete covariate-dependent/covariate-independent graph models}

\noindent{\bf Assumption W}: Let  $c_{l}= \underset{k: z_{k} \neq z_{l}}{\mbox{min}}{|\bz_k -\bz_l|}$ and assume that $c_l$ is lower bounded by positive constants for $l=1,...,K$. Suppose the Gaussian kernel $K(x) \propto e^{-x^2}$ is used and the tuning parameter $\tau$ for the kernel satisfies $\tau= c \min_l c_l/\sqrt{\log n} $ for some positive constant $c<1$. Then we have the following result, for a graph estimate of an individual with covariate level $\bz_l$.

\noindent{\bf Assumption T}: We assume that the underlying data at each covariate level is generated from a homogeneous dependence structure, as described in \eqref{homogen}.
Given a covariate value $\mz$, ${\Omega}^*(\mz)$ has maximum and minimum eigenvalues bounded away from $0$ and $\infty$. Assume that $\beta_j^*(\mz)$ has at most $s_j^*$ non-zero elements, and let $s^*= \underset{j}{\max}\{ s_j^*\}$.

\noindent{\bf Assumption A}: Assume that $\beta_j^*$ has at most $s_j^*$ non-zero elements, and let $s^*= \underset{j}{\max}\{ s_j^*\}$.  
Assume that $s^* \log (np) /n \rightarrow 0$ for $j=1,...,p$.

\subsection{Proofs of main theorems}\label{ssec:proofs}
\paragraph{Notations.} We first define the following terms:
$$\klt (s, \Omega^*(\mz)) = \inf \left \{ \frac{u^{\mathrm{T}}{\Omega}^*u(\mz)}{n\|u\|_2^2}: u \in \mathbb{R}^p, 1 \leq {\|u\|}_0 \leq s\right \},$$
$$\tilde{\kappa} (s, \Omega^*(\mz))=\sup \left \{ \frac{u^{\mathrm{T}}{\Omega}^*u(\mz)}{n\|u\|_2^2}: u \in \mathbb{R}^p, 1 \leq {\|u\|}_0 \leq s\right \},$$
where $\klt (s, \Omega^*(\mz))$ is the minimum eigenvalue, and $\tilde{\kappa} (s, \Omega^*(\mz))$ is the maximum eigenvalue of ${\Omega}^*(\mz)$ for $s$-sparse matrices. 
Define $\gamma_{jk}(\mz)=\mathbb{I}\{\beta_{jk}(\mz) \neq 0\}$ which denotes the number of non-zero entries of the coefficient parameter $\beta_{jk}(\mz)$.
Also, define:
$$\klt (s,\bX) = \inf \left \{ \frac{u^{\mathrm{T}}{\bX}^{\mathrm{T}}\bX u}{\|u\|_2^2}: u \in \mathbb{R}^p, 1 \leq {\|u\|}_0 \leq s\right \}, \quad \tilde{\kappa} (s,\bX)=\sup \left \{ \frac{u^{\mathrm{T}}{\bX}^{\mathrm{T}}\bX u}{\|u\|_2^2}: u \in \mathbb{R}^p, 1 \leq {\|u\|}_0 \leq s\right \}$$
where $\klt (s,\bX)$is the minimum and $\tilde{\kappa} (s,\bX)$ is the maximum eigenvalue of $\bX$. Let $s^*=\max\{s_j^*, j=1, \ldots, p\}$, where $s_j*$ be the true sparsity structure of the $j$-th row of ${\Omega}^*(\mz)$.
Let 
$$\klb (s,\bX) = \underset{\gamma_j, {\|\gamma_j\|}_0 \leq s}{\inf} \inf \left \{ \frac{u^{\mathrm{T}}{\bX}^{\mathrm{T}}\bX u}{n\|u\|_2^2}: u \in \mathbb{R}^p, \sum_{k:\gamma_{jk=0}} |u| \leq 7 \sum_{k:\gamma_{jk}=1} |u|\right \}
$$ and
$$\klb (\Omega^*(\mz))  =  \inf \left \{ \frac{u^{\mathrm{T}}{\Omega}^*(\mz) u}{\|u\|_2^2}: u \in \mathbb{R}^p, \sum_{k:\gamma_{jk=0}} |u| \leq 7 \sum_{k:\gamma_{jk}=1} |u|\right \}
$$
Define the set $\mathcal{G}_{n,p} = \{\mz \in \mathbb{R}^{n \times p} : \tilde{\kappa} (s^*,\mz) \leq c_1 \tilde{\kappa}(s^*, \Omega^*(\mz)), \tilde{\kappa} (1,\mz) \leq c_1 \tilde{\kappa}(1, \Omega^*(\mz)),\klb(s^*,\mz) \geq c_2 \klb (\Omega^*(\mz))\}$. Then, following \cite{atchade2019quasi} we have that if $\bX \in \mathcal{G}_{n,p}$, then $\bX_{-j} \in \mathcal{G}_{n,p-1}$ for any $j \in \{ 1, 2, \ldots, p\}$ .

Let $A_{-j}$ be the submatrix of matrix $A$ except the $j$-th row and $A_{-j,-j}$ be the submatrix of matrix $A$ except the $j$-th row and column.  	We denote      \begin{equation*}
\Psi^{w}(q_{\theta^l_j}(\bz))= \int  \log \frac{p^{w}(x_j \mid \tilde{\theta}^l_j(\bz),\bX_{-j},\bz)}{p^{w}(x_j \mid {\theta}^l_j(\bz),\bX_{-j},\bz)}{q}_{\theta^l_j}(\theta^l_j(\bz)) d\theta^l_j(\bz) + \alpha^{-1} \mathrm{D_{KL}}({q}_{\theta^l_j} \| p_{\theta^l_j}).
\end{equation*}

Because subject index $l$ do not change during the proof of Theorem 1, Theorem 3, and Corollary 1, we ignore them in these proofs.

\subsubsection{Proof of Lemma~\ref{lem2}}
\begin{proof} 
	For any observation $l$, we have the marginalized KL divergence:    
	\begin{align*}
	& \int p(\bX \mid  \Theta^{*}(\bz_1),...,\Theta^{*}(\bz_n))  \log  \frac{p(\bX \mid  \Theta^{*}(\bz_1),...,\Theta^{*}(\bz_n)  ) }{\prod_{j=1}^p  p^{w_l}(x_{j} \mid \bX_{-j}, \Theta_{-j}^{l}(\bz) )}d\bX  \\         
	& =c- \sum_{j=1}^p  \int p(\bX \mid  \Theta^{*}(\bz_1),...,\Theta^{*}(\bz_n))  \log p^{w_l}  (x_{j} \mid \bX_{-j}, \Theta_{-j}^{l}(\bz) )d\bX  \\
	& = c- \sum_{j=1}^p  \int p(\bX_{-j} \mid \Theta_{-j,-j}^{*}(\bz_1),...,\Theta_{-j,-j}^{*}(\bz_n))p(x_{j} \mid \bX_{-j} ,\Theta_{-j}^{*}(\bz_1), ...,\Theta_{-j}^{*}(\bz_n))\\
	& 	\log p^{w_l}  (x_{j} \mid \bX_{-j}, \Theta_{-j}^{l}(\bz) )d x_{j} d \bX_{-j},   
	\end{align*}
	
	where $\Theta_{-j}^{l}(\bz)$ is the targeted coefficient with $l$-th covariate.
	
	Therefore, $\tilde \beta_j^{l}(\bz)$ should be the minimizer of the following objective function
	\begin{align*}
	&-\int  p(\bX_{-j} \mid \Theta_{-j,-j}^{*}(\bz_1),...,\Theta_{-j,-j}^{*}(\bz_n))p(x_{j} \mid \bX_{-j} ,\Theta_{-j}^{*}(\bz_1), ...,\Theta_{-j}^{*}(\bz_n)) \\
	& \times\log {p^{w_l}(x_j \mid \bX_{-j}, \beta^l_j(\bz),\gamma^l_j(\bz))} d x_j d \bX_{-j} \\
	& =  \E_{\bX_{-j},x_j} \left\{\sum_{k =1}^n (x_{kj} - \bx_{k,-j}^{\mathrm{T}}\beta_j^{l}(\bz))^{\mathrm{T}} \frac{\mathrm{w}_l(\bz_k)}{2\sigma_*^2} (x_{kj} -\bx_{k,-j}^{\mathrm{T}} \beta_j^{l}(\bz))\right\} \\
	& =	 \E_{\bX_{-j}} \left\{\sum_{k =1}^n (\beta_j^{*}(\bz_k) - \beta_j^{l}(\bz))^{\mathrm{T}} \frac{\mathrm{w}_l(\bz_k)}{2\sigma_*^2} \bx_{k,-j}\bx_{k,-j}^{\mathrm{T}} (\beta_j^{*}(\bz_k) - \beta_j^{l}(\bz))\right\} \\
	& =    \sum_{k =1}^n (\beta_j^{*}(\bz_k) - \beta_j^{l}(\bz))^{\mathrm{T}} \frac{\mathrm{w}_l(\bz_k)}{ 2\sigma_*^2}\Sigma^{*}_{-j,-j}(\bz_k) (\beta_j^{*}(\bz_k) - \beta_j^{l}(\bz)),    	
	\end{align*}
	under the constraint $\|\beta_j^{l}(\bz)\|_0 \leq C_0 s_j^*$ for $C_0 \geq 1$.

	Since $\|\beta_j^{*}(\bz_l)\|_0 \leq  s_j^*$ is in the constrained region, by basic inequality, we have
	\begin{align*}
	\sum_{k =1}^n (\beta_j^{*}(\bz_k) - \tilde \beta_j^{l}(\bz))^{\mathrm{T}} \frac{\mathrm{w}_l(\bz_k)}{ 2\sigma_*^2}\Sigma^{*}_{-j,-j}(\bz_k) (\beta_j^{*}(\bz_k) - \tilde \beta_j^{l}(\bz)) \\
	\leq \sum_{k =1}^n (\beta_j^{*}(\bz_k) - \beta_j^{*}(\bz_l))^{\mathrm{T}} \frac{\mathrm{w}_l(\bz_k)}{ 2\sigma_*^2}\Sigma^{*}_{-j,-j}(\bz_k) (\beta_j^{*}(\bz_k) - \beta_j^{*}(\bz_l)),
	\end{align*}
	where $\Sigma^{*}_{-j,-j}(\bz_k)$ is the submatrix of the $k$-th true covariance except the $j$-th row and column. 
	After some algebra, we have
	\begin{equation}\label{eq17}
	({ \beta_j^{l}(\bz)}^* - \tilde \beta_j^{l}(\bz))^{\mathrm{T}}  \sum_{k =1}^n {\mathrm{w}_l(\bz_k)}\Sigma^{*}_{-j,-j}(\bz_k) ({ \beta_j^{l}(\bz)}^* - \tilde \beta_j^{l}(\bz)) \leq 2 (\tilde \beta_j^{l}(\bz)-\beta_j^{*}(\bz_l)) (\sum_{k =1}^n \mathrm{w}_{l}(\bz_k)\Sigma^{*}_{-j,-j}(\bz_k) \beta_j^{*}(\bz_k)  )
	\end{equation}

	Since the eigenvalues of $\Sigma^{*}_{-j,-j}(\bz_k) $ are all lower bounded by constant,  by Weyl's inequality, we have $\lambda_{\min}(\sum_{k =1}^n \mathrm{w}_{l}(\bz_k)\Sigma^{*}_{-j,-j}(\bz_k) ) $ lower bounded by constant multiplied by $\sum_{k =1}^n \mathrm{w}_{k}$.
	
	Note that \begin{equation*}
	\begin{aligned}
	\E  \left\{\sum_{k =1}^n \mathrm{w}_{l}(\bz_k)\right\} = \E\left\{ \frac{1}{\tau} \sum_{k=1}^{n} K\left(\frac{\bz_k-\bz_l}{\tau}\right)\right\} = n \int K \left({ u}\right) f(\bz_l+\tau u) du \\
	= c_0 n f(\bz_l) +o(n).
	\end{aligned}
	\end{equation*}
	
	Therefore, we have $\E(\sum_{k =1}^n \mathrm{w}_{l}(\bz_k)) \geq  cn$ for some positive constant $c$. Applying the above lower eigenvalues and Cauchy-Schwartz in equality on equation~\eqref{eq17} after taking expectation to $\bz$, we have
	\begin{equation*}
	\E_z\| \beta_{j}^{*}(\bz_l) - \tilde \beta_j^{l}(\bz)\|^2_2 \lesssim  \E_z\|\frac{1}{n}\sum_{k =1}^n \mathrm{w}_{l}(\bz_k)\Sigma^{*}_{-j,-j}(\bz_k) (\beta_j^{*}(\bz_k)- \beta_j^{*}(\bz_l)) \|^2_2.
	\end{equation*}
	
	\begin{equation*}
	\frac{1}{n}\sum_{k =1}^n \mathrm{w}_{l}(\bz_k)\Sigma^{*}_{-j,-j}(\bz_k) (\beta_j^{*}(\bz_k)- \beta_j^{*}(\bz_l)) = \frac{c_l}{n\tau} \sum_{k=1}^{n}K\left(\frac{\bz_k-\bz_l}{\tau}\right) \Sigma_{-j,-j}^*(\bz_k)(\beta_{j}^*(\bz_k)-\beta^*_{j}(\bz_l)).
	\end{equation*}

	Then we have
	\begin{equation*}
	\begin{aligned}
	\E_{\bz_k}(\frac{1}{n}\sum_{k =1}^n \mathrm{w}_{l}(\bz_k)\Sigma^{*}_{-j,-j}(\bz_k) (\beta_j^{*}(\bz_k)- \beta_j^{*}(\bz_l))) = \frac{c_l}{\tau} \int K\left(\frac{\bz-\bz_l}{\tau}\right)\Sigma^*_{-j,-j}(\bz) \left(\beta^*_{j}(\bz)-\beta_{j}^*(\bz_l)\right)f(\bz) dz \\
	= c_l\int K(u)\Sigma^*_{-j,-j}(\bz_l+\tau u)(\beta^*_{j}(\bz_l+\tau u)-\beta_{j}^*(\bz_l)) f(\bz_l+\tau u) du.
	\end{aligned}
	\end{equation*}
	
	Expanding $\Sigma^*_{-j,-j}(\bz_l+\tau u)$, $\beta^*_{j}(\bz_l+\tau u)$ and $f(\bz_l+\tau u)$ component-wisely in Taylor expansion, we have 
	\begin{equation*}
	\begin{aligned}
	&\E(\sum_{k =1}^n \mathrm{w}_{l}(\bz_k)\Sigma^{*}_{-j,-j}(\bz_k) (\beta_j^{*}(\bz_k)- \beta_j^{*}(\bz_l)))\\
	&= c_l\int K(u)    \left\{\Sigma^*_{-j,-j}(\bz_l) +\tau \mu \dot{\Sigma}^*_{-j,-j}(\bz_l^{(1)})  \right\}\left\{\tau u \dot{\beta}_{j}^{*}(\bz_l)+ \frac{ \tau^2}{2} u^2 \ddot{\beta}_{j}^{*}(\bz_l^{(2)})\right\} \times \\
	&\left\{f(\bz_l) + u \tau\dot{f}(\bz_l^{(3)})\right\} du \\
	&=c_l\left( \int u K(u) du \right)\tau  \Sigma^{*}_{-j,-j}(\bz_l)\dot{\beta}_{j}^{*}(\bz_l) f(\bz_l) + \\
	&\left(\int u^2 K(u)du \right) \tau^2 \left(  \Sigma^{*}_{-j,-j}(\bz_l) \frac{1}{2}\ddot{\beta}_{j}^{*}(\bz_l^{(2)}) f(\bz_l) + \Sigma^{*}_{-j,-j}(\bz_l)\dot{\beta}_{j}^{*}(\bz_l)\dot{f}(\bz_l^{(3)}) + \dot{\Sigma}^*_{-j,-j}(\bz_l^{(1)}) \dot{\beta}_{j}^{*}(\bz_l) f(\bz_l)\right) +o(\tau^2) \\
	&= c_lc_2 \tau^2  \left( \dot{\Sigma}^*_{-j,-j}(\bz_l^{(1)}) \dot{\beta}_{j}^{*}(\bz_l) f(\bz_l)+  \frac{1}{2}\Sigma^{*}_{-j,-j}(\bz_l)\ddot{\beta}_{j}^{*}(\bz_l^{(2)})f(\bz_l) + \Sigma^{*}_{-j,-j}(\bz_l)\dot{\beta}_{j}^{*}(\bz_l)\dot{f}(\bz_l^{(3)})\right) +o(\tau^2).
	\end{aligned}
	\end{equation*}
	
	where $\bz_l^{(1)},\bz_l^{(2)}, \bz_l^{(3)}$ in the first equation are between $\bz_l$ and $\bz_l+\tau\mu$. Note that the $\ell_2$ norm of the reminder term is no larger than $s^*_j$ up to some constant factor given that $\dot{\beta}_{j}^{*}(\bz_l)$ and $\ddot{\beta}_{j}^{*}(\bz^{(2)}_l)$ are $s^*_j$ sparse and $\|\Sigma^{*}_{-j,-j}(\bz_l) \|_2$ and $\|\dot{\Sigma}^*_{-j,-j}(\bz_l^{(1)})\|_2$ are upper bounded by some constant.

	In addition, denote $a^2$ as the element-wise square for a vector $a$,  the variance can also be similarly calculated:
	\begin{equation*}
	\begin{aligned}
	&\mbox{Var}_{\bz}(\frac{1}{n}\sum_{k =1}^n \mathrm{w}_{l}(\bz_k)\Sigma^{*}_{-j,-j}(\bz_k) (\beta_j^{*}(\bz_k)- \beta_j^{*}(\bz_l)))\\
	&= \mbox{Var}_{\bz}\left(\frac{c_1}{n\tau} \sum_{k=1}^{n}K\left(\frac{\bz_k-\bz_l}{\tau}\right)\Sigma^*_{-j,-j}(\bz_k)(\beta_{j}^*(\bz_k))\right) \\
	& = \frac{c_l^2}{n \tau^2}\E\left[\left(K\left(\frac{\bz_k-\bz_l}{\tau}\right)\Sigma^*_{-j,-j}(\bz_k)\beta_{j}^*(\bz_k)\right) ^2\right]- \frac{1}{n^2}\{\E(\sum_{k =1}^n \mathrm{w}_{l}(\bz_k)\Sigma^{*}_{-j,-j}(\bz_k) (\beta_j^{*}(\bz_k)))\}^2 \\
	& = \frac{c_l^2}{n \tau^2} \int K\left(\frac{\bz-\bz_l}{\tau}\right)^2 (\Sigma^*_{-j,-j}(\bz_k)\beta_{j}^*(\bz_k))^2 f(\bz) dz +o(\frac{1}{n\tau}) \\
	&=\frac{c_l^2}{n \tau} \int K(u)^2  (\Sigma^*_{-j,-j}(\bz_l+u\tau)\beta^*_{j}(\bz_l+u\tau))^{2} f(\bz_l+u\tau) du  +o(\frac{1}{n\tau}) \\
	& =c_l^2 c_0\frac{  (\Sigma^{*}_{-j,-j}(\bz_l)\beta_{j}^{*}(\bz_l))^2 f(\bz_l) }{n \tau} +o(\frac{1}{n\tau}),
	\end{aligned}
	\end{equation*}
	where we use component-wisely Taylor expansion again in the last equation.  Since each component of $\|\Sigma^{*}_{-j,-j}(\bz_l)\|_2$ is bounded by constant and  $\beta_{j}^{*}(\bz_l),\dot{\beta}_{j}^{*}(\bz_l), \ddot{\beta}_{j}^{*}(\bz_l)$ are all $s_j^*$ sparse, we have $\|\Sigma^{*}_{-j,-j}(\bz_l)\ddot{\beta}_{j}^{*}(\bz_l)\|^2_2 \lesssim s_j^* $, $\|\Sigma^{*}_{-j,-j}(\bz_l)\dot{\beta}_{j}^{*}(\bz_l)\|_2^2 \lesssim s_j^*$ and $\|\Sigma^{*}_{-j,-j}(\bz_l)\beta_{j}^{*}(\bz_l)\|_2^2 \lesssim s_j^*$. Therefore, the final conclusion holds by aggregating the bias and variance. 
\end{proof}

\subsubsection{Proof of Theorem~\ref{cor3}}

\begin{proof}
	
	We first prove that given a single subject (the index is omitted for notation simplicity), and a single component $j$,  we have  with probability at least $1-c_2\exp(-c_3 n)-c_4/p^{c_0+1} -\xi$, 
	\begin{equation}\label{eq:single}
	\int \frac{1}{n}d_{\alpha}(\theta_j(\bz), \tilde{\theta}_j(\bz))\hat{q}_{\theta_j}(\theta_j(\bz)) d\theta_j(\bz) \leq C\frac{\alpha}{1-\alpha} \left( \frac{s^*_j \log(np)}{n} + \frac{s_j^*}{n^{\frac{3}{5}}} \right)  +\frac{\log(1/\xi)}{n(1-\alpha)}, 
	\end{equation}
	for positive constants $c_0, c_1,c_2,c_3,c_4,C>0$.

	Define the density function $\tilde{q}_{{\theta}}$ as the restriction of the prior $p_{\theta_j(\bz)}$ restricted in the neighborhood
	\begin{equation}\label{eq:def_theta}
	\begin{aligned}
	\mathcal{N}( \tilde{\theta}_j(\bz), \epsilon) := \{\theta_j(\bz)=(\beta_j(\bz),\gamma_j(\bz)) :  \beta_{j,k}(\bz)=0,\\
	\mbox{ for } \tilde{\beta}_{j,k}(\bz)=0, \mbox{ and } |\beta_{j,k}(\bz) -\tilde{\beta}_{j,k}(\bz)| \leq c_0 \tau \epsilon/\sqrt{s_j^* } \mbox{ for } \tilde{\beta}_{j,k}(\bz) \ne 0\}
	\end{aligned}
	\end{equation}  with $\epsilon = \sqrt{s_j^* \log(np)/n}+\sqrt{s_j^*n^{-3/5}}$ for small enough constant $c_0>0$ and $\tau =n^{-1/5}$. 	Then the measure $\tilde{q}_{\beta_j(\bz),\gamma_j(\bz)}$ belongs to the specified variational family.  The choice of $\epsilon$ is decided by the rate of the misspecified KL ball, which is upper bounded by $\|\mathrm W^{1/2}_l\bX_{-j}(\tilde \beta^l_j(\bz)-\beta_j^{*}(\bz_l))\|^2_2$ as shown in Lemma~\ref{lemC2}.
	
	First, by Lemma~\ref{lemC1}, it follows with probability at least $1-\xi$, we have  
	\begin{equation*}
	\int \frac{1}{n}d_{\alpha,\tilde{\theta}_j(\bz)}(\theta_j(\bz), \tilde{\theta}_j(\bz))\hat{q}_{\theta_j}(\theta_j(\bz)) d\theta_j(\bz) \leq \frac{\alpha}{n(1-\alpha)}  \Psi^{w}(q_{\theta_j(\bz)}) +\frac{\log(1/\xi)}{n(1-\alpha)},
	\end{equation*} 
	for any measure $\hat{q}_{\theta_j} \ll p_{\theta_j(\bz)}$.

	Then, by  Lemma~\ref{lemC2}, we have with probability $1-c_1\exp(-c_2n)-c_3/p^{c_0+1}$,
	$$-\int  \log \left\{ \frac{p^{w}(x_j \mid \theta_j(\bz),\bX_{-j},\bz)} {p^{w}(x_j \mid \tilde{\theta}_j(\bz),\bX_{-j},\bz)}\right\} \tilde{q}_{\theta_j}(\theta_j(\bz)) d\theta_j(\bz) < Dn\epsilon^2.$$
	Finally, by the KL divergence of restricted measure vs. original measure, we have $\mathrm{D_{KL}}(\tilde{q}_{\theta_j} \|  p_{\theta_j(\bz)}) = -\log (p(\theta \in \mathcal{N}( \tilde{\theta}_j(\bz), \epsilon)) )  \lesssim s_j^* \log p + s_j^* \log((s_j^*)^{1/2} \tau^{-1}/\epsilon) \lesssim n\epsilon^2$.  Then equation~\eqref{eq:single} holds given that $  \Psi^{w}(\hat{q}_{\theta_j}) \leq   \Psi^{w}(\tilde{q}_{\theta_j})$.
	
	Given conclusion in equation~\eqref{eq:single}, for each $j$, we choose $\xi = (np)^{-c_5 s^*}$ such that $\log(1/\xi) = c_5 s^* \log(np)$. Then by union bound for $\theta_j(\bz)$, $j=1,...,p$, we have with probability at least,
	$1-c_2p\exp(-c_3 n)-c_4/p^{c_0} -pe^{-  c_5 s^*\log(np) }$, 
	\begin{equation}\label{eq:sub}
	\max_{j=1,...,p}\int \frac{1}{n}d_{\alpha}(\theta_j(\bz), \tilde{\theta}_j(\bz))\hat{q}_{\theta_j}(\theta_j(\bz)) d\theta_j(\bz) \leq C\frac{1+\alpha}{1-\alpha} \left( \frac{s^* \log(np)}{n} + \frac{s^*}{n^{\frac{3}{5}}} \right) 
	\end{equation}
	for positive constants $c_0, c_1,c_2,c_3,c_4,C>0$.
	Finally, by the union bound applying across subject $l=1,...,n$, we have the conclusion of the theorem.
\end{proof}

\subsubsection{Proof of Lemma~\ref{lem:2r}}
\begin{proof} 
	Similarly, for any observation $l$, we have the marginalized KL divergence:    
	\begin{align*}
	& \int p(\bX \mid  \Theta^{*}(\bz_1),...,\Theta^{*}(\bz_n))  \log  \frac{p(\bX \mid  \Theta^{*}(\bz_1),...,\Theta^{*}(\bz_n)  ) }{\prod_{j=1}^p  p^{w_l}(x_{j} \mid \bX_{-j}, \Theta_{-j}^l(\bz) )}d\bX  \\         
	& =c- \sum_{j=1}^p  \int p(\bX \mid  \Theta^{*}(\bz_1),...,\Theta^{*}(\bz_n))  \log p^{w_l}  (x_{j} \mid \bX_{-j}, \Theta_{-j}^l(\bz) )d\bX  \\
	& = c- \sum_{j=1}^p  \int p(\bX_{-j} \mid \Theta_{-j,-j}^{*}(\bz_1),...,\Theta_{-j,-j}^{*}(\bz_n))p(x_{j} \mid \bX_{-j} ,\Theta_{-j}^{*}(\bz_1), ...,\Theta_{-j}^{*}(\bz_n)) \\
	&\log p^{w_l}  (x_{j} \mid \bX_{-j}, \Theta_{-j}^l(\bz) )d x_{j} d \bX_{-j}.  
	\end{align*}
	
	Therefore, $\tilde \beta_j^{l}(\bz)$ should be the minimizer of the following objective function
	\begin{align*}
	&-\int  p(\bX_{-j} \mid \Theta_{-j,-j}^{*}(\bz_1),...,\Theta_{-j,-j}^{*}(\bz_n))p(x_{j} \mid \bX_{-j} ,\Theta_{-j}^{*}(\bz_1), ...,\Theta_{-j}^{*}(\bz_n)) \\
	&\times \log {p^{w_l}(x_j \mid \bX_{-j}, \beta^l_j(\bz),\gamma^l_j(\bz))} d x_j d \bX_{-j} \\
	& =  \E_{\bX_{-j},x_j} \left\{\sum_{k =1}^n (x_{kj} - \bx_{k,-j}^{\mathrm{T}}\beta_j^{l}(\bz))^{\mathrm{T}} \frac{\mathrm{w}_l(\bz_k)}{2\sigma_*^2} (x_{kj} -\bx_{k,-j}^{\mathrm{T}} \beta_j^{l}(\bz))\right\} \\
	& =	 \E_{\bX_{-j}} \left\{\sum_{k =1}^n (\beta_j^{*}(\bz_k) - \beta_j^{l}(\bz))^{\mathrm{T}} \frac{\mathrm{w}_l(\bz_k)}{2\sigma_*^2} \bx_{k,-j}\bx_{k,-j}^{\mathrm{T}} (\beta_j^{*}(\bz_k) - \beta_j^{l}(\bz))\right\} \\
	& =    \sum_{k =1}^n (\beta_j^{*}(\bz_k) - \beta_j^{l}(\bz))^{\mathrm{T}} \frac{\mathrm{w}_l(\bz_k)}{ 2\sigma_*^2}\Sigma^{*}_{-j,-j}(\bz_k) (\beta_j^{*}(\bz_k) - \beta_j^{l}(\bz)),    	
	\end{align*}
	under the constraint $\|\beta_j^{l}(\bz)\|_0 \leq C_0 s_j^*$ for $C_0 \geq 1$.
	
	Since $\|\beta_j^{*}(\bz_l)\|_0 \leq  s_j^*$ is in the constrained region, by basic inequality, we have
	\begin{align*}
	\sum_{k =1}^n (\beta_j^{*}(\bz_k) - \tilde \beta_j^{l}(\bz))^{\mathrm{T}} \frac{\mathrm{w}_l(\bz_k)}{ 2\sigma_*^2}\Sigma^{*}_{-j,-j}(\bz_k) (\beta_j^{*}(\bz_k) - \beta_j^{l}(\bz)) \\
	\times \leq \sum_{k =1}^n (\beta_j^{*}(\bz_k) - \beta_j^{*}(\bz_l))^{\mathrm{T}} \frac{\mathrm{w}_l(\bz_k)}{ 2\sigma_*^2}\Sigma^{*}_{-j,-j}(\bz_k) (\beta_j^{*}(\bz_k) - \beta_j^{*}(\bz_l)),
	\end{align*}
	where $\Sigma^{*}_{-j,-j}(\bz_k)$ is the submatrix of kth true covariance except jth row and column. 
	After some algebra, we have
	\begin{equation*}
	({\beta_j^{l}(\bz)}^* - \tilde \beta_j^{l}(\bz))^{\mathrm{T}}  \sum_{k =1}^n {\mathrm{w}_l(\bz_k)}\Sigma^{*}_{-j,-j}(\bz_k) ({\beta_j^{l}(\bz)}^* - \tilde \beta_j^{l}(\bz)) \leq 2 (\beta_j^{l}(\bz)-\beta_j^{*}(\bz_l)) (\sum_{k =1}^n \mathrm{w}_{l}(\bz_k)\Sigma^{*}_{-j,-j}(\bz_k) \beta_j^{*}(\bz_k)  )
	\end{equation*}

	Since the eigenvalues of $\Sigma^{*}_{-j,-j}(\bz_k) $ are all lower bounded by constant,  by Weyl's inequality, we have $\lambda_{\min}(\sum_{k =1}^n \mathrm{w}_{l}(\bz_k)\Sigma^{*}_{-j,-j}(\bz_k) ) $ lower bounded by constant multiplied by $\sum_{k =1}^n \mathrm{w}_{k}$.

	Suppose that $\bz$ takes $K$ distinct values $z_0^1, z_0^2, \ldots, z_0^K$.

	Note that \begin{equation*}
	\begin{aligned}
	\left\{\sum_{k =1}^n \mathrm{w}_{l}(\bz_k)\right\} = \left\{ \frac{1}{\tau} \sum_{k=1}^{n} K\left(\frac{\bz_k-\bz_l}{\tau}\right)\right\}.
	\end{aligned}
	\end{equation*}
	For a Gaussian kernel, the terms $K((\bz_l - \bz_k)/\tau)$ are bounded away from zero for $\bz_k=\bz_l$, and hence $\sum_{k=1}^n \mathrm{w}_k(\bz_l)\geq c_0 n_l/\tau$ for some positive constant $c_0$.
	Applying the above lower eigenvalues and Cauchy-Schwartz in equality on equation~\eqref{eq17} after taking expectation to $\bz$, we have
	\begin{equation*}
	\| \beta_j^{*}(\bz_k) - \tilde \beta_j^{l}(\bz)\|^2_2 \lesssim  \|\frac{\tau}{n_l}\sum_{k =1}^n \mathrm{w}_{l}(\bz_k)\Sigma^{*}_{-j,-j}(\bz_k) (\beta_j^{*}(\bz_k)- \beta_j^{*}(\bz_l)) \|^2.
	\end{equation*}

	Now, we have     
	
	\begin{equation*}
	\frac{\tau}{n_l}\sum_{k =1}^n \mathrm{w}_{l}(\bz_k)\Sigma^{*}_{-j,-j}(\bz_k) (\beta_j^{*}(\bz_k)- \beta_j^{*}(\bz_l)) = \frac{1}{n_l} \sum_{k=1}^{n}K\left(\frac{\bz_k-\bz_l}{\tau}\right) \Sigma^*_{k,-j,-j}(\beta_{j}^*(\bz_k)-\beta^*_{j}(\bz_l)).
	\end{equation*}
	
	Let  $c_{l}= \underset{k: z_{k} \neq z_{l}}{\mbox{min}}{|\bz_k -\bz_l|}$.    Then we have, using the fact that the kernel is Gaussian,
	\begin{equation*}
	\begin{aligned}
	\frac{\tau}{n_l}\sum_{k =1}^n \mathrm{w}_{l}(\bz_k)\Sigma^{*}_{-j,-j}(\bz_k) (\beta_j^{*}(\bz_k)- \beta_j^{*}(\bz_l)) \leq \frac{c}{n_l} \sum_{l:\bz_k \ne \bz_l} K \left(\frac{\bz_l - \bz_k}{\tau}\right) \Sigma^{*}_{-j,-j}(\bz_k) (\beta_j^{*}(\bz_k)- \beta_j^{*}(\bz_l)) \\
	\leq \frac{c}{n_l}  \exp{ (- c_{l}^2/\tau^2)}  \sum_{l:\bz_k \ne \bz_l} \Sigma^{*}_{-j,-j}(\bz_k) (\beta_j^{*}(\bz_k)- \beta_j^{*}(\bz_l)) .
	\end{aligned}
	\end{equation*}
	where $c$ is a positive constant that changes between steps but does not affect the overall rate.
	Also, note that  $\|\Sigma^{*}_{-j,-j}(\bz_l) \|_2$ is upper bounded by some constant. 	
	Therefore, given the sparsity of $\beta_j^{*}(\bz_k)- \beta_j^{*}(\bz_l)$ and bounded eigenvalues of $\Sigma^{*}_{-j,-j}(\bz_k)$,  we have the $\ell_2$ norm of right hand side of the above inequality is bounded by $c (n-n_l)\sqrt{s_j^*}$, therefore   	
	\begin{equation*}
	\|\tilde \beta_j^{l}(\bz)- \beta_{j}^{*}(\bz_l) \|^2_2 \leq c  \exp{ (- 2c_{l}^2/\tau^2+ 2\log(n/n_l-1))} s_j^*,
	\end{equation*}
	which converges to zero faster than $s_j^*/n$ as long as $c_{l}^2/\tau^2>\log(n) $.
	
\end{proof}

\subsubsection{Proof of Lemma \ref{lemwtequiv}}
\begin{proof}
	
	Based on the proof of Lemma~\ref{lem:2r}, we have  $\tilde \beta_j^{l}(\bz)$ should be the minimizer of the following objective function
	\begin{align*}
	\sum_{k =1}^n (\beta_j^{*}(\bz_k) - \beta_j^{l}(\bz))^{\mathrm{T}} \frac{\mathrm{w}_l(\bz_k)}{ 2\sigma_*^2}\Sigma^{*}_{-j,-j}(\bz_k) (\beta_j^{*}(\bz_k) - \beta_j^{l}(\bz)).   	
	\end{align*}
	Note that under the homogeneous assumption $\beta^{1*}=\beta^{2*}=...\beta^{n*}=\beta^*$,  the objective function becomes
	
	\begin{align*}
	({\beta_{j}}^{*} - \beta_j^{l}(\bz))^{\mathrm{T}} \frac{n}{ 2\sigma_*^2}\Sigma^{*}_{-j,-j} ({\beta_{j}}^{*} - \beta_j^{l}(\bz)).   	
	\end{align*}
	Given that $\Sigma^{*}_{-j,-j}$ is positive definite,  the Kullback-Leibler minimizer satisfies $\tilde{\beta}_j=\beta_j^*$. 
\end{proof}

\subsubsection{Proof of Theorem \ref{thm2}}

We have, 
\begin{eqnarray*}
	\E_{-j}\E_j \exp \left \{ \alpha \frac{p^w(x_j \mid \bX_{-j}, \tilde{\theta}_j(\bz))}{p^w(x_j \mid \bX_{-j}, {\theta}_j(\bz))}\right \} = \exp \left \{ -(1-\alpha) d_{\alpha, \theta_j^*}(\theta_j(\bz), \tilde{\theta}_j(\bz) \mid \bX_{-j}) \right \}.
\end{eqnarray*}
Following the steps of Lemma \ref{lem1}, we have
\begin{equation}
\begin{split}
\E_{-j} \Bigg [ P \bigg (\int (1-\alpha)d_{\alpha,\theta_j^*}(\theta_j(\bz), \tilde{\theta}_j(\bz))\hat{q}_{\theta_j}(\theta_j(\bz)) d\theta_j(\bz) \leq& -\alpha  \int \log \frac{p^w(\xj \mid \tilde{\theta}_j(\bz),\bX_{-j})}{p^w(\xj \mid {\theta}_j(\bz),\bX_{-j})}\hat{q}_{\theta_j}(\theta_j(\bz)) d\theta_j(\bz)\\
&+ \mathrm{D_{KL}} (\hat{q}_{\theta_j} \| p_{\theta_j}) + \log \left(\frac{1}{\zeta}\right) \bigg ) \Bigg ] \geq 1 - \zeta.
\end{split}
\label{expprobwt}
\end{equation}

Define a specific Kullback-Leibler ball around $\tilde{\theta}_j(\bz)$ as
\begin{equation*}
\begin{split}
\mathcal{B}_{n, \theta_j^*}(\tilde{\theta}_j,\epsilon, \bX_{-j}) =&  \left \{ \theta_j : \int \log \frac{p^w(x_j \mid \bX_{-j}, \tilde{\theta}_j(\bz))}{p^w(x_j \mid \bX_{-j}, {\theta}_j)}p(x_j \mid \bX_{-j}, \theta_j^*) d x_j \leq n\epsilon^2, \right. \\
&\left.  \int \log^2 \frac{p^w(x_j \mid \bX_{-j}, \tilde{\theta}_j(\bz))}{p^w(x_j \mid \bX_{-j}, {\theta}_j)}p(x_j \mid \bX_{-j}, \theta_j^*) dx_j \leq n \epsilon^2\right \}.
\end{split}
\end{equation*}

Next, define the set 
$$
\mathcal{A}^w(\bX_{-j},\epsilon)= \left \{x_j : \int  \log \frac{p^w(\xj \mid \tilde{\theta}_j(\bz),\bX_{-j},\bz)}{p^w(\xj \mid {\theta}_j(\bz),\bX_{-j},\bz)}{q}_{\theta_j}(\theta_j(\bz)) d\theta_j(\bz) \leq Dn \epsilon^2\right \} 
$$ for some positive constant $D$. Following the steps of Lemma \ref{lema2}, we have
\begin{eqnarray*}
	P( x_j \notin \mathcal{A}^w(\bX_{-j},\epsilon) )\leq \frac{c_2}{p^{c_1+1}}.
\end{eqnarray*}  

Let the precision matrix of the $p$-variate data generating distribution be $s^*$-sparse and have eigenvalues bounded away from $0$ and $\infty$. Then, it follows that for any $\zeta \in (0,1)$ and $n \geq a_1 s^* \log p$, and any measure $q_{\theta_j} \in \Gamma$ such that $q_{\theta_j} \ll p_{\theta_j}$, we have  
\begin{equation*}
\begin{aligned}
P \Bigg (\int \frac{1}{n}d_{\alpha,\theta_j^*}(\theta_j(\bz), \tilde{\theta}_j(\bz))\hat{q}_{\theta_j}(\theta_j(\bz)) d\theta_j(\bz) \leq \frac{\alpha}{n(1-\alpha)}  \Psi(q_{\theta_j}) \\ + \frac{1}{n(1-\alpha)}\log (1/\zeta) \Bigg )
\geq 1 - \zeta - \frac{c_2}{p^{c_1+1}} - \exp\{-a_2n\}.
\end{aligned}
\end{equation*} for some positive constants $a_1$, $D$ and $a_2$.
In the covariate-independent setup, since $\theta_j^*=\tilde\theta_j(\bz),$ the variational estimate $\hat{q}_{\theta_j}(\theta_j(\bz))$ assumes the following form:
\begin{eqnarray*}
	\begin{aligned}
		\hat{q}_{\theta_j}(\theta_j(\bz))=\underset{q_{\beta_j,\gamma_j}=\prod_{k=1}^p q_{\beta_{jk},\gamma_{jk}}}{\argmin} \Bigg \{-\int \sum_{\delta \in {\{0,1\}}^{p-1}}  \log \frac{p^w(\xj \mid \beta_j(\bz),\gamma_j(\bz),\bX_{-j})}{p^w(\xj \mid \beta_j^*,\gamma_j^*,\bX_{-j})}{q}_{\theta_j}(\theta_j(\bz)) d\theta_j(\bz) + \\
		\alpha^{-1} \mathrm{D_{KL}}({q}_{\theta_j}(\bz) \| p_{\theta_j})\Bigg \}.
	\end{aligned}
\end{eqnarray*}
Since we have $x_j \in \mathcal{A}(\bX_{-j},\epsilon),$ we have
$\int  \log \left\{ p(\xj \mid \theta_j(\bz),\bX_{-j})/p(\xj \mid \theta_j^*,\bX_{-j})\right\}{q}_{\theta_j}(\theta_j(\bz)) d\theta_j(\bz) < Dn\epsilon^2$.

Define the density function $q_{\beta_j,\gamma_j}^*$ as the restriction of the prior in the neighborhood $\mathcal{N}_n ( \theta_j^*, \epsilon)=\{\theta_j=(\beta_j,\gamma_j) :\beta_{j,\gamma_j^*=0}=0; |\beta_{j,\gamma_j} - \beta_{j,\gamma_j}^*| < c_0   \tau \epsilon/\sqrt{s_j^* }, \mbox{for} \, \gamma_j^* \ne 0\}$, where $c_0$ is a sufficiently small constant. 
Then the measure $q_{\beta_j^*,\gamma_j^*}$ belongs to the variational family. Following the steps of the proof of Corollary \ref{indep}, we have the statement of Theorem \ref{thm2}.

\subsubsection{Proof of Corollary \ref{indep}}
When we model the covariate levels independently, the covariate values themselves have no effect on the analysis. This scenario results in Kullback-Leibler balls around the true parameter since the models are well-specified, corresponding to the weights being one for all observations. That is, $p(\xj \mid \theta_j(\bz),\bX_{-j})$ corresponds to \eqref{condwt} with $\mathrm{W}$ as the identity matrix.
Consider the following term
\begin{equation}
\E_{\theta_j^*} \left [ \exp \left (\alpha \log \frac{p(\bX \mid \theta_j(\bz))}{p(\bX \mid \theta_j^*)} \right)\right ] = \sum_{j=1}^n \E_{-j} \left [\E_j \left \{ \exp \left (\alpha \log \frac{p(x_j \mid \theta_j(\bz))}{ p(x_j \mid \theta_j^*)} \right) \Bigm\vert \bX_{-j}\right \} \right ].
\label{exp}
\end{equation}
Note that the expectation on the left hand side of \eqref{exp} is with respect to the original data distribution (multivariate Gaussian), whereas the expression within the $j$-th expectation is with respect to the conditional distribution ${p}(x_j \mid \bX_{-j},\theta_j^*).$ We focus on the $j$-th term on the right hand side, given $\bX_{-j}$. Thus,
$$
\E_j \left [ \exp \left (\alpha \log \frac{p(x_j \mid \theta_j(\bz))}{ p(x_j \mid \theta_j^*)} \right) \Bigm\vert \bX_{-j} \right ] = \exp \left \{ -(1- \alpha)\mathrm{D}_{\alpha}(\theta_j(\bz), \theta_j^*| \bX_{-j})\right \}.
$$
Next we have, for well-specified models,
\begin{equation}
\Psi(q_{\theta_j}(\bz))= -\int  \log \frac{p(\xj \mid \theta_j(\bz),\bX_{-j})}{p(\xj \mid \theta_j^*,\bX_{-j})}{q}_{\theta_j}(\theta_j(\bz)) d\theta_j(\bz) + \alpha^{-1} \mathrm{D_{KL}}({q}_{\theta_j}(\bz) \| p_{\theta_j}).
\label{psiq}
\end{equation}
Then by Lemma~\ref{lem1}, define an $\epsilon$-ball around the true parameter $\theta_j$ as
\begin{eqnarray*}
	\begin{aligned}
		\mathcal{B}_{n} (\theta_j^*, \epsilon, \bX_{-j})= & \left \{\theta_j : \mathrm{D_{KL}}(p(x_j \mid \bX_{-j},\theta_j^*) \| p(x_j \mid \bX_{-j} ,\theta_j)) \leq n\epsilon^2, \right.\\
		& \left. V(p(x_j \mid \bX_{-j},\theta_j^*) \| p(x_j \mid \bX_{-j} ,\theta_j)) \leq n\epsilon^2 \right \}.
	\end{aligned}
\end{eqnarray*}
Here $V(p \, \| \, q)= \int p \log^2 (p/q) dx$ is a discrepency measure called $V$-divergence.
Next, define the following set $\mathcal{A}(\bX_{-j}, \epsilon)$ as
\[
\mathcal{A}(\bX_{-j},\epsilon)=\left \{ x_j : \int \hat{q}_{\theta_j}(\theta_j(\bz)) \log \frac{p(\xj \mid \theta_j(\bz),\bX_{-j})}{p(\xj \mid \theta_j^*,\bX_{-j})} d\theta_j(\bz) \leq Dn\epsilon^2\right \}
\] for some positive constant $D>1$. 

Next, define the set $\mathcal{C}(\bX)=\{ \bX:\tilde{\mbox{A}}(\bX_{-j},x_j) \leq 0, x_j \in \mathcal{A}(\bX_{-j},\epsilon), \bX_{-j} \in \mathcal{G}_{n,p-1})$\}. 
Consider the current problem of high dimensional Bayesian linear regression with spike-and-slab priors for the coefficient parameters. Using the mean-field variational family for the parameter $\theta_j=(\beta_j,\gamma_j)$, we have the following form for $q_{\theta_j}(\theta_j(\bz))$.
\begin{eqnarray*}
	q_{\theta_j}(\theta_j(\bz))=\prod_{k=1}^{p-1} q_{\beta_{jk},\gamma_{jk}}(\beta_{jk}(\bz),\gamma_{jk}(\bz))\\
\end{eqnarray*}

Now, consider the term $\Psi(q_{\theta_j})$ as in \eqref{psiq}. It is combination of a model fit term and a regularization term. In the current setup, the variational estimate $\hat{q}_{\theta_j}(\theta_j)$ assumes the following form:
\begin{eqnarray*}
	\hat{q}_{\theta_j}(\theta_j)=\underset{q_{\beta_j,\gamma_j}=\prod_{k=1}^p q_{\beta_{jk},\gamma_{jk}}}{\argmin} \Bigg \{-\int \sum_{\delta \in {\{0,1\}}^{p-1}}  \log \frac{p(\xj \mid \beta_j(\bz),\gamma_j(\bz),\bX_{-j})}{p(\xj \mid \beta_j^*,\gamma_j^*,\bX_{-j})}{q}_{\theta_j}(\theta_j(\bz)) d\theta_j(\bz) + \alpha^{-1} \mathrm{D_{KL}}({q}_{\theta_j} \| p_{\theta_j})\Bigg \}.
\end{eqnarray*}
In the set $\mathcal{C}(\bX)$, we have $x_j \in \mathcal{A}(\bX_{-j},\epsilon).$ Therefore, 
$\int  \log \left \{ p(\xj \mid \theta_j(\bz),\bX_{-j})/p(\xj \mid \theta_j^*,\bX_{-j})\right \}{q}_{\theta_j}(\theta_j(\bz)) d\theta_j(\bz) < Dn\epsilon^2$.

Define the density function $q_{\beta_j,\gamma_j}^*$ as the restriction of the prior in the neighborhood $\mathcal{N}_n ( \theta_j^*, \epsilon)=\{\theta_j=(\beta_j,\gamma_j) :\beta_{j,\gamma_j^*=0}=0; |\beta_{j,\gamma_j} - \beta_{j,\gamma_j}^*| < c_0   \tau \epsilon/\sqrt{s_j^* }, \mbox{for} \, \gamma_j^* \ne 0\}$, where $c_0$ is a sufficiently small constant. 
Then the measure $q_{\beta_j,\gamma_j^*}$ belongs to the variational family.

If $\bX_{-j} \in \mathcal{G}_{n,p-1}$, we have
\begin{eqnarray*}
	\begin{aligned}
		c_2 \klt (s^*,\bX_{-j}){(\beta_j^* - \beta_j(\bz))}^{\mathrm{T}}  (\beta_j^* - \beta_j(\bz)) \leq {(\beta_j^* - \beta_j(\bz))}^{\mathrm{T}} \bX_{-j}^{\mathrm{T}}\bX_{-j}(\beta_j^* - \beta_j(\bz))\\
		\leq c_1\tilde{\kappa}(s^*,\bX_{-j}){(\beta_j^* - \beta_j(\bz))}^{\mathrm{T}} (\beta_j^* - \beta_j(\bz)).
	\end{aligned}
\end{eqnarray*}
Then, $\mathcal{N}_n ( \theta_j^*, \epsilon) \subset \mathcal{B}_n(\theta_j^*,\epsilon, \bX_{-j})$. Note that since $\bX_{-j} \in \mathcal{G}_{n, p-1}$, by the volume of the neighborhood $\mathcal{N}_n ( \theta_j^*, \epsilon)$,  we have $\mathrm{D_{KL}} ({q}_{\theta_j} \| p_{\theta_j})<- \log p_{\theta_j}[\mathcal{N}_n ( \theta_j^*, \epsilon)]/n(1-\alpha) < \frac{s_j^*}{n(1-\alpha)} \log (s_j^*/\epsilon) + \log p_{\gamma_j}(\gamma_j^*) $ where $s_j^*$ is the number of non-zero entries. Since, $- \log p_{\gamma_j}(\gamma_j^*) \leq s_j^*\log p$, based on Lemma~\ref{lema2}, it follows that with probability at least $1 - \zeta - \frac{c_2}{p^{c_1+1}} - \exp\{-a_2n\}$, for some positive constants $a_1, a_2$ and $D$,
\begin{equation*}
\int \frac{1}{n}d_\alpha(\theta_j(\bz), \theta_j^*)\hat{q}_{\theta_j}(\theta_j(\bz)) d\theta_j(\bz) \leq \frac{\alpha \epsilon^2}{(1-\alpha)} + \frac{s_j^*}{n(1-\alpha)} \log (\frac{s_j^*}{\epsilon})   + \frac{1}{n(1-\alpha)}s_j^* \log p + \frac{1}{n(1-\alpha)}\log \left(\frac{1}{\zeta}\right).
\end{equation*} for all $j \in \{1,2, \ldots, p\}$.
The statement of the Corollary follows from noting that $d_{\alpha}(\Theta(\bz),\Theta^*)=\underset{j}{\max} \, d_\alpha (\theta_j(\bz),\theta_j^*),$ and that $\hat{q}_{\Theta}(\Theta(\bz)) = \prod_{j=1}^p\hat{q}_{\theta_j}(\theta_j(\bz))$, and replacing $n$ with $n_l$. 

\subsection{Auxiliary results}\label{ssec:aux}

\begin{lem}\label{lemC1} Under {Assumptions} in Theorem~\ref{cor3}, for any variational estimate $\hat{q}_{\theta_j}(\theta_j(\bz))$ such that $\hat{q}_{\theta_j} \ll p_{\theta_j}$, we have
	\begin{equation*}
	\begin{aligned}
	P \bigg (\int (1-\alpha)d_\alpha (\theta_j(\bz), \tilde{\theta}_j(\bz))\hat{q}_{\theta_j}(\theta_j(\bz)) d\theta_j(\bz) \leq& -\alpha  \int \log \frac{p^{w}(x_j \mid \theta_j(\bz),\bX_{-j},\bz)}{p^{w}(x_j \mid \tilde{\theta}_j(\bz),\bX_{-j},\bz)}\hat{q}_{\theta_j}(\theta_j(\bz)) d\theta_j(\bz) + \\
	&\mathrm{D_{KL}} (\hat{q}_{\theta_j} \| p_{\theta_j}) + \log \left(\frac{1}{\zeta}\right) \bigg ) \geq 1 - \zeta.
	\end{aligned}
	\end{equation*}
	\label{lem3}
\end{lem}

\begin{proof}
	First
	\begin{equation*}
	\E_{\bz}	\E_{-j} \E_j \exp \left \{\alpha  \log \frac{p^{w}(x_j \mid \theta_j(\bz),\bX_{-j},\bz)}{p^{w}(x_j \mid \tilde{\theta}_j(\bz),\bX_{-j},\bz)} \right \} = \exp \left \{ -(1- \alpha)d_\alpha(\theta_j(\bz), \tilde{\theta}_j(\bz))\right \}.
	\end{equation*}
	Thus, for any $\zeta \in (0,1)$, we have
	\begin{equation*}
	\E_{\bz}	\E_{-j} \E_j \left [\exp \left \{ \alpha \log \frac{p^{w}(x_j \mid \theta_j(\bz),\bX_{-j},\bz)}{p^{w}(x_j \mid \tilde{\theta}_j(\bz),\bX_{-j},\bz)} + (1-\alpha)d_\alpha(\theta_j(\bz), \tilde{\theta}_j(\bz)) - \log (1/\zeta) \right \} \right ]\leq \zeta.
	\end{equation*}
	Integrating both sides of this inequality with respect to the prior distribution $p_{\theta_j}$ and interchanging the integrals using Fubini's theorem, we have
	\begin{equation*}
	\E_{\bz}	\E_{-j} \E_j \int \exp \left \{ \alpha \log \frac{p^{w}(x_j \mid \theta_j(\bz),\bX_{-j},\bz)}{p^{w}(x_j \mid \tilde{\theta}_j(\bz),\bX_{-j},\bz)} + (1-\alpha)d_\alpha(\theta_j(\bz), \tilde{\theta}_j(\bz)) - \log (1/\zeta) \right \} p_{\theta_j} (\theta_j(\bz)) d \theta_j(\bz) \leq \zeta.
	\end{equation*}
	
	Next we use the variational duality of the KL divergence. If $\mu$ is a probability measure and $h$ is a measurable function such that $e^h \in L_1(\mu)$, then
	\begin{align}
	\log \int e^h d\mu = \underset{\rho \ll \mu}{\sup}\left [\int h d\rho - \mathrm{D_{KL}}(\rho \| \mu) \right].
	\end{align}
	We set $h = \alpha \log \frac{p^{w}(x_j \mid \theta_j(\bz),\bX_{-j},\bz)}{p^{w}(x_j \mid \tilde{\theta}_j(\bz),\bX_{-j},\bz)} + (1-\alpha)d_\alpha(\theta_j(\bz), \tilde{\theta}_j(\bz)) - \log (1/\zeta)$ and $\rho=\hat{q}_{\theta_j}(\theta_j(\bz))$ in the above result where $\hat{q}_{\theta_j}(\theta_j)$ is the variational estimate of the fractional posterior distribution.
	
	\begin{equation*}
	\begin{aligned}
	\E_{\bz} \E_{-j} \E_j \exp \Bigg [  \int \Bigg \{ \alpha \log \frac{p^{w}(x_j \mid \theta_j(\bz),\bX_{-j},\bz)}{p^{w}(x_j \mid \tilde{\theta}_j(\bz),\bX_{-j},\bz)} + (1-\alpha)d_\alpha(\theta_j(\bz), \tilde{\theta}_j(\bz)).\\
	- \log (1/\zeta) \Bigg \}\hat{q}_{\theta_j}(\theta_j(\bz)) d\theta_j(\bz) - \mathrm{D_{KL}} (\hat{q}_{\theta_j} \| p_{\theta_j}) \Bigg] \leq \zeta .
	\end{aligned}
	\end{equation*}
	Let
	\begin{equation}
	\tilde{\mbox{A}}(\bX_{-j},\bz)=\int \left \{ \alpha \log \frac{p^{w}(x_j \mid \theta_j(\bz),\bX_{-j},\bz)}{p^{w}(x_j \mid \tilde{\theta}_j(\bz),\bX_{-j},\bz)} + (1-\alpha)d_\alpha(\theta_j(\bz), \tilde{\theta}_j(\bz)) - \log \frac{1}{\zeta} \right \}\hat{q}_{\theta_j}(\theta_j(\bz)) d\theta_j - \mathrm{D_{KL}} (\hat{q}_{\theta_j} \| p_{\theta_j} ),
	\end{equation} and
	\begin{equation}
	\mbox{A}(\bX_{-j},\bz)=\left \{x_j,\bz :	\tilde{\mbox{A}}(\bX_{-j},\bz)\leq 0 \right \}.
	\end{equation}
	Now we apply Markov's inequality to get a probability statement.
	Thus, the required statement follows from the following:
	\begin{equation*}
	\begin{aligned}
	\E_{\bz}	\E_{-j} P \left[\exp \left \{\tilde{\mbox{A}}(\bX_{-j},\bz)\right\} > 1 \mid \bX_{-j},\bz\right ] &\leq 	\E_{\bz} \E_{-j} \E_j \left[\exp \left \{\tilde{\mbox{A}}(\bX_{-j},\bz)\right \} \mid \bX_{-j},\bz\right] \leq \zeta \\
	\E_{\bz}	\E_{-j}  \left\{P(\mbox{A}(\bX_{-j} ,\bz) \mid \bX_{-j},\bz )\right\} &=\int_{\ind_{A(\bX_{-j},\bz)}} f(x_j\mid \bX_{-j} ,\bz) f(\bX_{-j} ,\bz)  dx_j d \bX_{-j} dz \geq 1- \zeta,
	\end{aligned}
	\end{equation*}
	which implies the conclusion.
\end{proof}

\begin{lem}\label{lemC2}
	Under Assumptions in Theorem~\ref{cor3}. Suppose $\mathcal{N}( \tilde{\theta}_j, \epsilon)$ is defined in equation~\eqref{eq:def_theta} and $\tilde {q}_{  \theta_j}(\theta_j(\bz))$ is a density restricted in $\mathcal{N}( \tilde{\theta}_j(\bz), \epsilon)$, then we have
	\begin{equation*}
	P \left (  -\int  \log \left\{ \frac{p^{w}(x_j \mid \tilde \theta_j(\bz),\bX_{-j},\bz)} {p^{w}(x_j \mid \tilde{\theta}_j(\bz),\bX_{-j},\bz)}\right\}  {q}_{ \tilde \theta_j}(\theta_j(\bz)) d\theta_j < Dn\epsilon^2 \right )
	\geq 1 - c_1 e^{-c_2 n} -\frac{c_3}{p^{c_1+1}}.
	\end{equation*} for some positive constants  $c_1,c_2,c_3,D$.
	\label{lem4}
\end{lem}

\begin{proof}
	By the log-likelihood of Gaussian distribution,  for some constant $c,C>0$, we have
	\begin{equation*}
	\begin{aligned}
	\log \left\{ \frac{p^{w}(x_j \mid \tilde{\theta}_j(\bz),\bX_{-j},\bz)} {p^{w}( x_j \mid {\theta}_j(\bz),\bX_{-j},\bz)}\right\}  = C (\|\mathrm{W}_{l}^{1/2} x_j-\mathrm{W}_{l}^{1/2}\bX_{-j}{\beta}_j(\bz)\|^2 -\|x_j-\mathrm{W}_{l}^{1/2}\bX_{-j}\tilde{\beta}_j(\bz) \|_2^2 )+c.
	\end{aligned}
	\end{equation*}
	Denote $x_j = \bX_{-j} \beta^*_j(\bz) +\varepsilon_j$. Note that 
	\begin{equation*}
	\begin{aligned}
	&\|\mathrm{W}_{l}^{1/2}x_j-\mathrm{W}_{l}^{1/2}\bX_{-j}{\beta}_j(\bz)\|^2 -\|\mathrm{W}_{l}^{1/2}x_j-\mathrm{W}_{l}^{1/2}\bX_{-j}\tilde{\beta}_j(\bz) \|_2^2 \\
	&= \|\mathrm{W}_{l}^{1/2}\bX_{-j}\tilde{\beta}_j(\bz) -\mathrm{W}_{l}^{1/2}\bX_{-j}{\beta}_j(\bz) \|^2 + 2 \langle \mathrm{W}_{l}^{1/2}\bX_{-j}(\tilde{\beta}_j(\bz)-\beta_j(\bz)),\mathrm{W}_{l}^{1/2}\bX_{-j}\beta^*_j(\bz_l) +\mathrm{W}_{l}^{1/2}\varepsilon_j -\mathrm{W}_{l}^{1/2}\bX_{-j}\tilde{\beta}_j(\bz)\rangle \\
	&\leq  2\|\mathrm{W}_{l}^{1/2}\bX_{-j}\tilde{\beta}_j(\bz) -\mathrm{W}_{l}^{1/2}\bX_{-j}{\beta}_j(\bz) \|^2 + \|\mathrm{W}_{l}^{1/2}\bX_{-j}\tilde{\beta}_j(\bz) -\mathrm{W}_{l}^{1/2}\bX_{-j}{\beta^*_j}(\bz_l) \|^2 +  2 \langle \mathrm{W}_{l}\bX_{-j}(\tilde{\beta}_j(\bz)-\beta_j(\bz)),\varepsilon_j \rangle
	\end{aligned}
	\end{equation*}
	
	Under Lemma~\ref{lemC3}, with probability at least $1- \exp(-a_1n)$, we have $\bX \notin \mathcal{G}_{n,p}$, this gives us 
	\begin{equation*}
	\|\mathrm{W}_{l}^{1/2}\bX_{-j}\tilde{\beta}_j(\bz) -\mathrm{W}_{l}^{1/2}\bX_{-j}{\beta^*_j}(\bz_l) \|^2 \leq cn\|\mathrm W\|_2\|\tilde{\beta}_j(\bz)-{\beta^*_j}(\bz_l)\|_2^2  \leq c n\tau^{-1} n^{-4/5} s_j^* \leq  c s_j^* n^{2/5} \leq  c n\epsilon^2.
	\end{equation*} In addition, by definition of $\tilde q_{{\theta}_j}(\theta_j(\bz))$, we have
	\begin{equation*}
	2\|\mathrm{W}_{l}^{1/2}\bX_{-j}\tilde{\beta}_j(\bz) -\mathrm{W}_{l}^{1/2}\bX_{-j}{\beta}_j(\bz) \|^2 \leq cn \tau^{-1} \|\tilde{\beta}_j(\bz)-{\beta^*_j}(\bz_l)\|_2^2 \leq  cn \epsilon^2.
	\end{equation*}  
	For the last term $2 \langle \mathrm{W}_{l}\bX_{-j}(\tilde{\beta}_j(\bz)-\beta_j(\bz)),\varepsilon_j \rangle$, first we have
	\begin{equation*}
	2 \langle \mathrm{W}_{l}\bX_{-j}(\tilde{\beta}_j(\bz)-\beta_j(\bz)),\varepsilon_j \rangle \leq 2 \|\tilde{\beta}_j(\bz)-\beta_j(\bz)\|_1 \|\bX^T_{-j} \mathrm{W}_{l} \varepsilon_j \|_\infty \leq c \sqrt{s_j^*} \tau \epsilon \|\bX^T_{-j} \mathrm{W}_{l} \varepsilon_j \|_\infty.
	\end{equation*}
	
	Note that $\bX^T_{-j} \mathrm{W}_{l} \varepsilon_j$ is a $p-1$ dimensional Gaussian vector, with and with probability great than $1-\exp(-a_2 n)$ by Lemma~\ref{lemC3}, the scale of each component of the Gaussian vector is  bounded by $ \sigma^*\sqrt{n}/\tau$ multiplied constant, by maximal inequality of Gaussian random vector, we have
	\begin{equation*}
	P( \|\bX^T_{-j} \mathrm{W}_{l} \varepsilon_j \|_\infty \geq t) \leq e^{-\frac{\tau^2 t^2}{2 \sigma^{*2} n}+\log p},
	\end{equation*} 
	and we can choose $t=  \sqrt{  (c_0+2) n \log p}/\tau$ for a constant $c_0>0$, then the probability upper bound becomes $p^{-(c_0+1)}$. Therefore, we have 
	\begin{equation*}
	2 \langle \mathrm{W}_{l}\bX_{-j}(\tilde{\beta}_j(\bz)-\beta_j(\bz)),\varepsilon_j \rangle \leq c \sqrt{s_j^*} \tau \epsilon \sqrt{n \log p} /\tau \leq c n\epsilon^2 + s_j^* \log p \leq c n\epsilon^2,
	\end{equation*}
	where the in the second inequality we use $2ab \leq a^2 +b^2$. 
\end{proof} 

\begin{lem}\label{lemC3}
	Under {\bf Assumption T1, T3, T4, T5}, we have  $P(\bX \notin \mathcal{G}_{n,p}) \leq \exp\{-a_1n\}$. In addition, with probability at least $1-\exp(-a_2 n)$, we also have the maximal of $\ell_2$ norm of column vectors of $\bX$ satisfies  $ \max_{i=1,...,p} \|X_{i}\|_2 \leq a_3 \sqrt{n}$ for some constant $a_3>0$. 
\end{lem}
\begin{proof}
	In order to bound $P(\bX \notin \mathcal{G}_{n,p})$ we use the Theorem $1.6$ in \cite{zhou2009restricted}: since the covariance function $\Sigma(\bz_i)$ is homogeneous and $\bz_i$, $i=1,...,n$ are i.i.d. samples from $f(\bz)$. Note that $x_1(z_1),...,x_n(z_n)$ are i.i.d samples form distribution $\int f(x \mid \Sigma(\bz))f(\bz)$, which is assumed to be sub-Gaussian by the Assumption~T5.	Thus we have $P(\bX \notin \mathcal{G}_{n,p}) \leq \exp\{-a_1n\}$, for $n$ larger than $a_4s_j^*\log p$, where $a_4$ and $a_1$ are positive constants, by the similar argument with Lemma 3 in \cite{atchade2019quasi}. Therefore, the following restricted eigenvalue conditions hold: $\klt (2s^*,\bX_{-j}^T\bX_{-j}/n)$ and $\tilde{\kappa}(2s^*,\bX_{-j}^T\bX_{-j}/n)$ are constants.
	
	For the second conclusion, first fix $i$, note that all eigenvalues of $\Sigma$ are uniformly  upper and lower bounded by constants. Then by Hanson-Wright inequality \citep{rudelson2013hanson}, we have
	\begin{equation*}
	P(\|X_i\|_2 \geq c_1\sqrt{n} ) \leq 2 e^{-c_2n},
	\end{equation*}
	for some constants $c_1,c_2>0$.
	Then by the union  bound and $n \geq c s^* \log p$, by choosing large enough constant $c_1',c_2'$, we have
	\begin{equation*}
	P(\|X_i\|_2 \geq c_1'\sqrt{n} ) \leq 2 e^{-c_2'n +\log p} \leq 2 e^{-c_3n}.
	\end{equation*}
\end{proof}

\begin{lem} Under assumptions of Theorem~\ref{thm2}, we have
	\begin{equation}
	\begin{split}
	\E_{-j} \Bigg [ P \bigg (\int (1-\alpha)d_\alpha(\theta_j(\bz), \theta_j^*)\hat{q}_{\theta_j}(\theta_j(\bz)) d\theta_j(\bz) \leq& -\alpha  \int \log \frac{p(\xj \mid \theta_j(\bz),\bX_{-j})}{p(\xj \mid \theta_j^*,\bX_{-j})}\hat{q}_{\theta_j}(\theta_j(\bz)) d\theta_j(\bz) \\
	&+ \mathrm{D_{KL}} (\hat{q}_{\theta_j} \| p_{\theta_j}) + \log \left(\frac{1}{\zeta}\right) \bigg ) \Bigg ] \geq 1 - \zeta.
	\end{split}
	\label{expprob}
	\end{equation}
	\label{lem1}
\end{lem}
\begin{proof}
	From \eqref{newdivwt}, we have
	\begin{eqnarray*}
		\E_{-j} \E_j \exp \left \{\alpha  \log \frac{p(\xj \mid \theta_j(\bz),\bX_{-j})}{p(\xj \mid \theta_j^*,\bX_{-j})} \right \} = \exp \left ( -(1- \alpha)d_\alpha(\theta_j(\bz), \theta_j^* \mid \bX_{-j})\right ).
	\end{eqnarray*}
	Thus, for any $\zeta \in (0,1)$, we have
	\begin{eqnarray*}
		\E_{-j} \E_j \left [\exp \left \{ \alpha \log \frac{p(\xj \mid \theta_j(\bz),\bX_{-j})}{p(\xj \mid \theta_j^*,\bX_{-j})} + (1-\alpha)d_\alpha(\theta_j(\bz), \theta_j^*) - \log (1/\zeta) \right \} \right ]\leq \zeta.
	\end{eqnarray*}
	Integrating both sides of this inequality with respect to the prior distribution $p_{\theta_j}$ and interchanging the integrals using Fubini's theorem, we have
	\begin{eqnarray*}
		\E_{-j} \E_j \int \exp \left \{ \alpha \log \frac{p(\xj \mid \theta_j(\bz),\bX_{-j})}{p(\xj \mid \theta_j^*,\bX_{-j})} + (1-\alpha)d_\alpha(\theta_j(\bz), \theta_j^*) - \log (1/\zeta) \right \} p_{\theta_j} (\theta_j(\bz)) d \theta_j(\bz) \leq \zeta.
	\end{eqnarray*}
	
	Next we use the following result from \cite{yang2020alpha}.If $\mu$ is a probability measure and $h$ is a measurable function such that $e^h \in L_1(\mu)$, then
	\begin{align}
	\log \int e^h d\mu = \underset{\rho \ll \mu}{\sup}\left [\int h d\rho - D_{KL}(\rho \| \mu) \right].
	\label{lemma71}
	\end{align}
	We set $h = \alpha \log \frac{p(\xj \mid \theta_j(\bz),\bX_{-j})}{p(\xj \mid \theta_j^*,\bX_{-j})} + (1-\alpha)d_\alpha(\theta_j(\bz), \theta_j^*) - \log (1/\zeta)$ and $\rho=\hat{q}_{\theta_j}(\theta_j(\bz))$ in the above result where $\hat{q}_{\theta_j}(\theta_j(\bz))$ is the variational estimate of the fractional posterior distribution.
	
	Thus, we get
	\begin{eqnarray*}
		\E_{-j} \E_j \exp \left [  \int \left \{ \alpha \log \frac{p(\xj \mid \theta_j(\bz),\bX_{-j})}{p(\xj \mid \theta_j^*,\bX_{-j})} + (1-\alpha)d_\alpha(\theta_j(\bz), \theta_j^*) - \log (1/\zeta) \right \}\hat{q}_{\theta_j}(\theta_j(\bz)) d\theta_j - D_{KL} (\hat{q}_{\theta_j} \| p_{\theta_j}) \right] \leq \zeta .
	\end{eqnarray*}
	Let
	\begin{equation}
	\tilde{A}(\bX_{-j})=\int \left \{ \alpha \log \frac{p(\xj \mid \theta_j(\bz),\bX_{-j})}{p(\xj \mid \theta_j^*,\bX_{-j})} + (1-\alpha)d_\alpha(\theta_j(\bz), \theta_j^*) - \log \frac{1}{\zeta} \right \}\hat{q}_{\theta_j}(\theta_j(\bz)) d\theta_j(\bz) - D_{KL} (\hat{q}_{\theta_j} \| p_{\theta_j} ),
	\label{ax0}
	\end{equation} and
	\begin{equation}
	A(\bX_{-j})=\left \{x_j :\tilde{A}(\bX_{-j})\leq 0 \right \}.
	\label{ax}
	\end{equation}
	Now we apply Markov's inequality to get a probability statement.
	Thus, the required statement follows from the following:
	\begin{eqnarray*}
		\E_{-j} P(\exp \left(\tilde{A}(\bX_{-j})\right) > 1 \mid \bX_{-j}) &\leq& \E_{-j} \E_j (\exp \left(A(X_{-j})\right) \mid \bX_{-j})=\zeta \\
		\E_{-j} P(A(\bX_{-j}) \mid \bX_{-j}) &\geq& 1- \zeta.
	\end{eqnarray*}

\end{proof}

\begin{lem} Under the assumptions of Theorem~\ref{thm2}, let the precision matrix of the $p$-variate data generating distribution be $s^*$-sparse and have eigenvalues bounded away from $0$ and $\infty$.  For any $\zeta \in (0,1)$ and $n \geq a_1 s^* \log p$, and any measure $q_{\theta_j} \in \Gamma$ such that $q_{\theta_j} \ll p_{\theta_j}$, we have  
	\begin{equation*}
	P \left (\int \frac{1}{n}d_\alpha (\theta_j(\bz), \theta_j^*)\hat{q}_{\theta_j}(\theta_j(\bz)) d\theta_j(\bz) \leq \frac{\alpha}{n(1-\alpha)}  \Psi(q_{\theta_j}) + \frac{1}{n(1-\alpha)}\log (1/\zeta) \right )
	\geq 1 - \zeta - \frac{c_2}{p^{c_1+1}} - \exp\{-a_2n\}.
	\end{equation*} for some positive constants $a_1$, $D$ and $a_2$ , $c_1,c_2$.
	\label{lema2}
\end{lem}

\begin{proof}
	Similar with Lemma~\ref{lemC2}, we need to provide upper bound for  $$2\|\mathrm{W}_{l}^{1/2}\bX_{-j}\tilde{\beta}_j(\bz) -\mathrm{W}_{l}^{1/2}\bX_{-j}{\beta}_j(\bz) \|^2 + \|\mathrm{W}_{l}^{1/2}\bX_{-j}\tilde{\beta}_j(\bz) -\mathrm{W}_{l}^{1/2}\bX_{-j}{\beta^*_j}(\bz_l) \|^2 +  2 \langle \mathrm{W}_{l}\bX_{-j}(\tilde{\beta}_j(\bz)-\beta_j(\bz)),\varepsilon_j \rangle$$
	
	By the similar argument with Lemma 3 in \cite{atchade2019quasi}, the following restricted eigenvalue conditions hold: $\klt (2s^*,\bX_{-j}^T\bX_{-j}/n)$ and $\tilde{\kappa}(2s^*,\bX_{-j}^T\bX_{-j}/n)$ are constants.
	Therefore, with probability at least $1- \exp(-a_1n)$, we have $\bX \notin \mathcal{G}_{n,p}$, this gives us 
	\begin{equation*}
	\|\mathrm{W}_{l}^{1/2}\bX_{-j}\tilde{\beta}_j(\bz) -\mathrm{W}_{l}^{1/2}\bX_{-j}{\beta^*_j}(\bz_l) \|^2 \leq cn\|\mathrm W\|_2\|\tilde{\beta}_j(\bz)-{\beta^*_j}(\bz_l)\|_2^2  \leq c n\tau^{-1} n^{-1} s_j^*  \leq  c n\epsilon^2,
	\end{equation*}
	given that $\tau^{-1}\leq c\sqrt{\log n}$.
	
	In addition, by definition of $\tilde q_{{\theta}_j}(\theta_j(\bz))$, similarly we have
	\begin{equation*}
	2\|\mathrm{W}_{l}^{1/2}\bX_{-j}\tilde{\beta}_j(\bz) -\mathrm{W}_{l}^{1/2}\bX_{-j}{\beta}_j(\bz) \|^2 \leq cn \tau^{-1} \|\tilde{\beta}_j(\bz)-{\beta^*_j}(\bz_l)\|_2^2 \leq  cn \epsilon^2.
	\end{equation*}  
	For the last term $2 \langle \mathrm{W}_{l}\bX_{-j}(\tilde{\beta}_j(\bz)-\beta_j(\bz)),\varepsilon_j \rangle$, first we have
	\begin{equation*}
	2 \langle \mathrm{W}_{l}\bX_{-j}(\tilde{\beta}_j(\bz)-\beta_j(\bz)),\varepsilon_j \rangle \leq 2 \|\tilde{\beta}_j(\bz)-\beta_j(\bz)\|_1 \|\bX^T_{-j} \mathrm{W}_{l} \varepsilon_j \|_\infty \leq c \sqrt{s_j^*} \tau \epsilon \|\bX^T_{-j} \mathrm{W}_{l} \varepsilon_j \|_\infty.
	\end{equation*}
	
	Note that $\bX^T_{-j} \mathrm{W}_{l} \varepsilon_j$ is a $p-1$ dimensional Gaussian vector, with and with probability great than $1-\exp(-a_2 n)$, the scale of each component of the Gaussian vector is  bounded by $ \sigma^*\sqrt{n}/\tau$ multiplied constant, by maximal inequality of Gaussian random vector, we have
	\begin{equation*}
	P( \|\bX^T_{-j} \mathrm{W}_{l} \varepsilon_j \|_\infty \geq t) \leq e^{-\frac{\tau^2 t^2}{2 \sigma^{*2} n}+\log p},
	\end{equation*} 
	and we can choose $t=  \sqrt{  (c_0+2) n \log p}/\tau$ for a constant $c_0>0$, then the probability upper bound becomes $p^{-(c_0+1)}$. Therefore, we have 
	\begin{equation*}
	2 \langle \mathrm{W}_{l}\bX_{-j}(\tilde{\beta}_j(\bz)-\beta_j(\bz)),\varepsilon_j \rangle \leq c \sqrt{s_j^*} \tau \epsilon \sqrt{n \log p} /\tau \leq c n\epsilon^2 + s_j^* \log p \leq c n\epsilon^2,
	\end{equation*}
	where the in the second inequality we use $2ab \leq a^2 +b^2$. 
\end{proof}

\subsection{Steps of the algorithm}\label{ssec:algo}
The following algorithm provides the steps for covariate-dependent graph estimation of $\bX\in\mathbb R^{n\times p}$.

We select the bandwidth hyperparameter $\tau\in\mathbb R^n$ using a 2-step approach for density estimation discussed in \cite{dasgupta2020two,abramson1982bandwidth,van2003adaptive}. Under this approach, bandwidths are initialized using Silverman’s rule of thumb, and the density is subsequently refined by updating the bandwidth values. We follow this methodology to estimate the density of $\bz$, and use the updated bandwidths from the second step for $\tau$. 

We next fix $x_j$ as the response and consider the task of performing $n$ weighted spike-and-slab regressions with $\bX_{-j}$ as predictors using the weights calculated using the bandwidth $\tau$ and the covariates. Each of these regressions requires the specification of three hyperparameters: $\pi,\sigma^2$, and $\sigma^2_\theta$. To select the hyperparameters, we use a hybrid of model averaging and grid search. We first generate candidate grids of $\pi, \sigma^2$, and $\sigma^2_\theta$ values. We denote the grid of $\pi$ candidates by $\Theta_\pi$, and the Cartesian product between the grid of $\sigma^2$ and $\sigma^2_\theta$ candidates as $\Theta$.

Next, for each $\pi\in\Theta_\pi$, we fit a spike-and-slab regression weighted with respect to individual $l$ for each $(\sigma^2,\sigma^2_\theta)\in\Theta,l\in1,...,n$. We make a global selection of $\sigma^2$ and $\sigma^2_\theta$ for each of the $\pi\in\Theta_\pi$ such that the sum of the ELBO across all $n$ weighted regressions is maximized. This grid search produces $\lvert\Theta_\pi\rvert$ models per individual. For each of these models, we calculate a model averaging weight by taking the softmax over the ELBOs and use these to average over the variational approximations to the posterior quantities to construct the final model. Finally, to obtain the graph estimate, we symmetrize the posterior inclusion probabilities from the final model and threshold them at $0.5$.

\subsection{Discrete Covariate Simulation Study}\label{app:disc}

For the discrete covariate, we perform experiments in which we vary the data dimensionality, the distribution of the covariate levels, and the strength of the signal in the ground-truth precision structures. As with the continuous covariate, we perform 50 trials per experiment. We compare the performance of W-PL to mgm \citep{haslbeck_mgm_2020}, as well as to the method of \cite{carbonetto2012scalable} applied in a pseudo-likelihood fashion (CS). That is, we fix each variable as the response in turn and perform a variational spike-and-slab regression. We obtain the final graphs using the same symmetrization and thresholding scheme as used for W-PL. To incorporate $\bz$, we apply this estimation procedure for each of the covariate levels independently. Because no information is shared between levels, this allows us to evaluate the impact of the weighting scheme in W-PL. We use the implementation of the variational spike-and-slab from \cite{carbonetto_varbvs_2017}, which employs a hybrid hyperparameter specification scheme, wherein the $\pi$ candidates are averaged and $\sigma^2$ and $\sigma^2_\theta$ are selected via Empirical Bayes for each of the $\pi$ candidates.

In each of the experiments, we assign individuals $1,...,n_1$ to the first level of a binary discrete covariate $\bz_i=1$, and the remaining individuals to the second level $\bz_i=2$. We refer to the structure of the sample as balanced when $n_1=n-n_1:=n_2$, and unbalanced when $n_1\neq n_2$. 

\subsubsection{Covariate-Independent Setting}\label{sec:covindep}

We first consider a covariate-independent setting where the ground truth dependence structure is independent of $\bz_i$ and set $n=100, p=10$. To construct the precision matrices, we first define
\begin{eqnarray*}
	\lambda_{\bz_i}&=& {[c{\bf 1}_4~~ , ~~{\bf 0}_{p-3}]}^{\mathrm{T}}, \text{ for both } \bz_i=1,2,
\end{eqnarray*}
where ${\bf 1}_{4}$ is a four-dimensional vector of ones, and ${\bf 0}_{p-3}$ is a $(p-3)$-dimensional vector of zeroes. We refer to $c=15$ as high signal, and $c=3$ as reduced signal. We next define the precision matrix for the $i$-th individual as $\Omega_i=(\lambda_{\bz_i} {\lambda_{\bz_i}}^{\mathrm{T}} + 10 \mathbb{I}_{p+1})$. The corresponding dependence structure we aim to estimate is
\[
\mathrm{G}^*=  \begin{bmatrix} 
\mathbb{J}_4 - \mathbb{I}_4 & {\bf 0}_{4,p-3}  \\
{\bf 0}_{p-3,4}  &{\bf 0}_{p-3,p-3}
\end{bmatrix}
\]
where $\mathbb J_k$ is a $k\times k$ matrix where all entries are $1$.

We perform experiments in the covariate-independent setting on high signal with balanced structure, reduced signal with balanced structure, and high signal with unbalanced structure ($n_1 = 80,n_2=20$). We present results for each of the experiments in Table \ref{tab:cov_indep}. When the signal strength is high and the sample is balanced, all three methods correctly detect all of the edges in the ground truth structure. However, the performance of both competitors suffers under the unbalanced sample structure, particularly for CS, while W-PL correctly detects all edges. In the reduced signal setting, the differential between W-PL and the competitors grows significantly. 

\begin{table}[]
	\centering
	\begin{tabular}[t]{llllll}
		\toprule
		$c$ & $n_1$ & $n_2$ & Method & Sensitivity$(\uparrow)$ & Specificity$(\uparrow)$\\
		\midrule
		&  &  & W-PL & $\mathbf{0.8517} (0.1591)$ & $0.9843 (0.0168)$\\
		
		&  &  & mgm & $0.1417 (0.1229)$ & $\mathbf{0.9988} (0.0033)$\\
		
		\multirow{-3}{*}{\raggedright\arraybackslash 3} & \multirow{-3}{*}{\raggedright\arraybackslash 50} & \multirow{-3}{*}{\raggedright\arraybackslash 50} & CS & $0.2950 (0.1489)$ & $0.9929 (0.0081)$\\
		\cmidrule{1-6}
		&  &  & W-PL & $\mathbf{1.0000} (0.0000)$ & $0.9941 (0.0111)$\\
		
		&  &  & mgm & $\mathbf{1.0000} (0.0000)$ & $\mathbf{0.9990} (0.0031)$\\
		
		\multirow{-3}{*}{\raggedright\arraybackslash 15} & \multirow{-3}{*}{\raggedright\arraybackslash 50} & \multirow{-3}{*}{\raggedright\arraybackslash 50} & CS & $\mathbf{1.0000} (0.0000)$ & $0.9955 (0.0080)$\\
		\cmidrule{1-6}
		&  &  & W-PL & $\mathbf{1.0000} (0.0000)$ & $\mathbf{1.0000} (0.0000)$\\
		
		&  &  & mgm & $0.9013 (0.0602)$ & $0.9992 (0.0033)$\\
		
		\multirow{-3}{*}{\raggedright\arraybackslash 15} & \multirow{-3}{*}{\raggedright\arraybackslash 80} & \multirow{-3}{*}{\raggedright\arraybackslash 20} & CS & $0.8147 (0.0215)$ & $0.9980 (0.0052)$\\
		\bottomrule
	\end{tabular}
	\caption{\it Results for the covariate-independent setting, presented as $\textit{mean} (\textit{standard deviation})$}
	\label{tab:cov_indep}
\end{table}

\subsubsection{Covariate-Free Setting}

We next examine a setting identical to the covariate-independent one, again with $n=100,p=10$, however, this time, assume that no information on the covariates is available. In the absence of covariate information, W-PL selects all weights to be equal to one. Thus, the graph estimates are identical for all the individuals in this setting, akin to the usual graph selection algorithms. Because mgm requires the timepoints to be specified, we omit it from these experiments. 

We present the results for experiments in the covariate-free setting with high and low signal in Table \ref{tab:cov_free}. Unsurprisingly, results for both methods are similar. The minor differences in performance may be attributed to the differing hyperparameter specification schemes.  

\begin{table}[]
	\centering
	\begin{tabular}[t]{llll}
		\toprule
		$c$ & Method & Sensitivity$(\uparrow)$ & Specificity$(\uparrow)$\\
		\midrule
		& W-PL & $\mathbf{0.9533} (0.1168)$ & $\mathbf{0.9996} (0.0029)$\\
		
		\multirow{-2}{*}{\raggedright\arraybackslash 3} & CS & $0.9233 (0.1313)$ & $0.9939 (0.0103)$\\
		\cmidrule{1-4}
		& W-PL & $\mathbf{1.0000} (0.0000)$ & $\mathbf{1.0000} (0.0000)$\\
		
		\multirow{-2}{*}{\raggedright\arraybackslash 15} & CS & $\mathbf{1.0000} (0.0000)$ & $0.9963 (0.0079)$\\
		\bottomrule
	\end{tabular}
	\caption{\it Results for covariate-free setting}
	\label{tab:cov_free}
\end{table}

\subsubsection{Covariate-Dependent Setting}\label{sec:cov_dep}

We next consider the setting in which the precision matrix varies with the covariate level. We define the relationship as 
\begin{eqnarray*}
	\lambda_{\bz_i}&=& {[c{\bf 1}_4~~ , ~~{\bf 0}_{p-3}]}^{\mathrm{T}}, \text{ if } \bz_i=1,\text { and}\\
	\lambda_{\bz_i} &=& {[{\bf 0}_{p-3}~~ , ~~c{\bf 1}_{4}]}^{\mathrm{T}}, \text{ if } \bz_i=2.
\end{eqnarray*}
As before, we define the precision matrices as $\Omega_i=(\lambda_{\bz_i} {\lambda_{\bz_i}}^{\mathrm{T}} + 10 \mathbb{I}_{p+1}) $, and thus, the true graph structure $\mathrm{G}^{*}(\bz)$ for an individual with covariate value $\bz$  is
\[
\mathrm{G}^{*}(1)=  \begin{bmatrix} 
\mathbb{J}_4 -  \mathbb{I}_4 & {\bf 0}_{4,p-3}  \\
{\bf 0}_{p-3,4}  &{\bf 0}_{p-3,p-3}
\end{bmatrix}
, \quad
\mathrm{G}^{*}(2)=  \begin{bmatrix} 
{\bf 0}_{p-3,p-3} & {\bf 0}_{p-3,4} \\
{\bf 0}_{4,p-3} &\mathbb{J}_4 -\mathbb{I}_4
\end{bmatrix}.
\]
We visualize these precision matrices and the corresponding dependence structures for $p=10$ in Figure \ref{fig:precmats}.
\begin{figure}[htbp]
	\begin{center}
		\begin{tabular}{ccc}
			\includegraphics[width=0.23\linewidth]{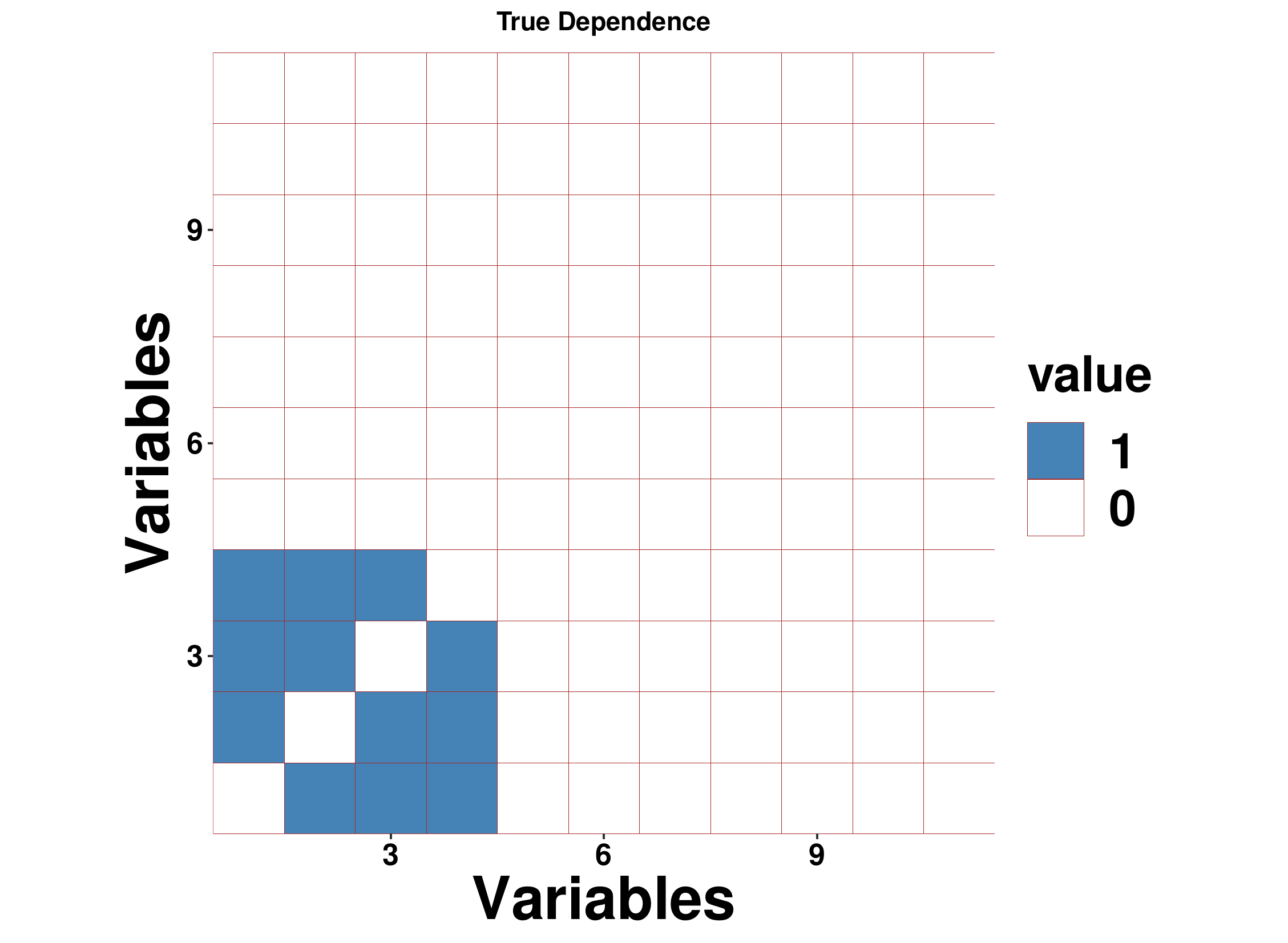} &
			\includegraphics[width=0.23\linewidth]{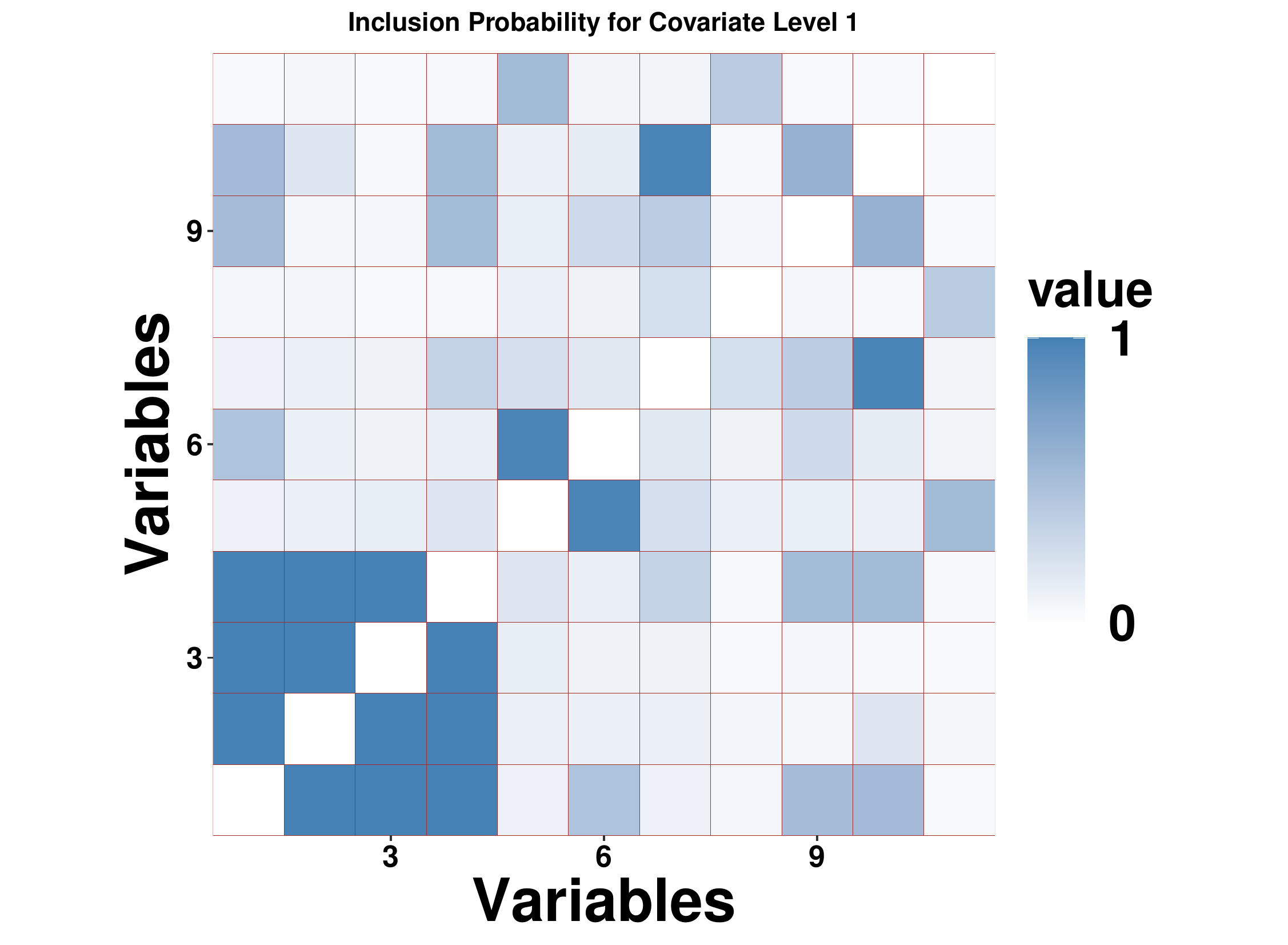}&
			\includegraphics[width=0.23\linewidth]{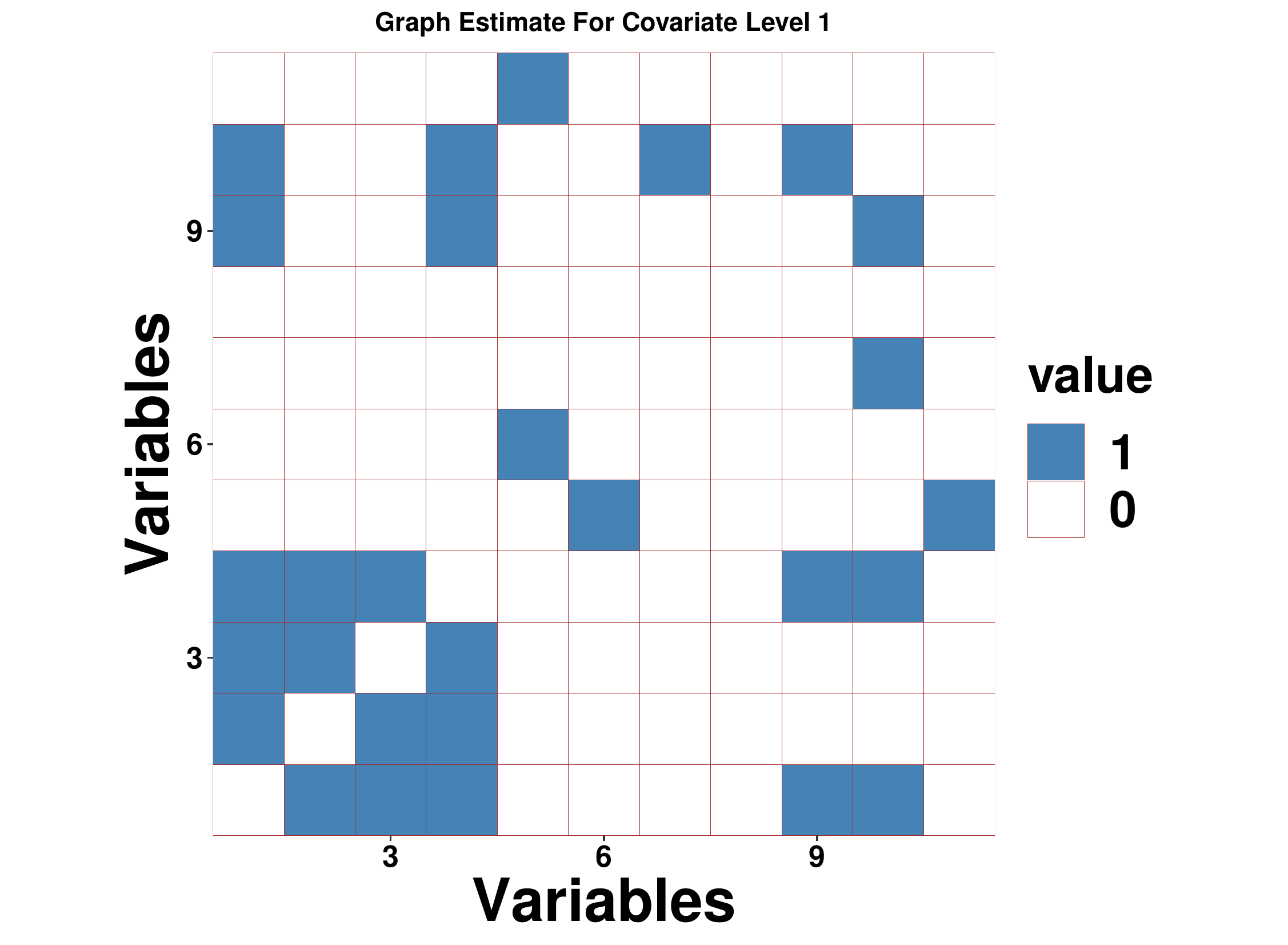}\\
			\includegraphics[width=0.23\linewidth]{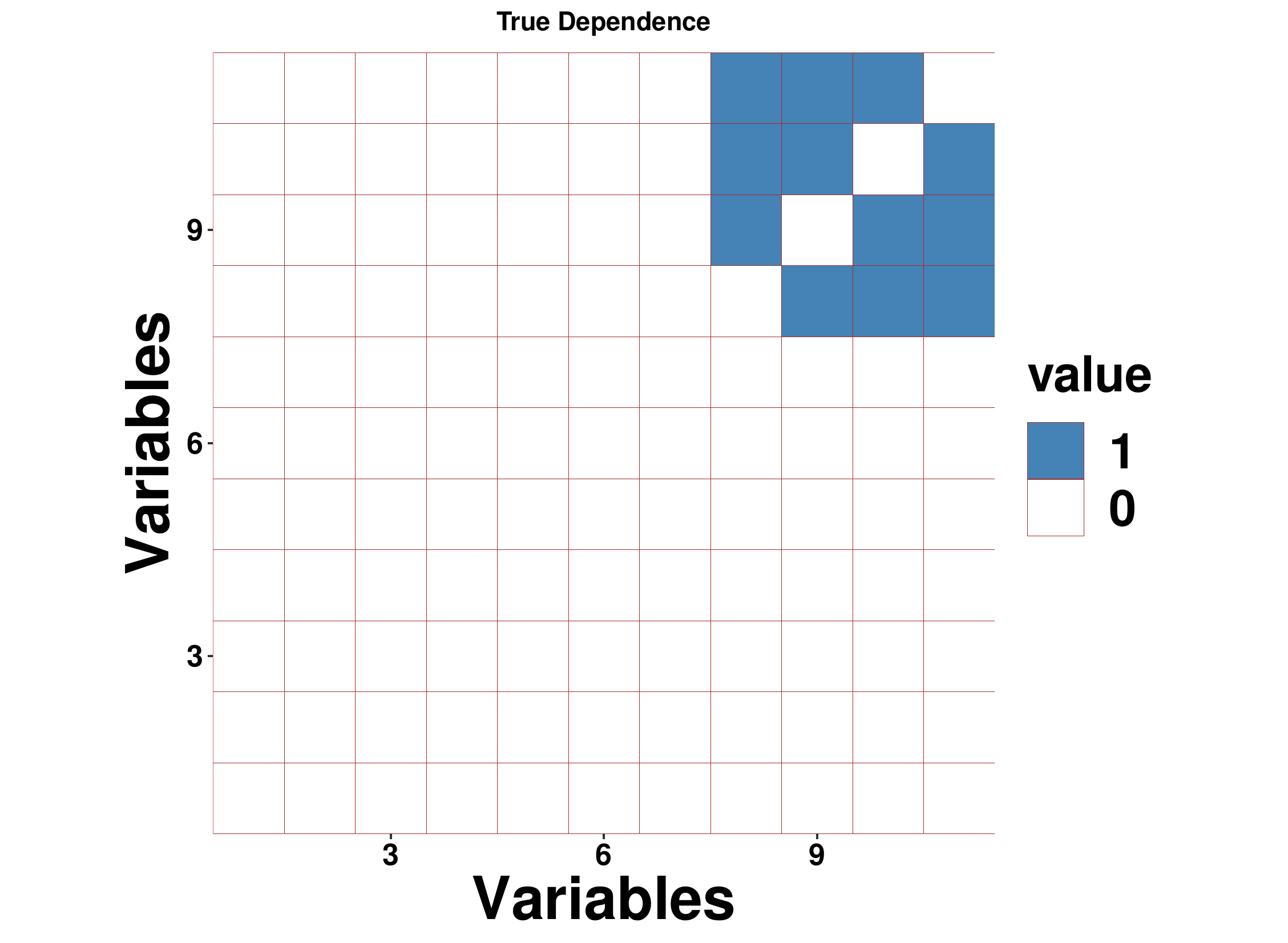} &
			\includegraphics[width=0.23\linewidth]{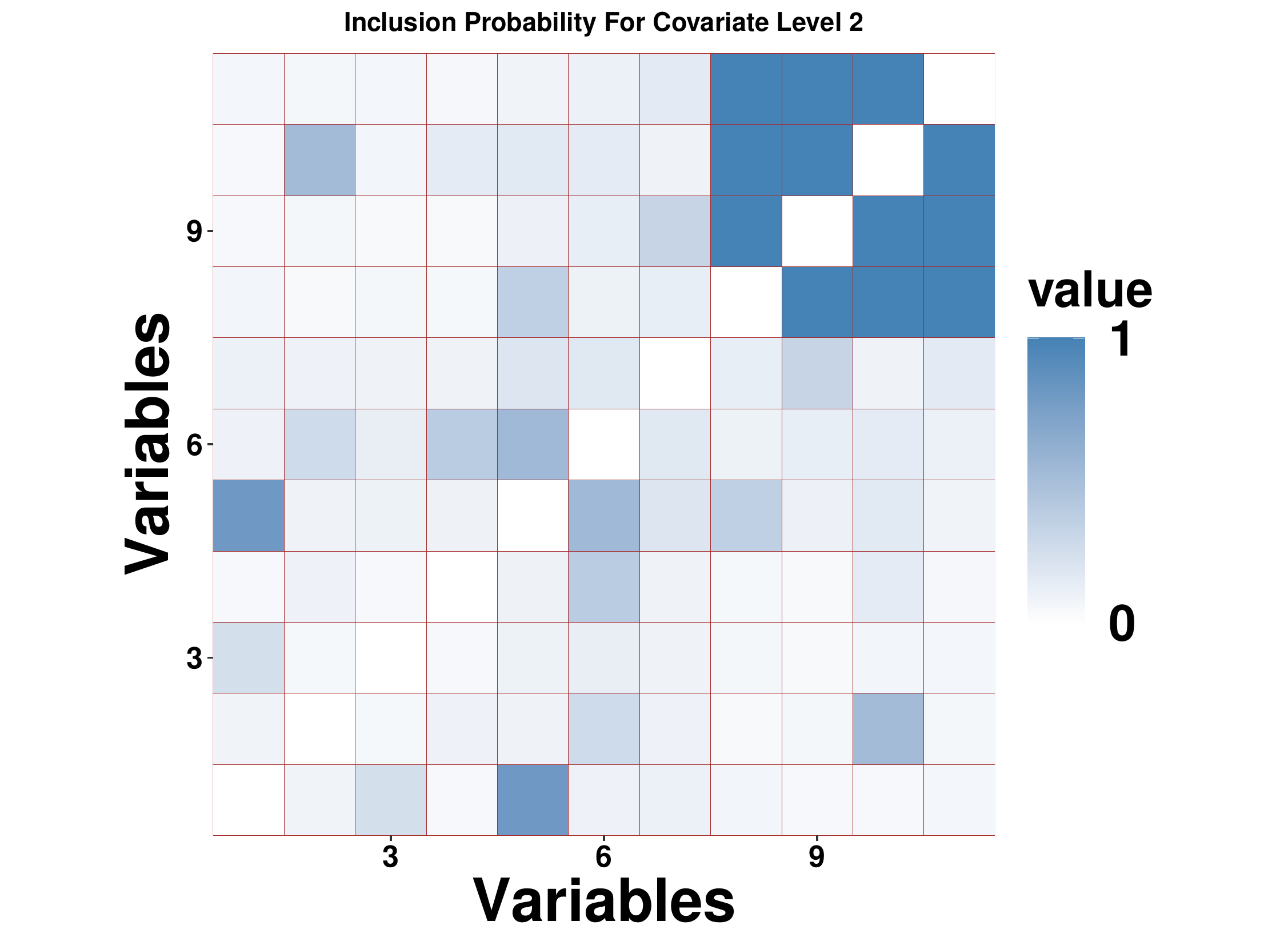} &
			\includegraphics[width=0.23\linewidth]{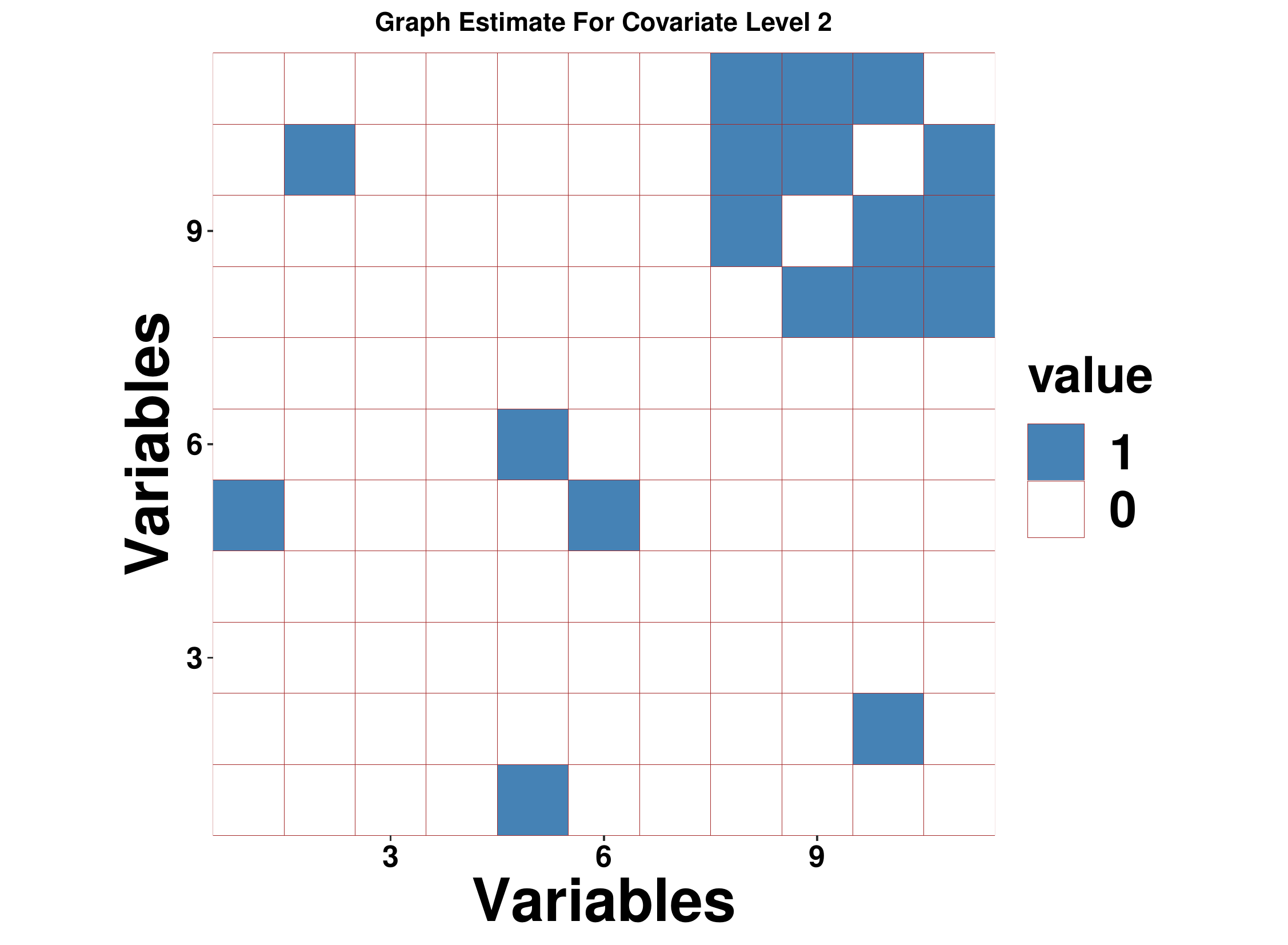}\\
		\end{tabular}
		\caption{\it Left to right: True dependence structures, estimated inclusion probabilities and estimated graphs for a sample simulation with two discrete covariate levels with p=10.}
		\label{fig:precmats}
	\end{center}
\end{figure}

In addition to varying signal strength and sample structure for $p=10$, we additionally vary the dimension of the data to $p=30$ and $p=50$ with high signal strength and balanced samples. In all experiments, we fix $n=100$. We present results for these experiments in Table \ref{tab:cov_dep}. 

While the performance of W-PL and mgm are similar for $p=10$, as $p$ increases, the performance of mgm deteriorates. On the other hand, W-PL and CS demonstrate robustness to the increased sample size. As in the covariate-independent setting, the performance of both mgm and CS is significantly harmed relative to W-PL when faced with reduced signal.   

\begin{table}[]
	\centering
	\begin{tabular}[t]{lllllll}
		\toprule
		$p$ & $c$ & $n_1$ & $n_2$ & Method & Sensitivity$(\uparrow)$ & Specificity$(\uparrow)$\\
		\midrule
		&  &  &  & W-PL & $\mathbf{0.5800} (0.1859)$ & $0.9900 (0.0133)$\\
		
		&  &  &  & mgm & $0.0867 (0.0937)$ & $\mathbf{0.9980} (0.0046)$\\
		
		\multirow{-3}{*}{\raggedright\arraybackslash 10} & \multirow{-3}{*}{\raggedright\arraybackslash 3} & \multirow{-3}{*}{\raggedright\arraybackslash 50} & \multirow{-3}{*}{\raggedright\arraybackslash 50} & CS & $0.2950 (0.1545)$ & $0.9927 (0.0090)$\\
		\cmidrule{1-7}
		&  &  &  & W-PL & $\mathbf{1.0000} (0.0000)$ & $0.9849 (0.0187)$\\
		
		&  &  &  & mgm & $0.9950 (0.0354)$ & $\mathbf{0.9982} (0.0053)$\\
		
		\multirow{-3}{*}{\raggedright\arraybackslash 10} & \multirow{-3}{*}{\raggedright\arraybackslash 15} & \multirow{-3}{*}{\raggedright\arraybackslash 50} & \multirow{-3}{*}{\raggedright\arraybackslash 50} & CS & $\mathbf{1.0000} (0.0000)$ & $0.9953 (0.0072)$\\
		\cmidrule{1-7}
		&  &  &  & W-PL & $0.7973 (0.0189)$ & $0.9867 (0.0106)$\\
		
		&  &  &  & mgm & $\mathbf{0.8127} (0.0212)$ & $\mathbf{0.9990} (0.0028)$\\
		
		\multirow{-3}{*}{\raggedright\arraybackslash 10} & \multirow{-3}{*}{\raggedright\arraybackslash 15} & \multirow{-3}{*}{\raggedright\arraybackslash 80} & \multirow{-3}{*}{\raggedright\arraybackslash 20} & CS & $0.8093 (0.0166)$ & $0.9983 (0.0046)$\\
		\cmidrule{1-7}
		&  &  &  & W-PL & $\mathbf{1.0000} (0.0000)$ & $0.9926 (0.0036)$\\
		
		&  &  &  & mgm & $0.7567 (0.1812)$ & $\mathbf{0.9997} (0.0006)$\\
		
		\multirow{-3}{*}{\raggedright\arraybackslash 30} & \multirow{-3}{*}{\raggedright\arraybackslash 15} & \multirow{-3}{*}{\raggedright\arraybackslash 50} & \multirow{-3}{*}{\raggedright\arraybackslash 50} & CS & $\mathbf{1.0000} (0.0000)$ & $0.9976 (0.0018)$\\
		\cmidrule{1-7}
		&  &  &  & W-PL & $0.9867 (0.0425)$ & $0.9958 (0.0016)$\\
		
		&  &  &  & mgm & $0.4550 (0.2022)$ & $\mathbf{0.9999} (0.0002)$\\
		
		\multirow{-3}{*}{\raggedright\arraybackslash 50} & \multirow{-3}{*}{\raggedright\arraybackslash 15} & \multirow{-3}{*}{\raggedright\arraybackslash 50} & \multirow{-3}{*}{\raggedright\arraybackslash 50} & CS & $\mathbf{0.9983} (0.0118)$ & $0.9982 (0.0009)$\\
		\bottomrule
	\end{tabular}
	\caption{\it Results for discrete covariate-dependent setting}
	\label{tab:cov_dep}
\end{table}

\subsubsection{High-Dimensional Setting}\label{hidim}

Our last series of experiments with the discrete covariate deals with a challenging high-dimensional setting where $p\geq n$. To handle the increased dimensionality, we found it necessary to modify our hyperparameter specification scheme for this experiment. We use  \cite{carbonetto2012scalable} to obtain an Empirical Bayes estimate to the hyperparameter $\sigma^2$, and use a grid search to optimize over the hyperparameters $\pi$ and $\sigma_{\theta}^2$ using the ELBO as our objective function. For the bandwidth hyperparameter, we consider an ad-hoc choice of $\tau=0.1$. We only perform 20 trials per experiment in this setting.

We maintain the relationship between the covariates and the ground truth structure as in Section \ref{sec:cov_dep} and first consider an unbalanced setting with $p=100, n_1=40, n_2=10$ and high signal. We present results from this setting in Figure \ref{fig:n50p100}, and exclude mgm from this experiment due to its deteriorating performance with large $p$ and high time-complexity. Note that as only $10$ observations belong to level 2, separate estimation through CS suffers significantly compared to W-PL when estimating the graph for level 2. 

\begin{figure}[htbp]
	\centering
	\includegraphics[width=0.8\linewidth]{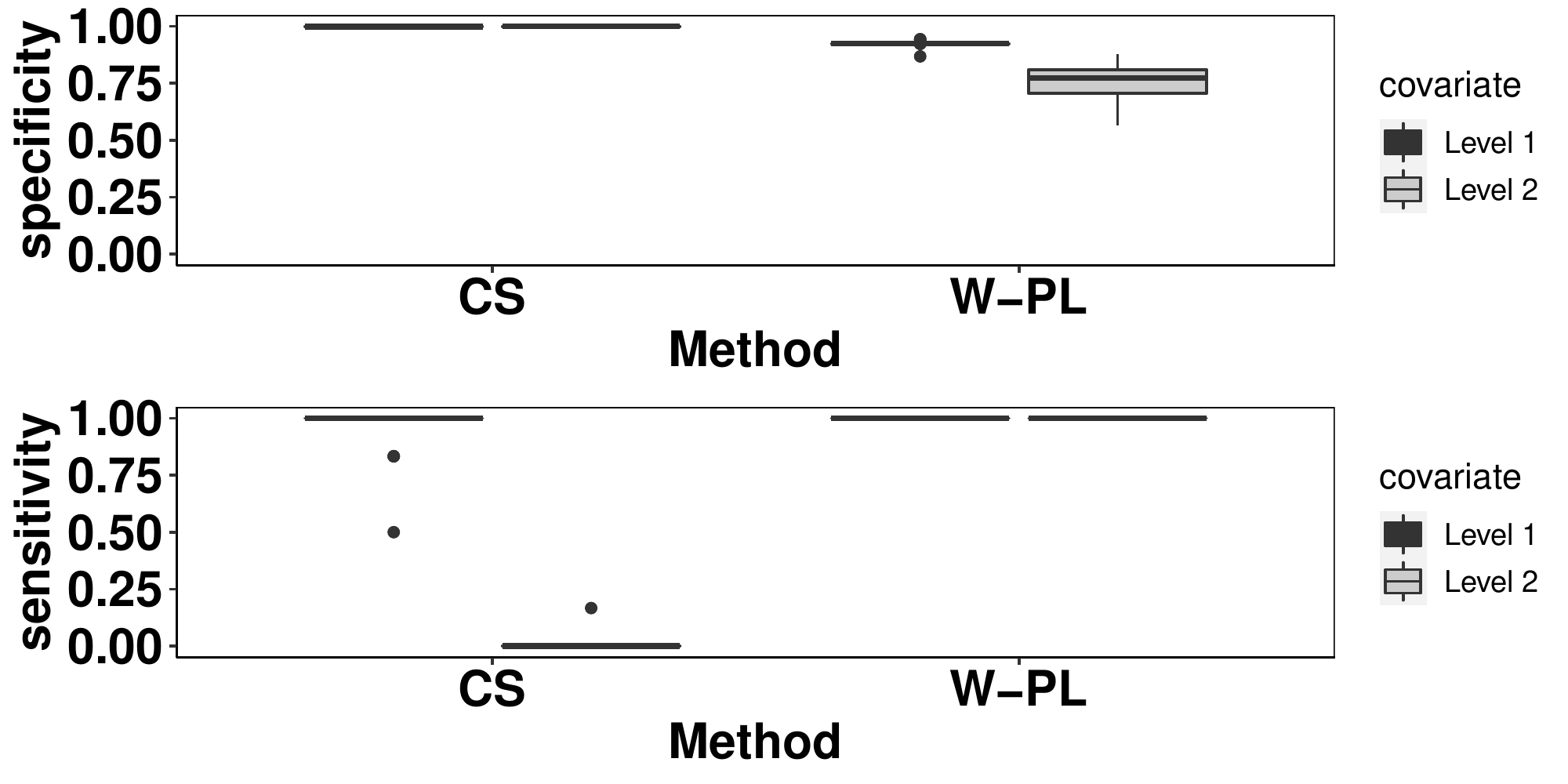}
	\caption{\it Results for covariate-dependent setting with $n=50,p=100$. Top row: Specificity CS vs W-PL; Bottom row: sensitivity CS vs W-PL}
	\label{fig:n50p100}
\end{figure}

Next, to demonstrate how the signal-to-noise ratio (SNR) influences the performance of our approach, we study several further experiments in the high-dimensional setting with $n=50,p=50$ and $n=50,p=100$, $n_1=20,n_2=30$, $\lambda=(c {\bf 1}_4,{\bf 0}_{p-3} )^{\mathrm{T}}$ keeping other settings the same. Note that the SNR is controlled by $c$. 

To assess performance under sparsity and with weak signal strength, we analyze the area under the receiver operating characteristic curve (AUC). By varying the threshold for a posterior inclusion probability to indicate an edge, we obtain a sequence of true and false positive ratios that we may use to calculate the corresponding AUC. AUC can also be defined by the fraction of pairs that the prediction ordered correctly: let $y_1,...,y_n$ be the $0$ and $1$ responses and $p_1,...,p_n$ be the corresponding predicted probabilities. The AUC can then be calculated as
$${\sum_{i=1}^n \sum_{j=1}^n \ind{\{y_i<y_j\}} \ind{\{p_i<p_j\}}}/\sum_{i=1}^n \sum_{j=1}^n \ind_{\{y_i<y_j\}}$$ 
We present results from these experiments in Figure~\ref{fig:auc_snr}. As the signal strength (i.e., $c$) increases, the AUC also increases from around $0.5$ to $1$. When the SNR is low, W-PL does not work well and produces an AUC close to $0.5$, which essentially is a random guess as to the presence of an edge. However, when there are sufficient observations and the SNR is high, the AUC for  level 2 exceeds $0.9$. 

Because of the low level of observations in the first level of the covariate ($n_1=20$), separate estimation with CS does not perform well. W-PL consistently outperforms CS for both level 1 and level 2.  

\begin{figure}
	\centering
	\begin{subfigure}{0.49\textwidth}
		\includegraphics[width=\textwidth]{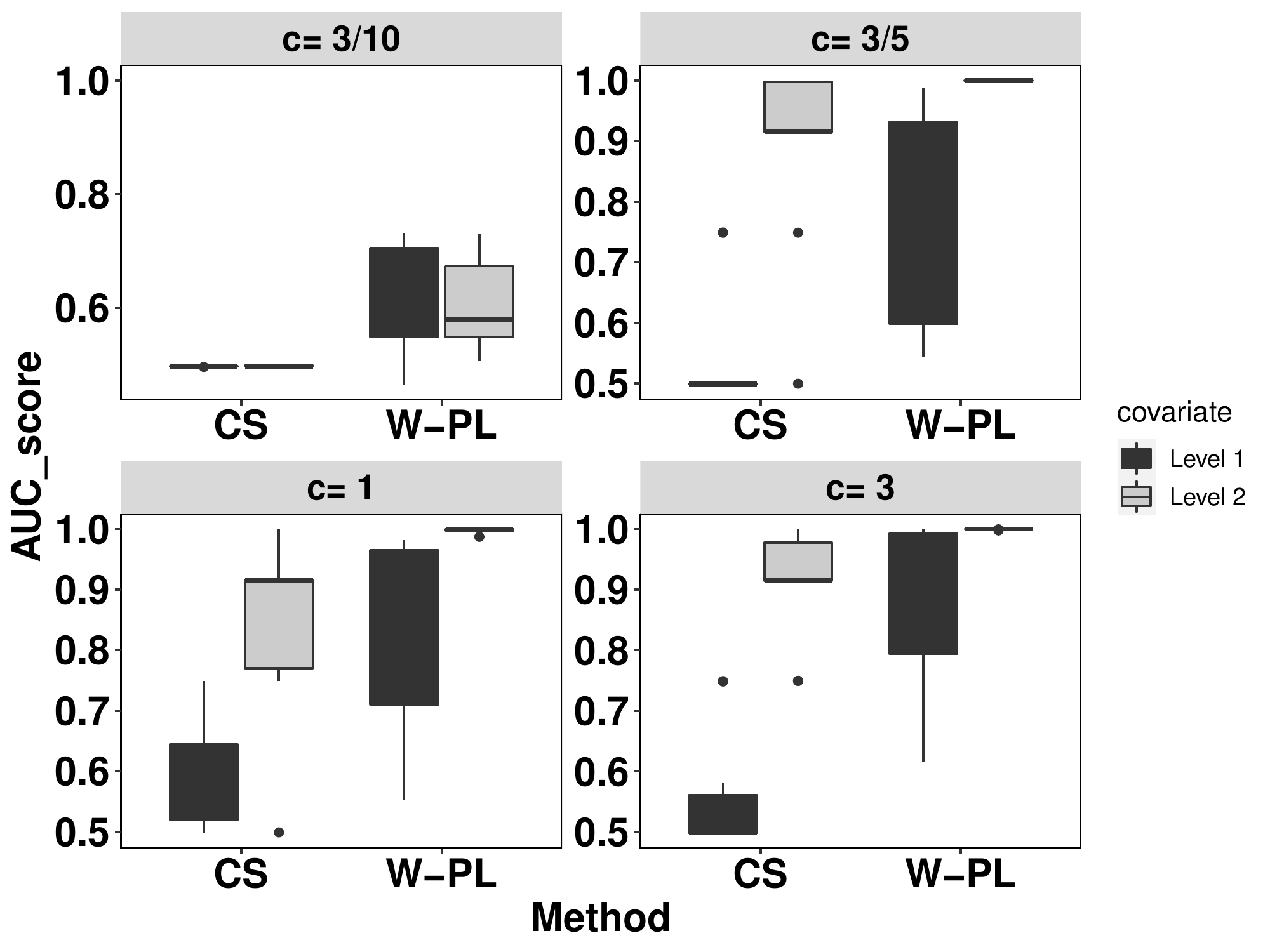}
		\caption{$n=50$, $p=50$}
	\end{subfigure}
	\begin{subfigure}{0.49\textwidth}
		\includegraphics[width=\textwidth]{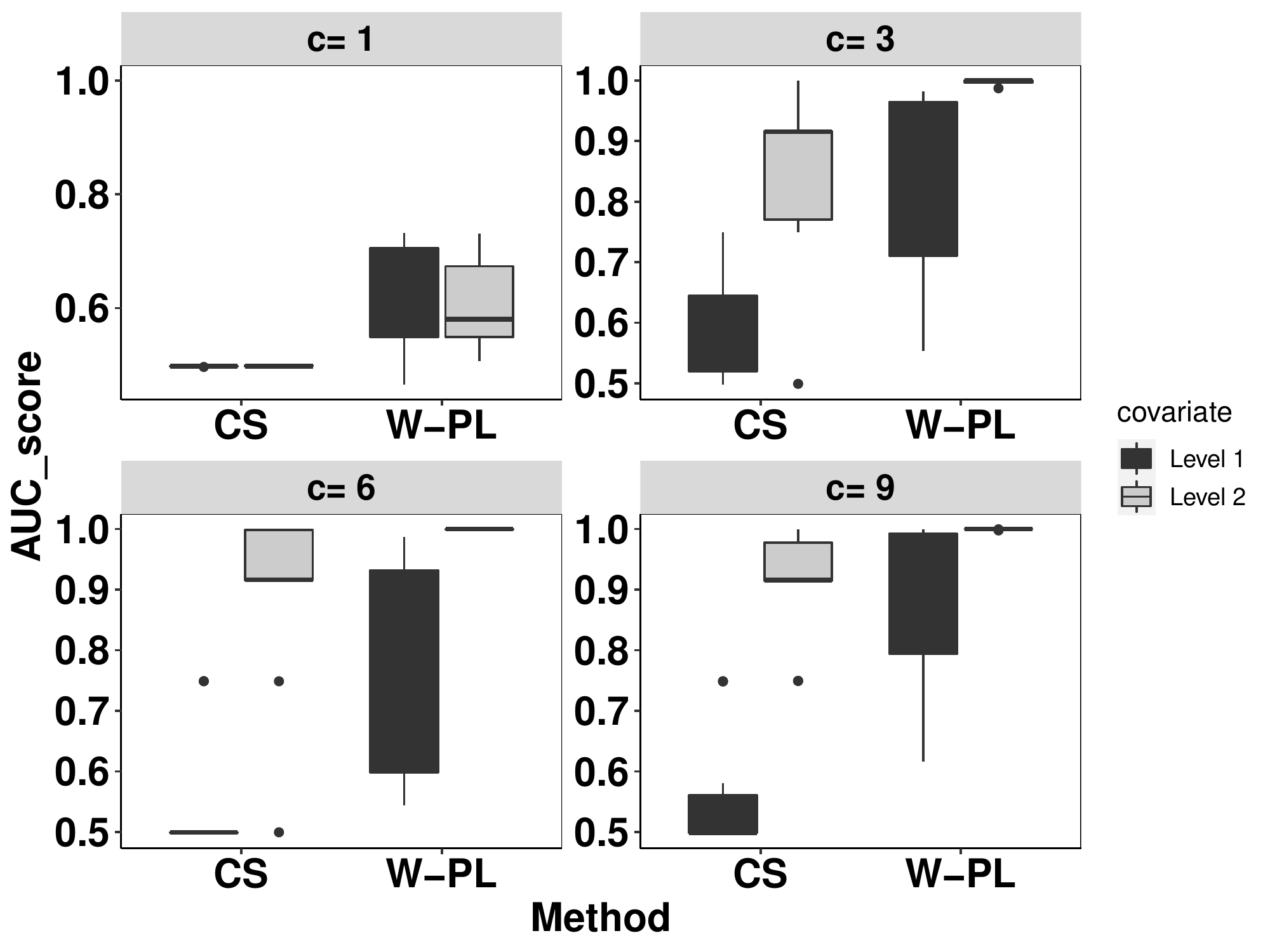}
		\caption{$n=50$, $p=100$}
	\end{subfigure}
	\caption{\it CS versus W-PL method, measured by the AUC scores between the true graph structure and the calculated inclusion probabilities.}
	\label{fig:auc_snr}
\end{figure}

\subsection{Departure from Gaussian assumption}\label{ssec:dGauss}
The method theoretically is built on the assumption that the true data generation is Gaussian, while the pseudo-likelihood approach is used mostly as a tool for estimation. To study the effects of departures from Gaussianity, we have investigated two scenarios. Firstly, we consider the situation where the data is contaminated, that is,  the data comes from a Gaussian  distribution with an independent structure, $c\%$ of which is contaminated by data coming from an unrelated independent Gaussian distribution. Figure \ref{fig1111} shows the results with $5\%$ contamination. The results, however, get worse as the amount of contamination increases. Secondly, we consider the $t$-distribution with varying degrees of freedom. Figure \ref{fig3612} shows the sensitivity and specificity for varying degrees of freedom. The results indicate that for degrees of freedom greater than $6$, the results are stable, and naturally shows improvement as the degrees of freedom increases. However, for degrees of freedom less than 6, the performance suffers, as shown in the left panel.

\begin{figure}[htbp]
	\begin{center}
		\begin{tabular}{ccc}
			\includegraphics[width=0.3\linewidth]{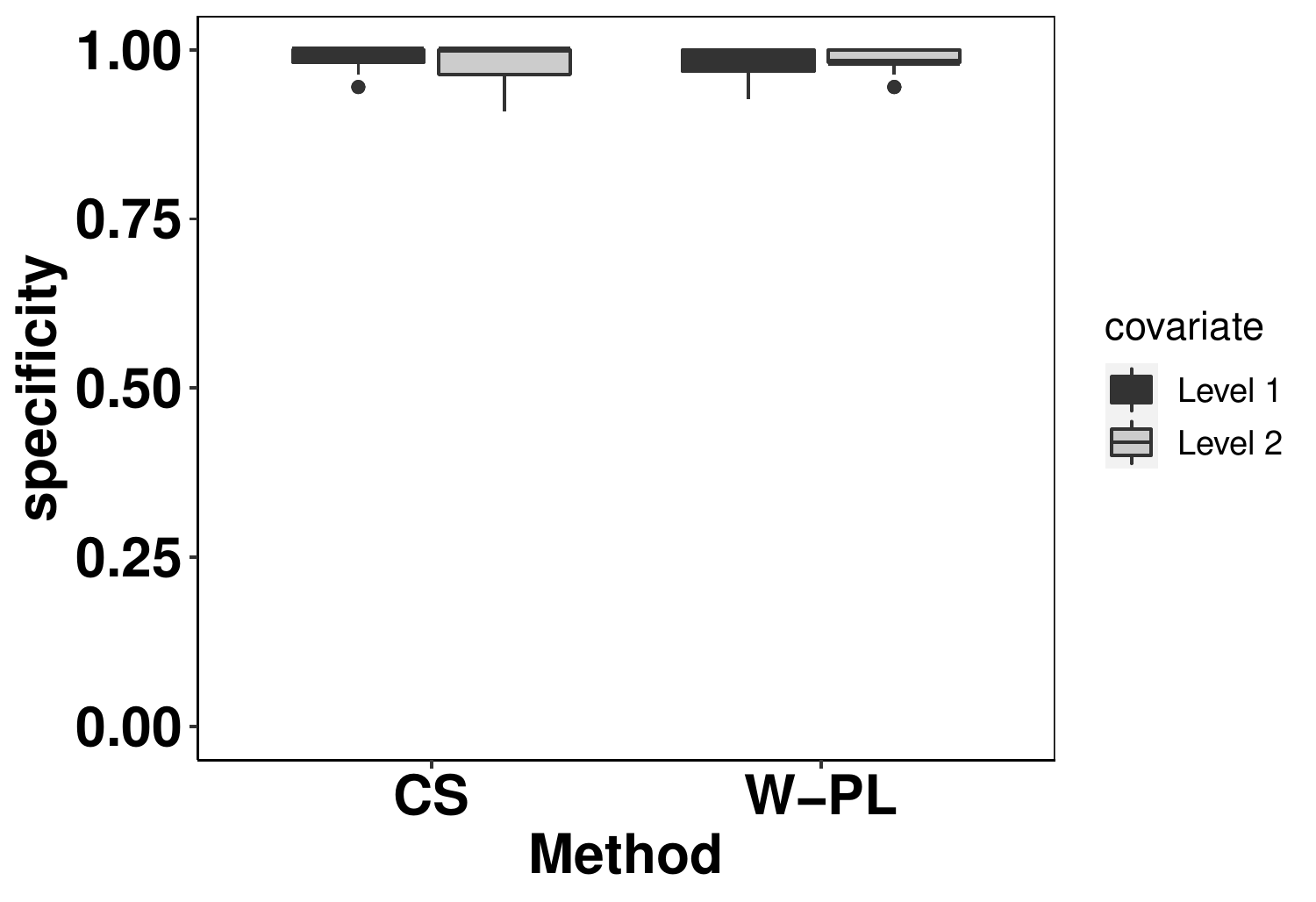}&
			\includegraphics[width=0.3\linewidth]{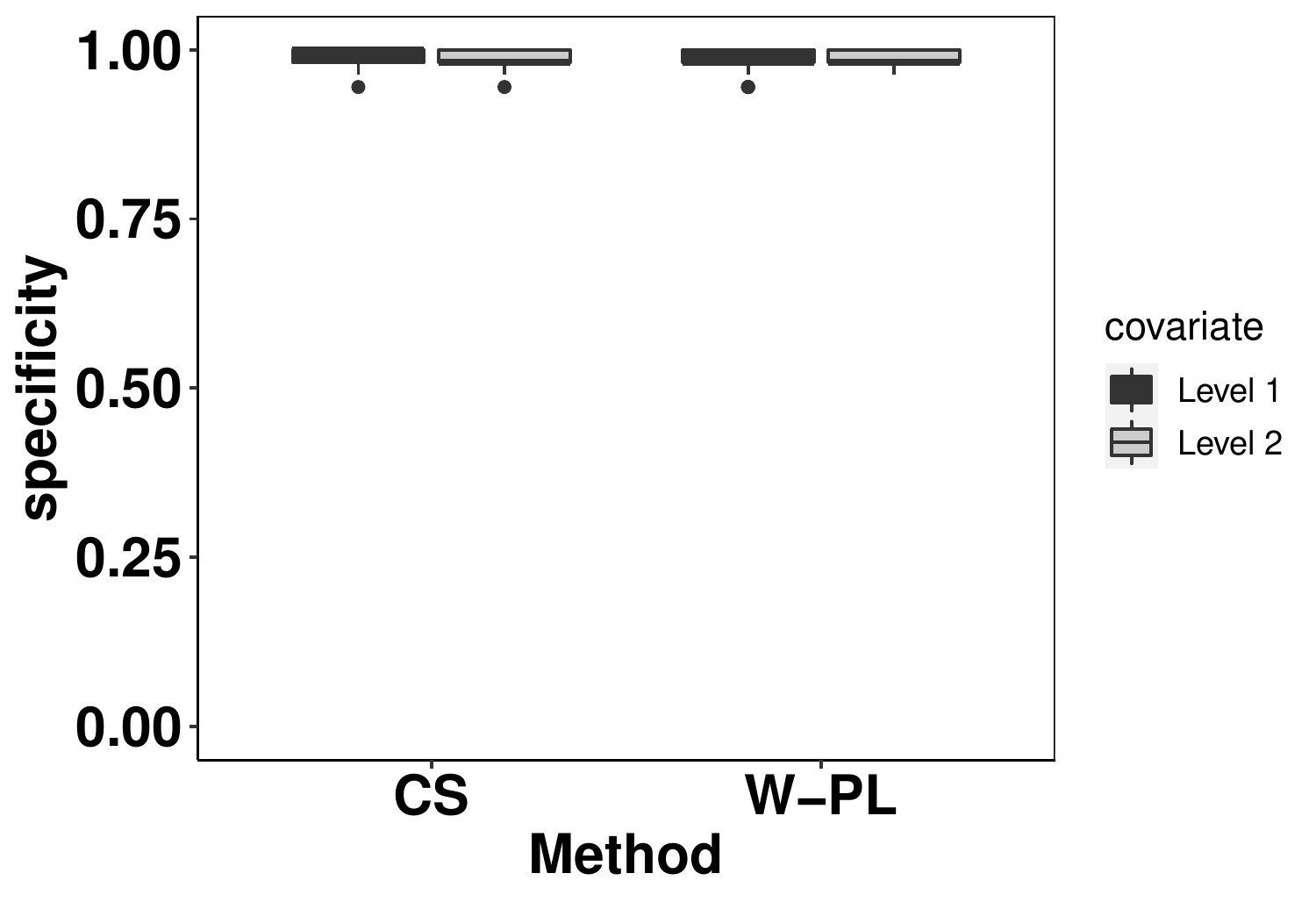} &
			\includegraphics[width=0.3\linewidth]{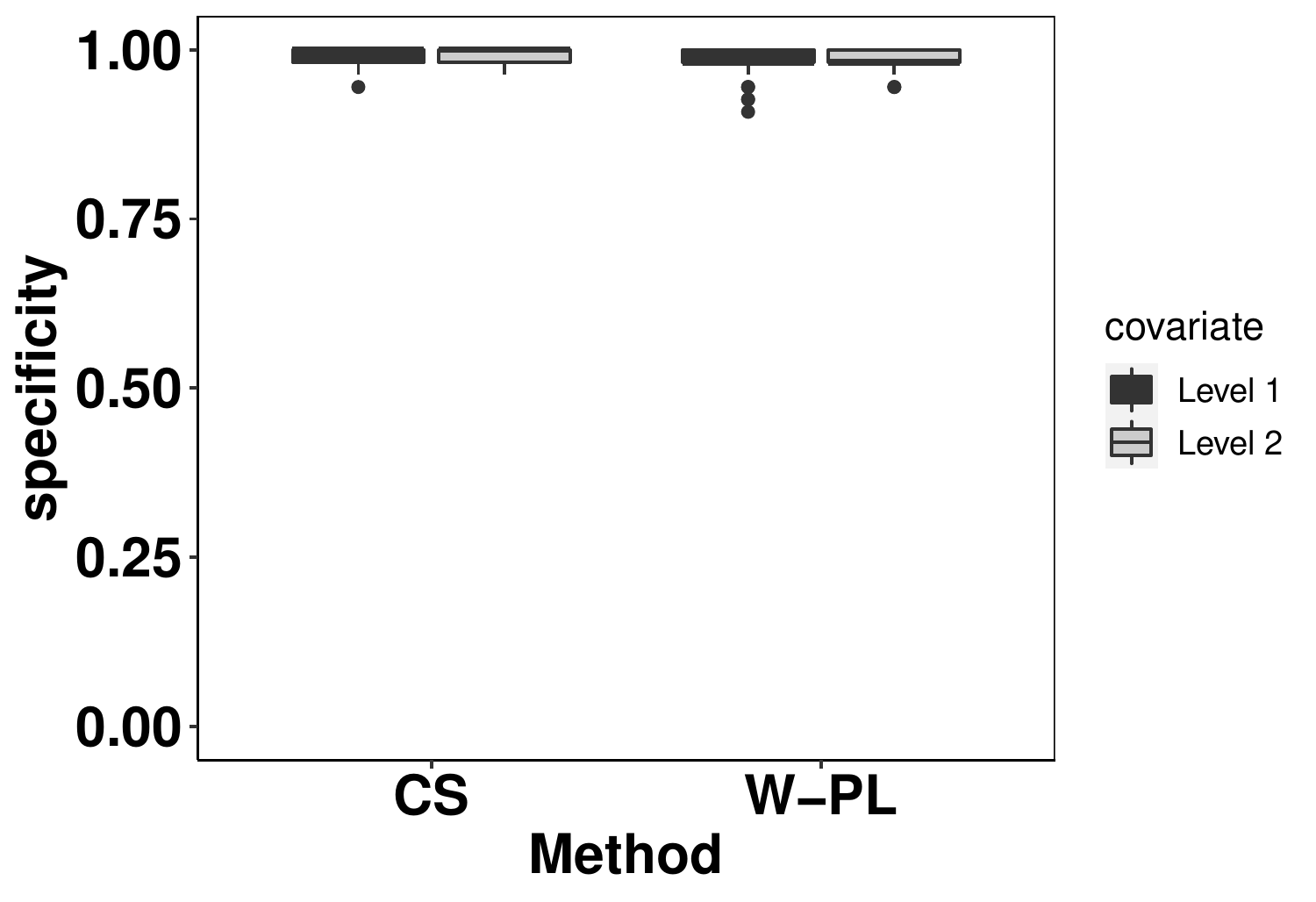}\\
			\includegraphics[width=0.3\linewidth]{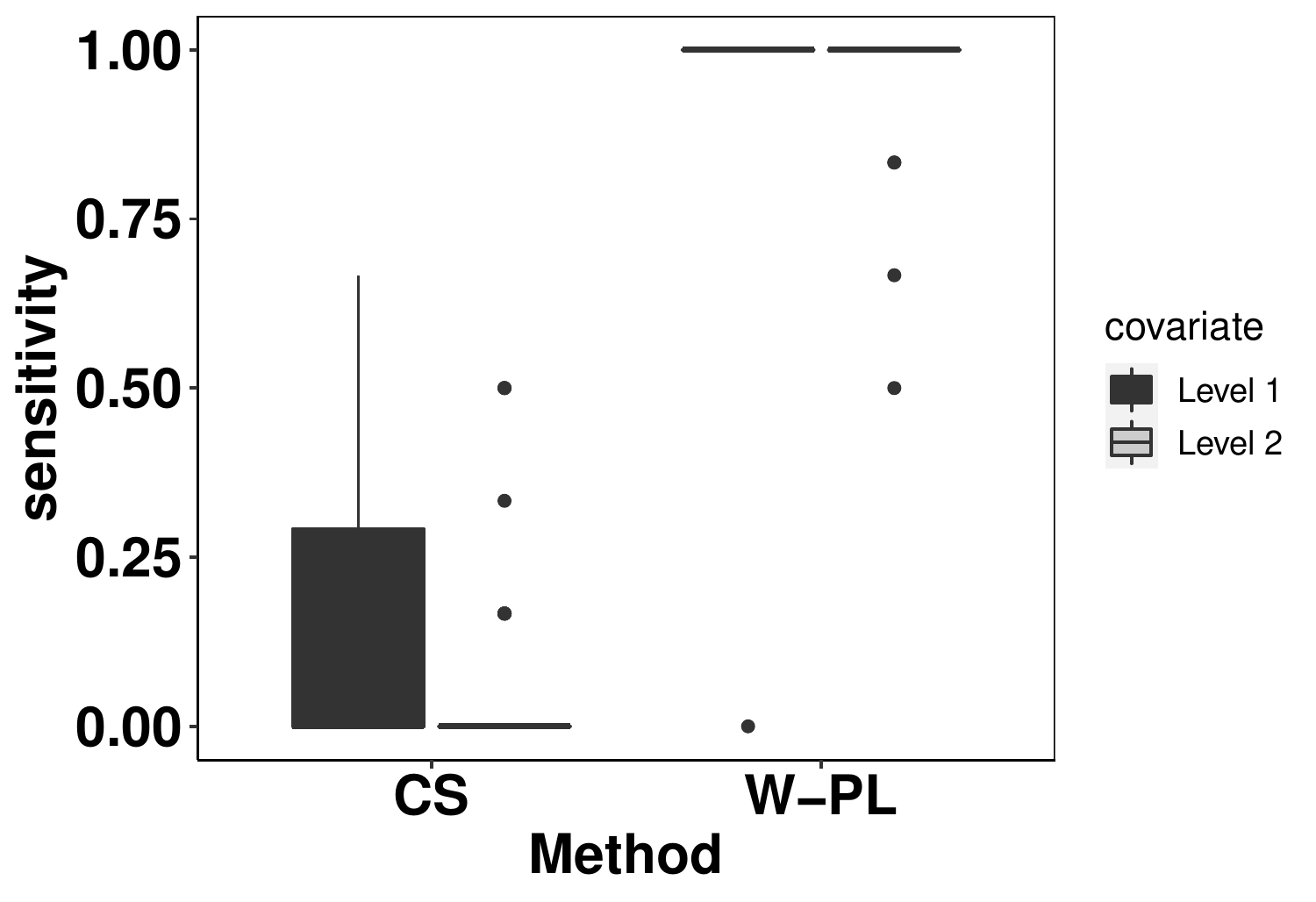}&
			\includegraphics[width=0.3\linewidth]{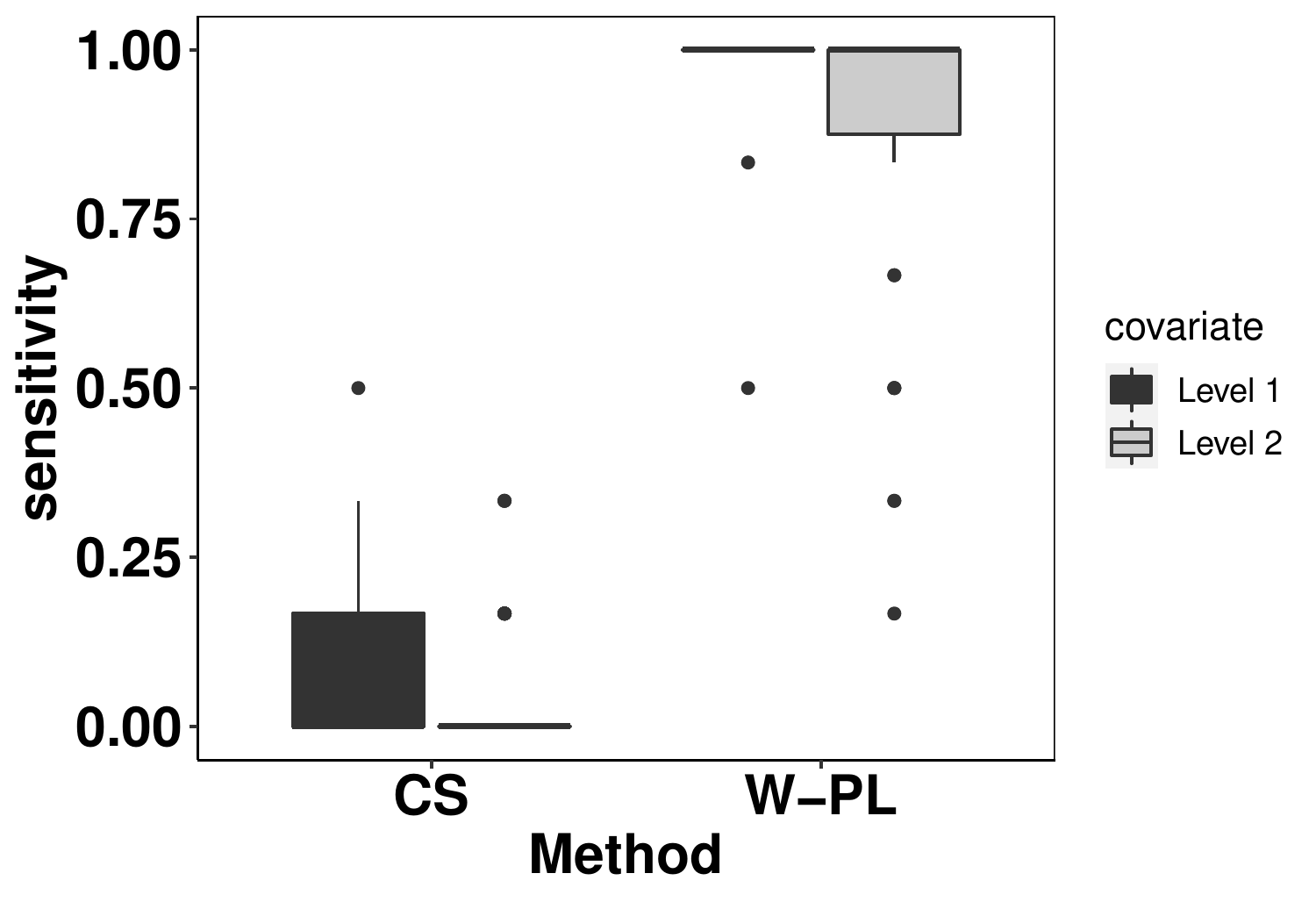}&
			\includegraphics[width=0.3\linewidth]{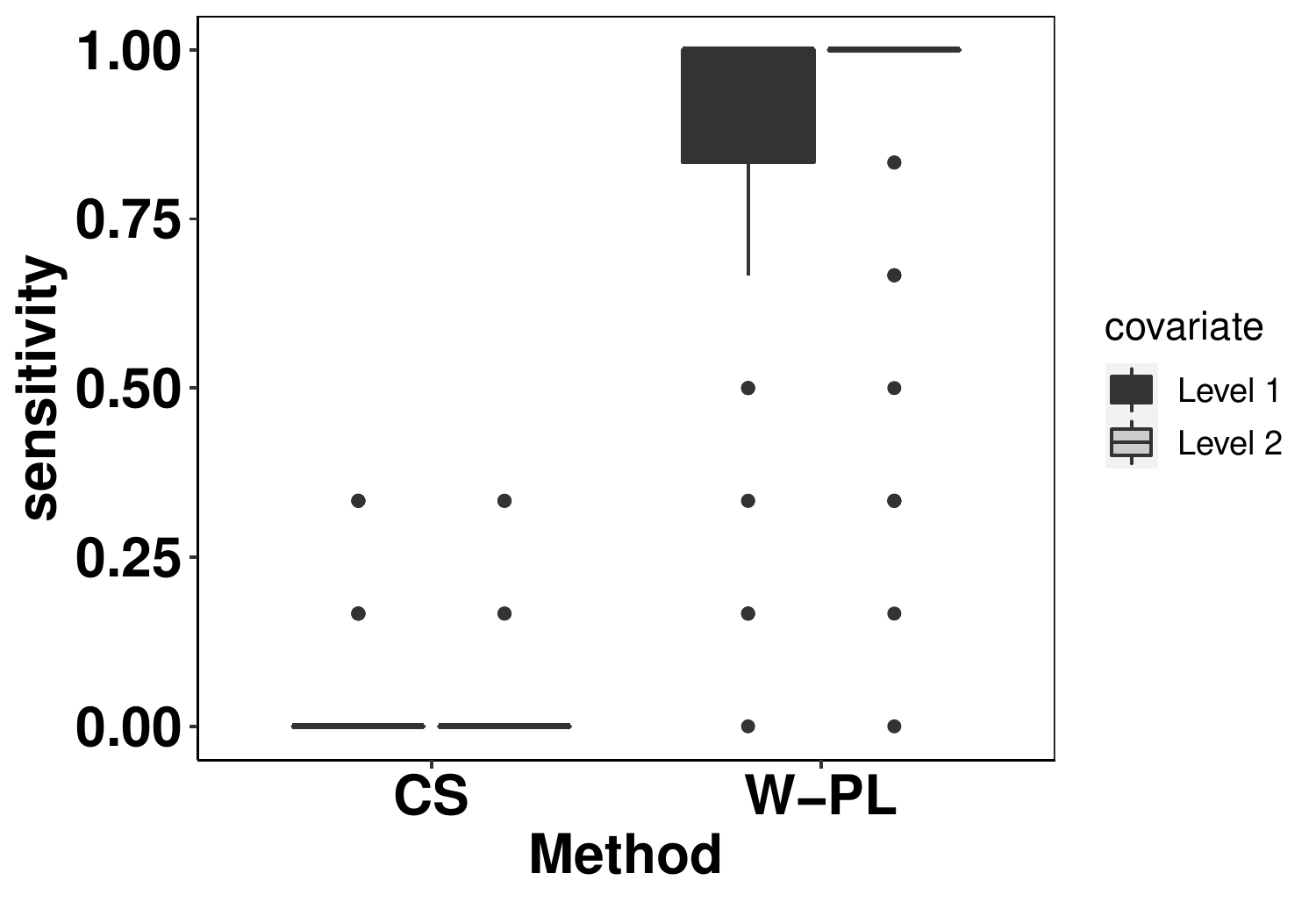}\\
		\end{tabular}
		\caption{\it Specificity and sensitivity comparisons between CS VS W-PL for data with $2\%$ (left), $5\%$ (middle) and $10\%$ (right) contamination. }
		\label{fig1111}
	\end{center}
\end{figure}

\begin{figure}[htbp]
	\begin{center}
		\begin{tabular}{ccc}
			\includegraphics[width=0.3\linewidth]{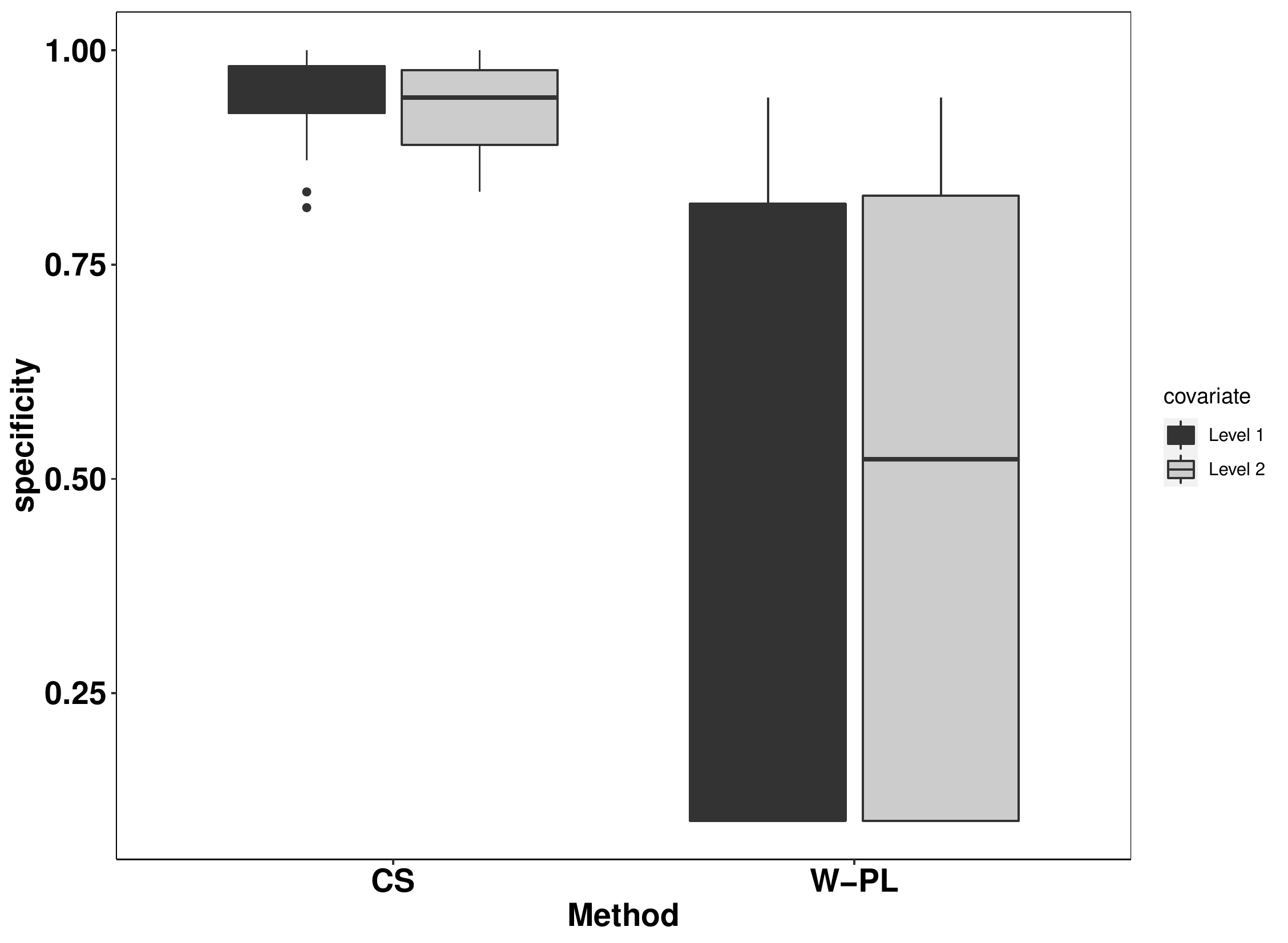} &
			\includegraphics[width=0.3\linewidth]{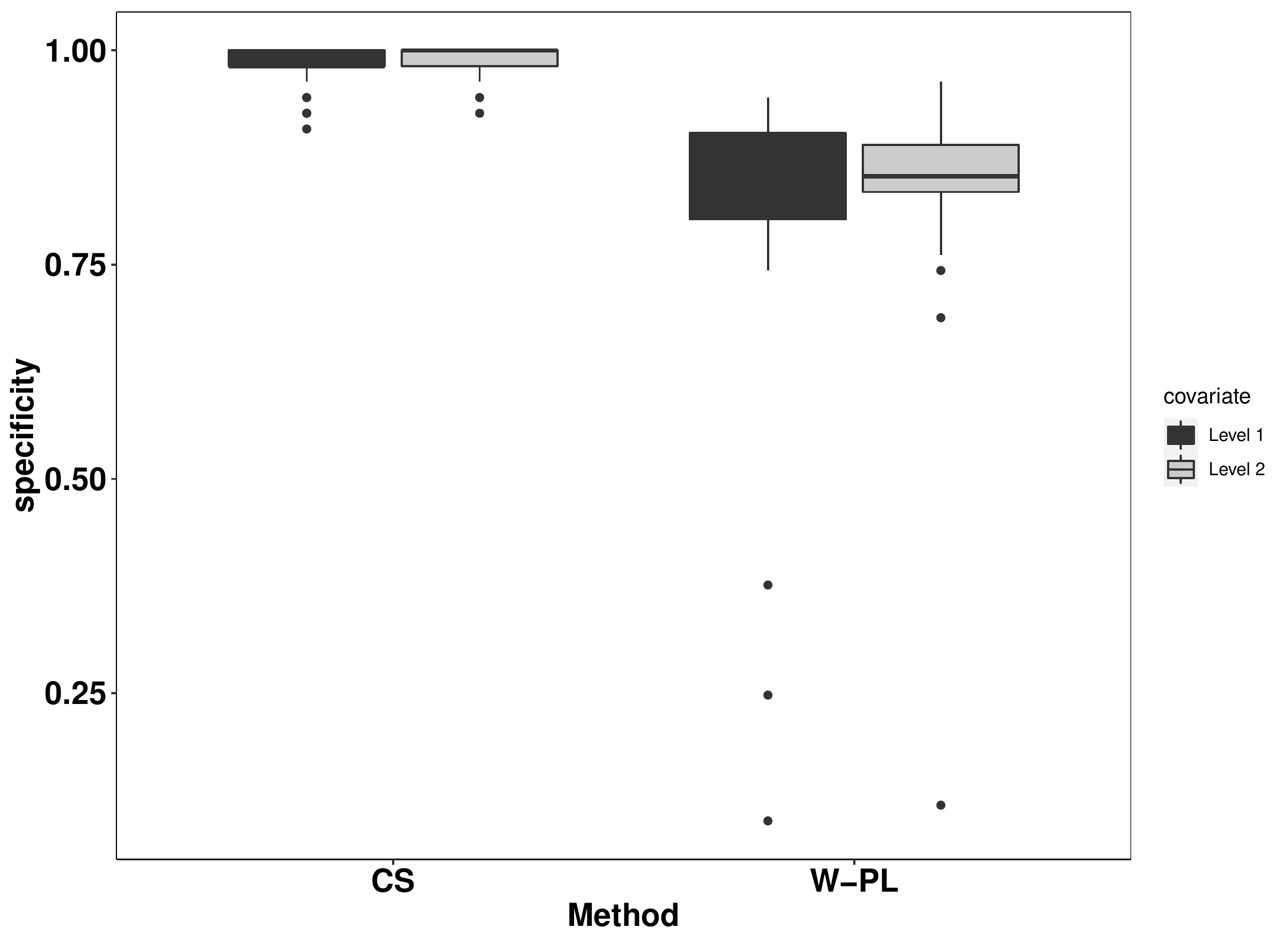}&
			\includegraphics[width=0.3\linewidth]{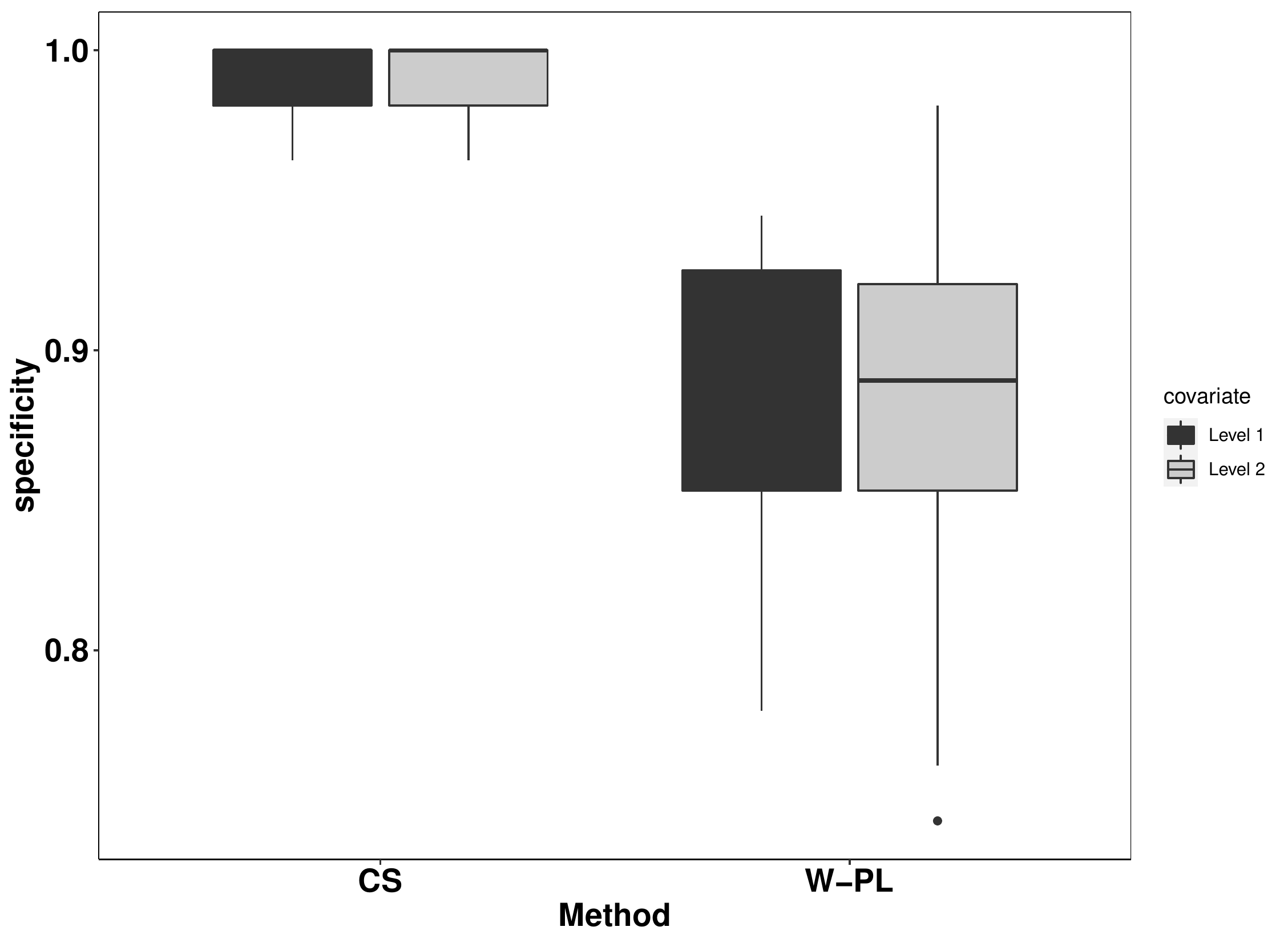}\\
			\includegraphics[width=0.3\linewidth]{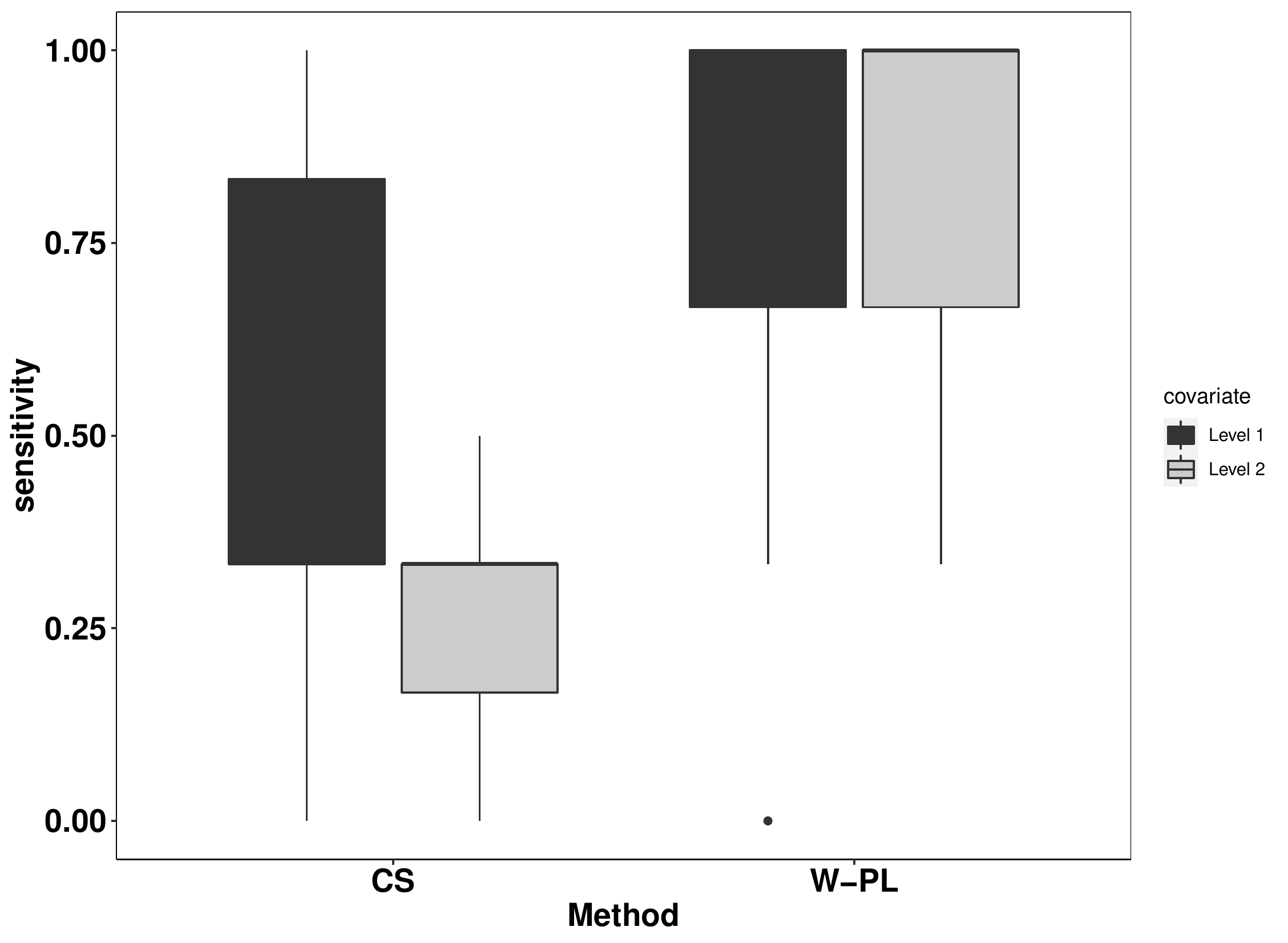} &
			\includegraphics[width=0.3\linewidth]{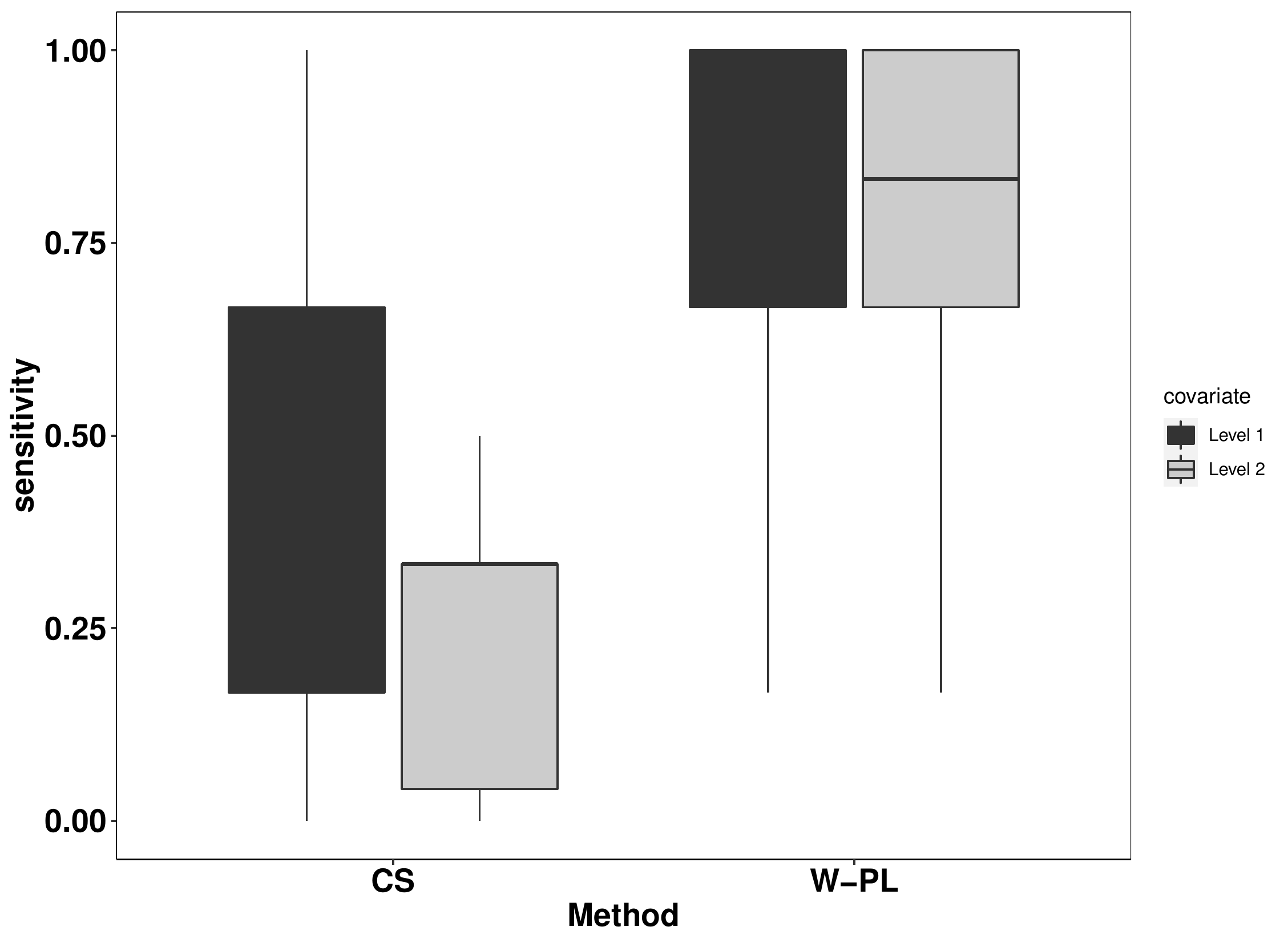}&
			\includegraphics[width=0.3\linewidth]{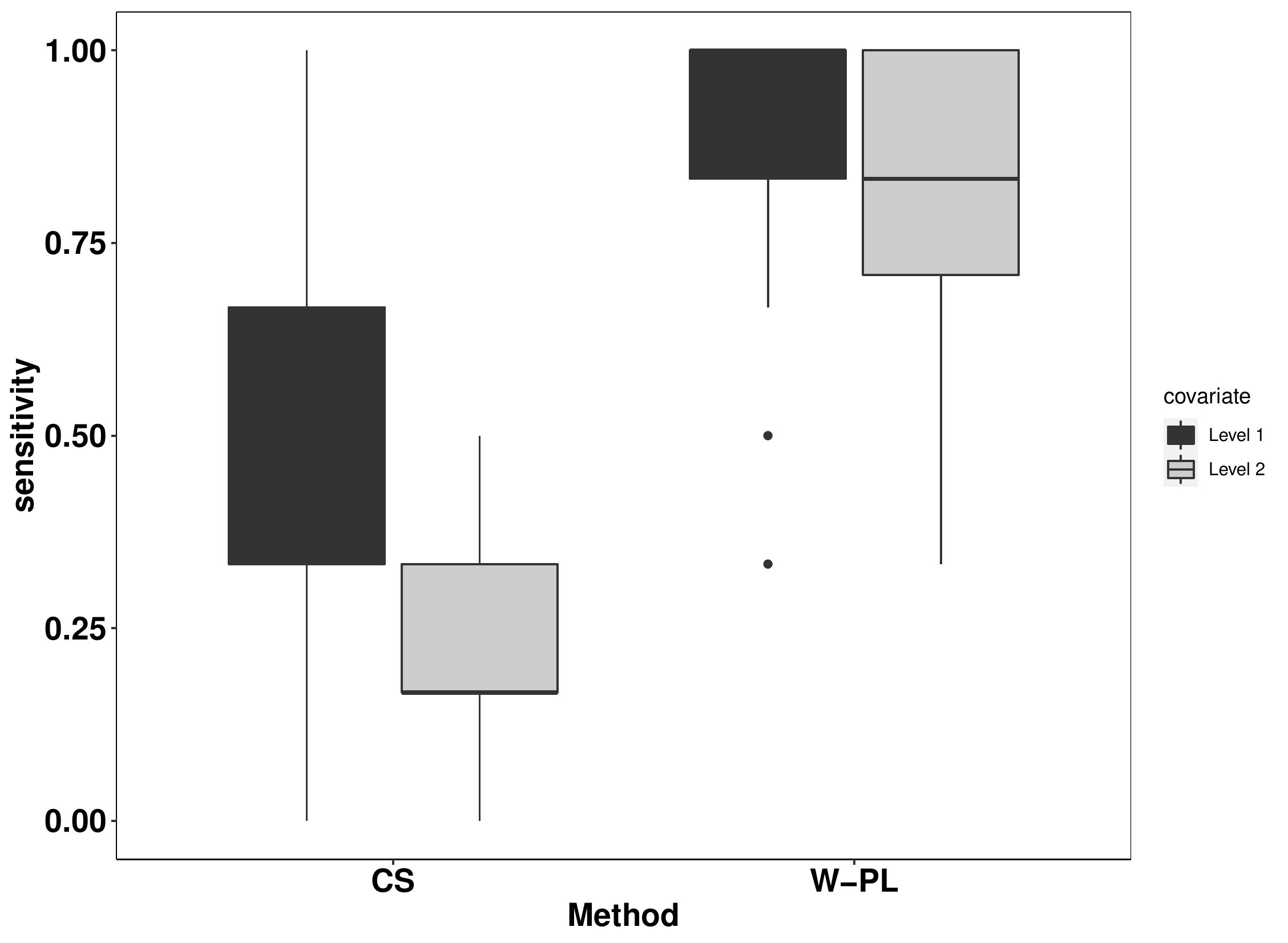}\\
		\end{tabular}
		\caption{\it Specificity and sensitivity comparisons between CS VS W-PL for t distributed data with df 3 (left), 6(middle) and 12(right), with $n=50$ in both groups and $p=11$. }
		\label{fig3612}
	\end{center}
\end{figure}

\subsection{Comparison to \cite{qiu2016joint}}\label{app:Qiu_comp}

Here, we provide a brief comparison of W-PL to the method of \cite{qiu2016joint}. Although their method is applicable in the continuous covariate setting, their work focuses on the case when there are time replicates per subject. It is possible to extend their method to the case where there are not replicates by estimating the subject-level covariance matrices using a kernel-weighted average, however, the implementation provided by the authors does not include this functionality. Thus, we mainly focused on loggle and mgm as competitors for W-PL in our experiments, since the available implementations of these models could directly handle data without replicates, and similar to \cite{qiu2016joint}, both are kernel-based models that frame the precision matrix as varying continuously with time. 

We also note that although there are methods for modeling heterogeneous dependence structures other than those that model the dependence structure as varying with time, such as \cite{ren_gaussian_2022}, we are not aware of any that are directly applicable to a continuous covariate. For example, \cite{ren_gaussian_2022} assumes that the data may be grouped into clusters such that the precision matrix is homogeneous within each of the clusters, and that the means of the clusters are sufficiently separable from one another. On the other hand, W-PL can model the precision matrices as varying continuously and does not place any restrictions on the mean structure of the data.   

For comparison to \cite{qiu2016joint}, we used their simulation setting with 2 (the minimum allowed) time replicates for 100 individuals with 10 variables. We consider estimating the graph for subject 1. For W-PL, we only used the information at the first time point, whereas for \cite{qiu2016joint}, we used the full information on both time points. In this experiment, we revert to the hyperparameter specification scheme from Section \ref{hidim}. We summarize results from this experiment in Figure \ref{fig3624}. Notably, W-PL obtains superior results even when the method of \cite{qiu2016joint} technically uses double the number of observations.

\begin{figure}[htbp]
	\begin{center}
		\begin{tabular}{cc}
			\includegraphics[width=0.43\linewidth]{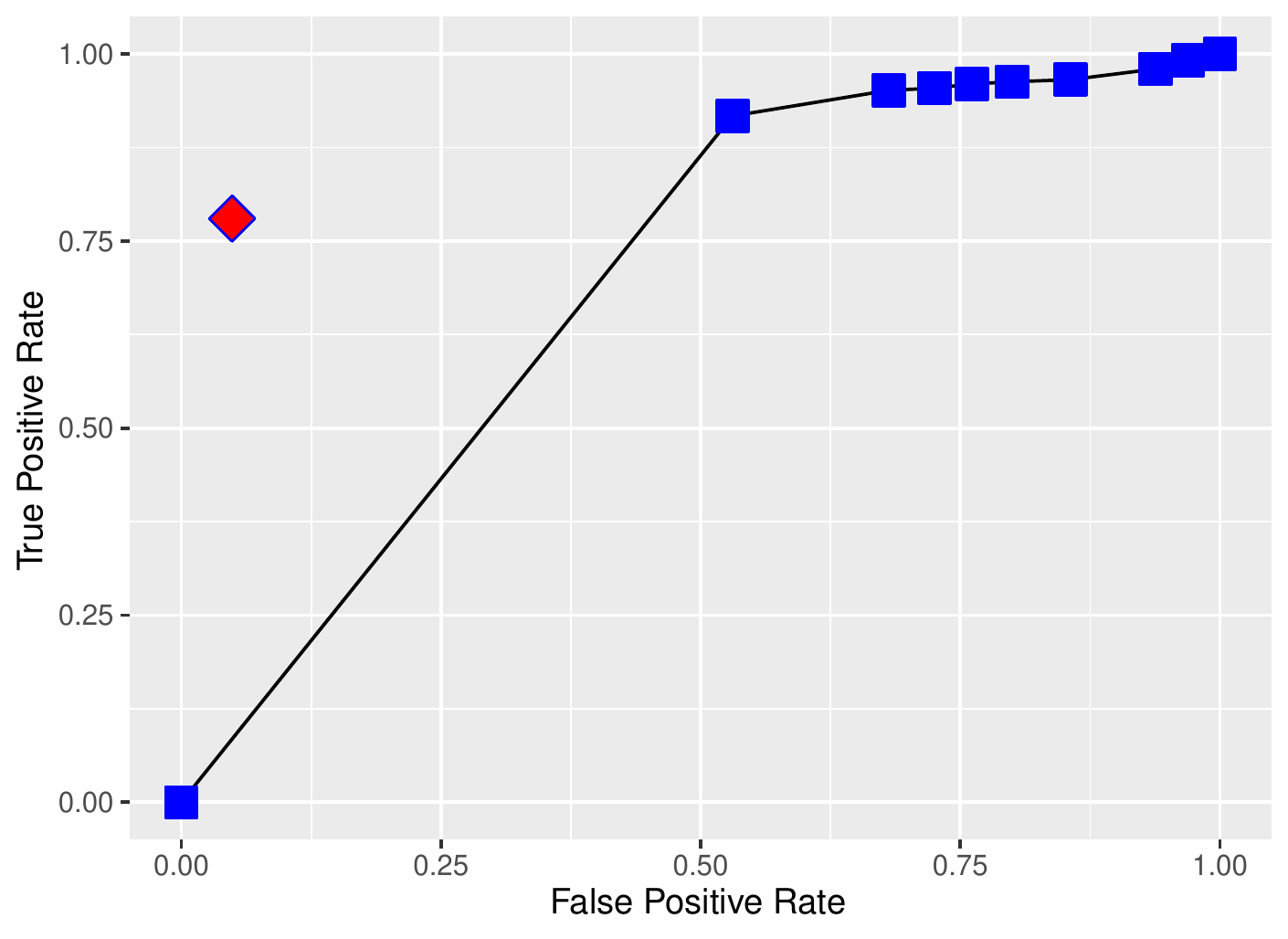} &
			\includegraphics[width=0.43\linewidth]{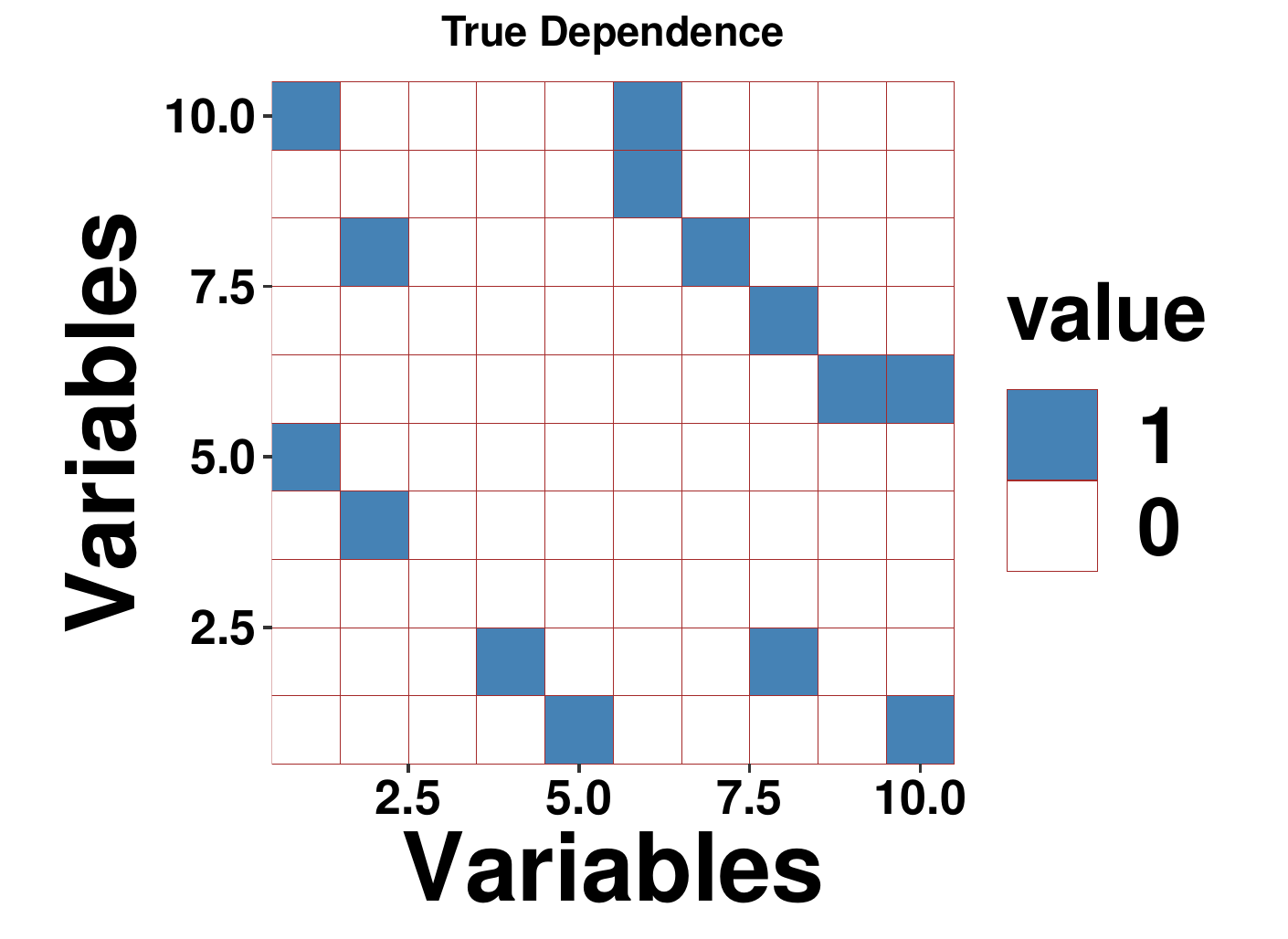}\\
			\includegraphics[width=0.43\linewidth]{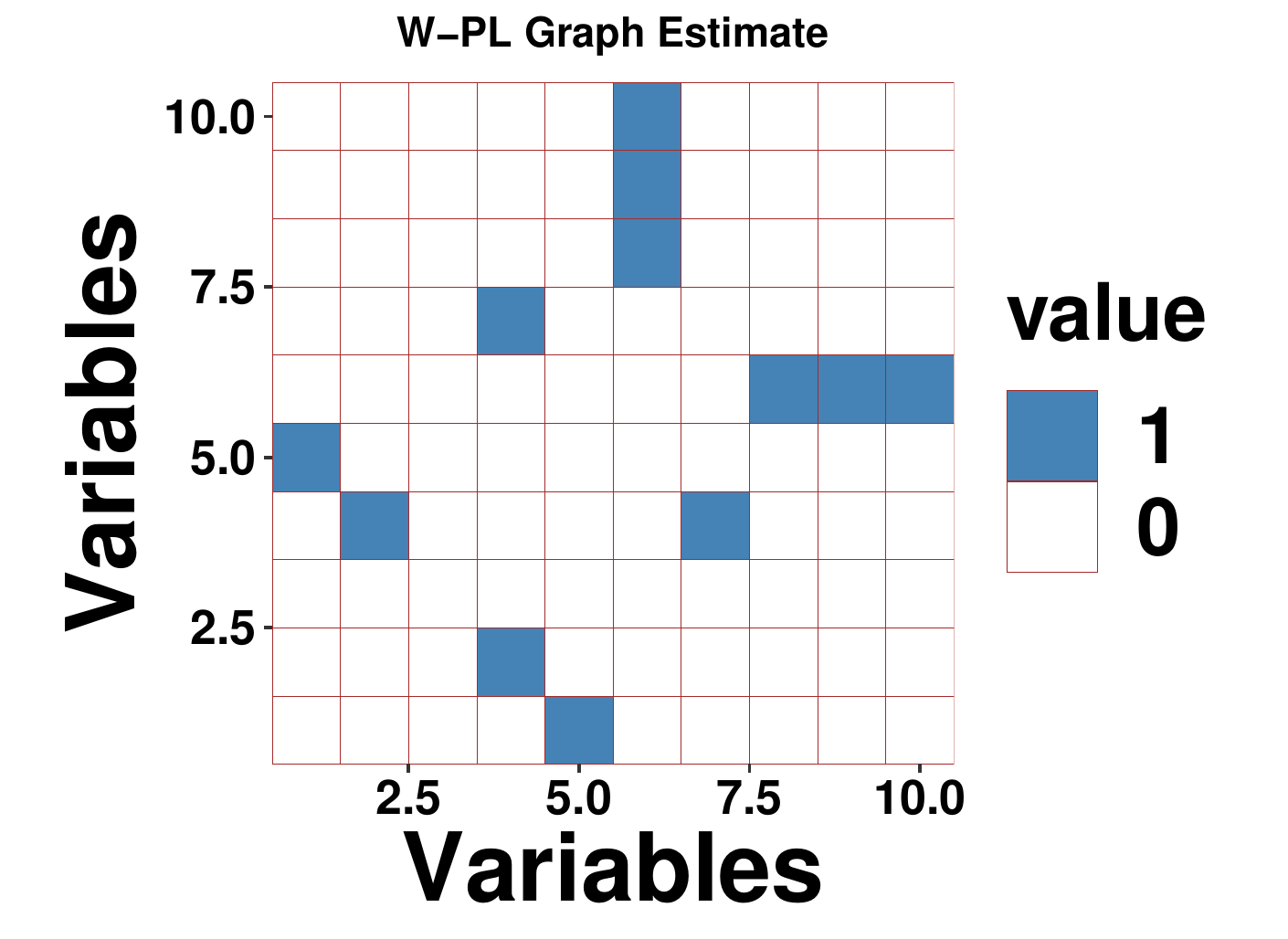}&
			\includegraphics[width=0.43\linewidth]{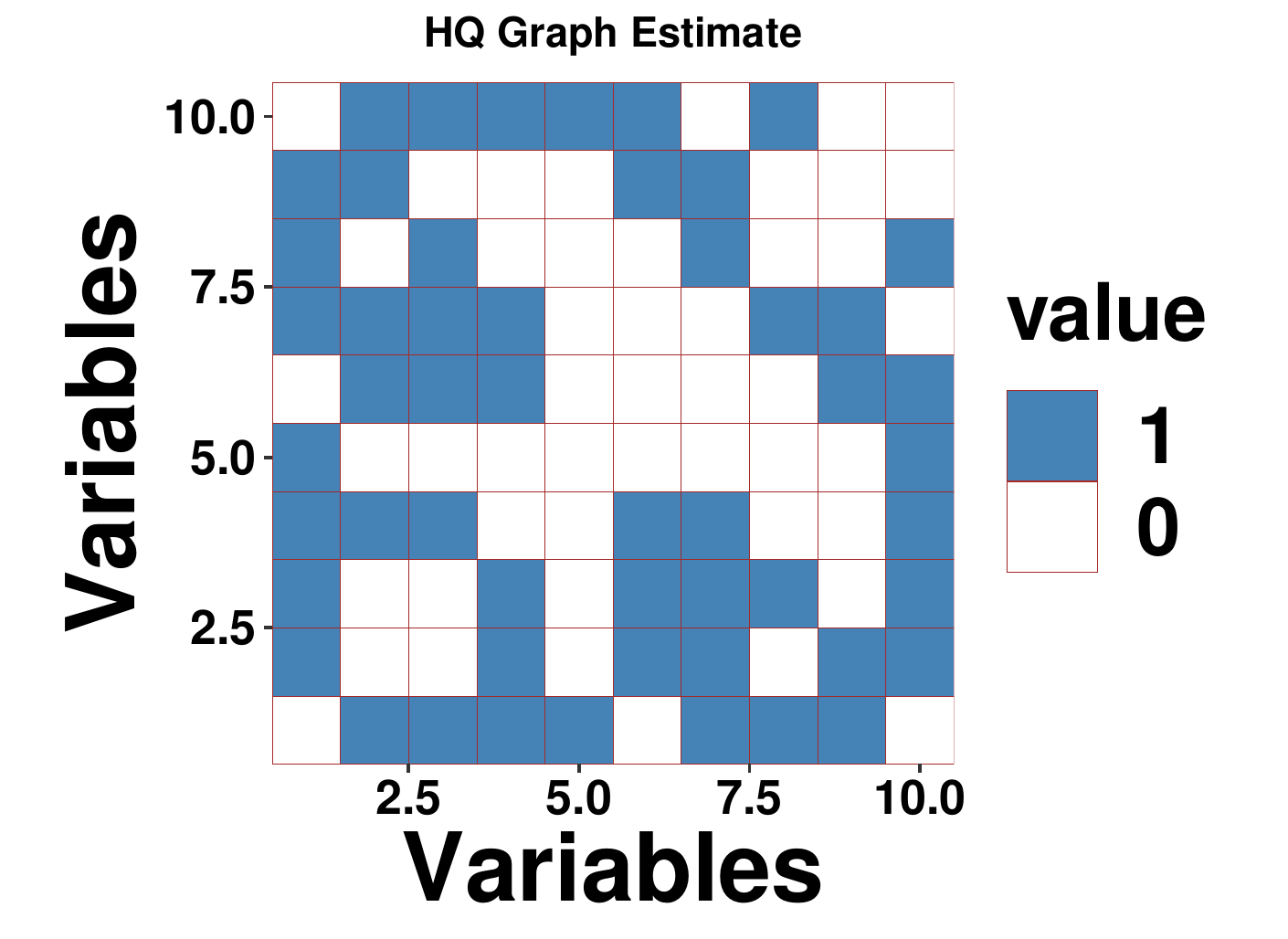} \\
		\end{tabular}
		\caption{\it Qiu's HQ method versus W-PL method. Top left shows the average TPR and FPR for HQ method (squares) versus our method (diamond) across 50 iterations. Top right shows the true dependence structure of the individual for a particular iteration, and the bottom left shows our estimate. Bottom right shows HQ estimate. }
		\label{fig3624}
	\end{center}
\end{figure}

\subsection{Ground-Truth Dependence Structures}

\subsubsection{Unidimensional Covariate}\label{uni_gr}

In the unidimensional covariate setting, we can split the individuals in three clusters based on their covariates: $\mathcal{C}_1 =\{i:-3 < \bz_i <-1\}, \mathcal{C}_2 =\{i:-1<\bz_i<1\}$ and $\mathcal{C}_3 =\{i: 1<\bz<3\}$. Then, the ground truth precision structures for each of the clusters is given by:

\begin{align*}    
\mathrm{G}^*(\mathcal{C}_1)=  \begin{bmatrix} 
0 & 1 & 0 &  \\
1 & 0 & 1 & {\bf 0}_{3,p-2} \\
0 & 1 & 0 & \\
& {\bf 0}_{p-2,3} &  & {\bf 0}_{p-2,p-2}
\end{bmatrix}
&&
\mathrm{G}^*(\mathcal{C}_2)=  \begin{bmatrix} 
0 & 1 & 1 &  \\
1 & 0 & 1 & {\bf 0}_{3,p-2} \\
1 & 1 & 0 & \\
& {\bf 0}_{p-2,3} &  & {\bf 0}_{p-2,p-2}
\end{bmatrix}
\\
\mathrm{G}^*(\mathcal{C}_3)=  \begin{bmatrix} 
0 & 0 & 1 &  \\
0 & 0 & 1 & {\bf 0}_{3,p-2} \\
1 & 1 & 0 & \\
& {\bf 0}_{p-2,3} &  & {\bf 0}_{p-2,p-2}
\end{bmatrix}
\end{align*}

\subsubsection{Multidimensional Covariate}\label{mul_gr}

In the multidimensional covariate setting, we can split the individuals in nine clusters based on their covariates: 

\begin{align*}
&\mathcal{C}_1 =\{i:\bz_i\in(-3, -1)\times(-3, -1)\} 
&&     
\mathcal{C}_2 =\{i:\bz_i\in(-3, -1)\times(-1, 1)\}
\\
&\mathcal{C}_3 =\{i:\bz_i\in(-3, -1)\times(1, 3)\} &&     \mathcal{C}_4 =\{i:\bz_i\in(-1, 1)\times(-3, -1)\}
\\
&\mathcal{C}_5 =\{i:\bz_i\in(-1, 1)\times(-1, 1)\} &&     \mathcal{C}_6 =\{i:\bz_i\in(-1, 1)\times(1, 3)\}
\\
&\mathcal{C}_7 =\{i:\bz_i\in(1, 3)\times(-3, -1)\} &&     \mathcal{C}_8 =\{i:\bz_i\in(1, 3)\times(-1, 1)\}
\\
&\mathcal{C}_9 =\{i:\bz_i\in(1, 3)\times(1, 3)\}
\end{align*}

Then, the ground truth precision structures for each of the clusters is given by:

\begin{align*}
\mathrm{G}^*(\mathcal{C}_1\cup\mathcal{C}_4)=  \begin{bmatrix} 
0 & 1 & 0 &  \\
1 & 0 & 1 & {\bf 0}_{3,p-2} \\
0 & 1 & 0 & \\
& {\bf 0}_{p-2,3} &  & {\bf 0}_{p-2,p-2}
\end{bmatrix}
&&
\mathrm{G}^*(\mathcal{C}_2\cup\mathcal{C}_3\cup \mathcal{C}_5\cup\mathcal{C}_6)=  \begin{bmatrix} 
0 & 1 & 1 &  \\
1 & 0 & 1 & {\bf 0}_{3,p-2} \\
1 & 1 & 0 & \\
& {\bf 0}_{p-2,3} &  & {\bf 0}_{p-2,p-2}
\end{bmatrix}
\\
\mathrm{G}^*(\mathcal{C}_7)=  \begin{bmatrix} 
0 & 0 & 0 &  \\
0 & 0 & 1 & {\bf 0}_{3,p-2} \\
0 & 1 & 0 & \\
& {\bf 0}_{p-2,3} &  & {\bf 0}_{p-2,p-2}
\end{bmatrix}
&&
\mathrm{G}^*(\mathcal{C}_8\cup \mathcal{C}_9)=  \begin{bmatrix} 
0 & 0 & 1 &  \\
0 & 0 & 1 & {\bf 0}_{3,p-2} \\
1 & 1 & 0 & \\
& {\bf 0}_{p-2,3} &  & {\bf 0}_{p-2,p-2}
\end{bmatrix}
\end{align*}

\bibliographystyle{biometrika} 
\bibliography{biblio}
\end{document}